\documentclass[prx,amsmath,amssymb, notitlepage, onecolumn,
nofootinbib,
superscriptaddress,
longbibliography,11pt,tightenlines
]{revtex4-1}

\usepackage{amsmath}
\usepackage{tikz-cd}
\usepackage[intoc]{nomencl}
\makenomenclature

\usetikzlibrary{cd, arrows.meta}
\usepackage{tabularx,graphicx}
\usepackage{epstopdf}
\usepackage{makecell}
\usepackage{graphicx}
\usepackage{latexsym}
\usepackage{amssymb}
\usepackage{amsthm}
\usepackage{color, colortbl}
\usepackage{psfrag}
\usepackage{bbm}
\usepackage{bm}

\usepackage{dsfont}
\usepackage{feynmp}
\usepackage{slashed}
\usepackage[utf8]{inputenc}
\usepackage[symbols,nogroupskip,sort=none]{glossaries-extra}

\usepackage{multirow}
\usepackage{url}
\usepackage{xr}
\usepackage{xcite}
\usepackage{tikz}
\usepackage{microtype}
\usetikzlibrary{shapes.geometric, arrows}
\tikzstyle{arrow} = [thick,->,>=stealth]
\tikzstyle{startstop} = [rectangle, rounded corners, minimum width=3cm, minimum height=1cm,text centered, draw=black, fill=white!30]
\tikzstyle{process} = [rectangle, minimum width=3cm, minimum height=1cm, text centered, draw=black, fill=orange!30]
\tikzstyle{decision} = [diamond, minimum width=3cm, minimum height=1cm, text centered, draw=black, fill=green!30]

\makeatletter
\newcommand*{\addFileDependency}[1]{
  \typeout{(#1)}
  \@addtofilelist{#1}
  \IfFileExists{#1}{}{\typeout{No file #1.}}
}
\makeatother

\newcommand{\beq}{\begin{eqnarray}}
	\newcommand{\eeq}{\end{eqnarray}}

\newcommand{\la}{\langle}
\newcommand{\ra}{\rangle}

\newcommand{\tr}{{\rm tr}}

\newcommand{\bsp}{\begin{aligned}}
	\newcommand{\esp}{\end{aligned}}
\newcommand{\const}{{\rm const}}

\newcommand{\ie}{{i.e., }}
\newcommand{\eg}{{e.g., }}

\newcommand{\rH}{\mathrm{H}}
\newcommand{\R}{\mathbb{R}}
\newcommand{\Aut}{\mathrm{Aut}}

\newcommand{\dist}{\mathrm{dist}}

\newcommand{\cS}{\mathcal{S}}
\newcommand{\cK}{\mathcal{K}}
\newcommand{\im}{\mathrm{im}}

\newcommand{\oA}{\overline{\A}}
\newcommand{\ovl}{\overline}
\newcommand{\cB}{\mathscr{B}}

\newcommand{\rd}{\mathrm{d}}

\newcommand{\Inn}{\mathrm{Inn}}
\newcommand{\Out}{\mathrm{Out}}
\newcommand{\N}{\mathbb{N}}

\newcommand{\bra}[1]{\langle #1|}
\newcommand{\ket}[1]{| #1 \rangle}

\usepackage{xcolor}
\definecolor{darkblue}{rgb}{0.,0.,0.4}
\definecolor{darkred}{rgb}{0.5,0.,0.}
\definecolor{BlueViolet}{RGB}{138,43,226}
\definecolor{SkyBlue}{RGB}{30,144,255}
\definecolor{DarkGreen}{RGB}{0,100,0}

\usepackage[colorlinks=true,linkcolor=darkblue,citecolor=blue,urlcolor=darkred]{hyperref}
\usepackage[normalem]{ulem}

\usepackage{mathrsfs}

\usepackage{comment}

\newcommand{\z}{\mathbb{Z}}

\newcommand{\A}{\mathscr{A}}
\newcommand{\B}{\mathcal{B}}
\newcommand{\G}{\mathscr{G}}
\newcommand{\cH}{\mathcal{H}}
\newcommand{\cU}{\mathscr{U}}
\newcommand{\Ad}{\mathrm{Ad}}

\newcommand{\bbC}{\mathbb{C}}

\newcommand{\rHom}{\mathrm{Hom}}
\newcommand{\QCA}{\mathrm{QCA}}

\newcommand{\id}{\mathrm{id}}
\newcommand{\cM}{\mathscr{M}}

\newcommand{\SOT}{\mathrm{SOT}}
\newcommand{\WOT}{\mathrm{WOT}}

\newcommand{\sC}{\mathscr{C}}

\newcommand{\E}{\mathscr{E}}
\newcommand{\odgr}{\mathbin{\widehat\odot}}
\newcommand{\grten}{\mathbin{\widehat\otimes}}
\newcommand{\Min}{\mathrm{gr},\min}
\newcommand{\Max}{\mathrm{gr},\max}

\newcommand{\diam}{\mathrm{diam}}
\def\U{\mathrm{U}(1)}

\newtheorem{corollary}{Corollary}
\newtheorem{theorem}{Theorem}
\newtheorem{lemma}{Lemma}
\newtheorem{definition}{Definition}
\newtheorem{example}{Example}
\newtheorem{proposition}{Proposition}
\newtheorem{remark}{Remark}

\newtheorem{question}{Question}

\newtheorem*{theorem*}{Theorem}

\numberwithin{equation}{section}
\numberwithin{corollary}{section}
\numberwithin{theorem}{section}
\numberwithin{lemma}{section}
\numberwithin{definition}{section}
\numberwithin{example}{section}
\numberwithin{proposition}{section}
\numberwithin{remark}{section}
\numberwithin{conjecture}{section}
\numberwithin{question}{section}
\numberwithin{claim}{section}

\usepackage{pst-all}

\usepackage{leftidx}

\usepackage[off]{auto-pst-pdf} 
\usepackage{booktabs}
\usepackage{yfonts}
\makeatletter

\usepackage[caption=false]{subfig}
\usepackage{enumitem}

\begin{document}
\title{A local description of strong symmetries and strong-to-weak symmetry breaking in quantum many-body systems}

\author{Ruizhi Liu}
\affiliation{Department of Mathematics and Statistics, Dalhousie University, Halifax, Nova Scotia, Canada, B3H 4R2}
\affiliation{Perimeter Institute for Theoretical Physics, Waterloo, Ontario, Canada N2L 2Y5}
\author{Jinmin Yi}
\affiliation{Perimeter Institute for Theoretical Physics, Waterloo, Ontario, Canada N2L 2Y5}
\affiliation{Department of Physics and Astronomy, University of Waterloo, Waterloo, Ontario, Canada N2L 3G1}
\author{Dominic V. Else}
\affiliation{Perimeter Institute for Theoretical Physics, Waterloo, Ontario, Canada N2L 2Y5}
\begin{abstract}

In mixed states of quantum systems, symmetries come in two types: \emph{strong} and \emph{weak}. Furthermore, it has been argued that in quantum \emph{many}-body systems, strong symmetries can be ``spontaneously broken'' down to weak symmetries. An issue is that as previously formulated, such ``strong-to-weak symmetry breaking'' appears to be a fairly non-local effect. In this paper, we show how to understand and diagnose strong symmetries and strong-to-weak symmetry breaking in an explicitly local way. Our main technical tool is a rigorous definition of strong symmetry in the limit of infinite volume, which generalizes the conventional finite-volume definitions, and for which we give several equivalent formulations, including one involving the concept of ``local charge coherence''. Finally, we introduce von Neumann symmetries, which in infinite-volume systems are intermediate between strong and weak symmetries. We derive a Lieb–Schultz–Mattis type anomaly constraint for von Neumann symmetries (and therefore, in particular, strong symmetries) in quantum spin chains.
\end{abstract}
\maketitle
\tableofcontents

\nomenclature{$\Aut$}{Automorphism group}

\nomenclature{$\Out$}{Outer-automorphism group}

\nomenclature{$\Inn$}{Inner-automorphism group}

\nomenclature{$\rH^{k}(G;A)$}{$k$-th group cohomology of $G$ with coefficient $A$, Def.~\ref{def:group_cohomology}.}

\nomenclature{$\otimes$}{The spatial (a.k.a minimal) tensor product between $C^*$-algebras.}

\nomenclature{$\Lambda$}{The infinite lattice $\Lambda\subseteq\R^{d}$}

\nomenclature{$\A$}{The algebra of quasi-local operators Eq.~\eqref{eq:quasi_local_algebra}}

\nomenclature{$\cS(\A)$}{The space of states of $\A$, Def.~\ref{def:states}}

\nomenclature{$\G^{\QCA}$}{The group of quantum cellular automata, Def.~\ref{def:QCA}}

\nomenclature{$\G^{lp}$}{The group of locality preserving automorphisms, Def.~\ref{def:LPA}}

\nomenclature{$\pi_{\psi}$}{The GNS representation of a positive linear functional $\psi$, Def.~\ref{def:GNS_triple}}

\nomenclature{$\cH_{\psi}$}{The GNS Hilbert space of $\psi$, Def.~\ref{def:GNS_triple}}

\nomenclature{$|\psi\ra$}{The cyclic vector of $\pi_{\psi}$ in $\cH_{\psi}$, Def.~\ref{def:GNS_triple}.}

\nomenclature{$\B(\cH)$}{The algebra of bounded operators on $\cH$.}

\nomenclature{$\sC'$}{The commutant of an algebra $\sC$, Lemma \ref{lemma:Schur}.}

\nomenclature{$\cM_{\psi}$}{The von Neumann algebra generated by state $\psi$, Def.~\ref{def:vN_algebra}}

\nomenclature{SOT}{strong operator topology, Def.~\ref{def:SOT_WOT}}

\nomenclature{WOT}{weak operator topology, Def.~\ref{def:SOT_WOT}}

\nomenclature{$\ovl{\sC}^{\SOT},\ovl{\sC}^{\WOT}$}{The closure of $\sC$ with repsect to $\SOT$ and $\WOT$, Def.~\ref{def:SOT_WOT}}

\nomenclature{$\psi_{1}\simeq\psi_{2}$}{A pure state $\psi_{1}$ is unitarily equivalent to $\psi_{2}$, Lemma \ref{lemma:unitary_equivalence}.}

\nomenclature{$\psi_{1}\sim\psi_{2}$}{Quasi-equivalence of states, Def.~\ref{def:quasi-equivalence}.}

\nomenclature{$\rho\leqslant\psi$}{Majorization of linear functionals: $\rho(a^{*}a)\leqslant \psi(a^{*}a),\quad \,\forall a\in\A$}

\nomenclature{$\Gamma^{c}$}{The complement of $\Gamma$ in the lattice $\Lambda$.}

\nomenclature{$\oA$}{The complex conjugate of $\A$.}

\nomenclature{$\tau$}{The tracial (a.k.a) maximally mixed state, Corollary.~\ref{coro:tracial_strong}.}

\nomenclature{$\cU_{\psi},\cU_{\psi}'$}{The group of unitary elements in $\cM_{\psi},\cM_{\psi}'$.}

\nomenclature{$F(\cdot,\cdot)$}{The fidelity between two positive linear functionals, Def.~\ref{def:fidelity}.}

\nomenclature{$I(\Gamma;\Gamma^{c})$}{The mutual information between a region $\Gamma$ and its complement $\Gamma^{c}$}

\nomenclature{$S(\cdot\|\cdot)$}{Araki's relative entropy.}

\nomenclature{iff}{ if and only if}

\section{Introduction}\label{sec:intro}
A central goal of quantum many-body physics is to classify and understand the phases of quantum matter \cite{Hastings2005,Chen2010,Bachmann2012automorphic,Chen2013,Senthil_2015, Gaiotto2017,Xiong_2018,moon2019automorphicequivalencegappedphases,Ogata2019split,Ogata2019TRS,Ogata2019reflection, Kapustin2020invertible,carvalho2024classificationsymmetryprotectedstates,deroeck2025puregappedgroundstates}. Traditionally, phases of matter have been studied in thermal equilibrium, at zero or nonzero temperature. However, it has more recently been appreciated that the concept of phases of matter is also applicable to more general non-equilibrium states of quantum many-body systems \cite{Buča_2012,de_Groot_2022,Lee_2023,xue2024tensornetworkformulationsymmetry,Lessa:2024wcw,sala2024spontaneousstrongsymmetrybreaking,lu2024bilayerconstructionmixedstate,Lee_2025,Ma_2025,Lessa_2025,Xu_2025,sang2025mixedstatephaseslocalreversibility,ogata2025mixedstatetopologicalorder,ogata2026noteinvariantsmixedstatetopological}. These states could occur, for example, in open systems in which the system undergoes some noise process or interaction with its environment. ``Phase transitions'' in this context can correspond to sharp transitions in quantum information-theoretic properties of the quantum many-body system such as decodability.

When classifying phases of matter, symmetries play an essential role. For mixed (i.e.\ non-pure) quantum states (which non-equilibrium states will generally be), an important distinction has emerged between so-called \emph{weak} and \emph{strong} symmetries. Strong symmetries are generally the appropriate concept when a system interacts with a bath but does not exchange charge with it. Moreover, in the presence of strong symmetry, it has been argued that there can be phases of matter characterized by \emph{strong-to-weak spontaneous symmetry breaking (SW-SSB)} \cite{Lee_2023,Chen:2024ofr,Ma_2025, Lessa_2025,ziereis2025strongtoweaksymmetrybreakingphases, Gu:2024wgc, Feng_2025}.

Nevertheless, compared to conventional spontaneous symmetry-breaking, in some ways SW-SSB remains a mysterious and elusive concept. The diagnostics of SW-SSB that have been proposed are not easy to interpret physically, and moreover are such that implementing them experimentally could in general could require a number of measurements that scales in a bad way (e.g.\ exponentially) with the system size. This had led to arguments that SW-SSB might in fact be \emph{impossible} to diagnose in an efficient way as the system size is increased \cite{Feng_2025}.

In this paper, we show that these issues can be overcome by a ``back-to-basics'' approach to quantum many-body phases of matter. Our central contention is that phases of matter should ideally be defined in terms of the expectation values of \emph{local} operators in the thermodynamic limit. We will show that this leads to a more robust and mathematically rigorous way to distinguish and diagnose phases of matter that have previously been described in terms of ``SW-SSB'', while also providing clearer physical interpretations.

Our key technical tool will be the formulation of several equivalent definitions that generalize the concept of ``strong symmetry'' to infinite systems. Our work is based on the framework of local quantum physics\footnote{More precisely, we work with its lattice version.} \cite{haag2012local, bratteli2013operator1, bratteli2013operator2, Naaijkens_2017, Landsman:2017hpa}. In this approach, the algebra of (quasi-)local observables is taken as the primary object of study.

We will also show that by invoking our definition of strong symmetries, it is possible to prove generalizations of the Lieb-Schultz-Mattis theorem \cite{Lieb1961,Oshikawa1999, Hastings2003, Cheng2015, Po2017, Jian2017,  Cho2017,Watanabe2018LSM,Metlitski2018,Cheng2018a, Kobayashi2018, Ogata2019LSM,Else2020, Jiang2019,Yao2021twisted, Ogata_2021, Aksoy2021, Ye2021a,Ma2022a, Cheng2022, Tasaki2022topological,Kawabata2023, Aksoy2023, Seifnashri2023, Zhou2023, kapustin2024anomalous,Garre_Rubio2024anomalous, Pace2024,Liu2024LRLSM,liu2025twistedlocalitypreservingautomorphismsanomaly,zou2026liebschultzmattisanomaliesanomalymatching} (in its modern formulation as a constraint from anomalous symmetries -- see also Refs.~\cite{Else_2014,rubio2024classifyingsymmetricsymmetrybrokenspin,kapustin2024anomalous,Kapustin2025higher,kawagoe2025anomalydiagnosissymmetryrestriction, Tu2025}) that are applicable to mixed states with strong symmetries. Finally, we will show that in infinite systems, there is a kind of symmetry which is intermediate between strong and weak symmetries, which we refer to as ``von Neumann symmetry''.

\subsection{Preview: local operators, strong symmetries and ``strong-to-weak symmetry breaking''}

The key 
physical assumption that we want to make is that it should always be possible to distinguish different phases of matter via the expectation values of \emph{local} operators\footnote{By local operator, we mean more precisely that the size of the support of the operator does not need to scale with the system size.}. Indeed, in pure state phases of matter this appears to always be the case. For example, topologically ordered states on an infinite plane (or disk) can be studied in terms of local operators; see, e.g., Ref.~\cite{Naaijkens_2017,Shi_2020,Ogata_2022_tensorcat}. 

Focusing on local operators has a practical advantage in that it might be much easier in experiment (or even numerics) to measure local operators than non-local operators. However, it also has a significant theoretical advantage, as follows. The concept of a ``phase of matter'' is strictly speaking only defined in the thermodynamic limit. In finite systems, phase transitions always become smooth crossovers. If one demands that all systems that we consider are strictly finite, then it is necessary to carefully consider various finite-size effects in order to define phase of matter. Therefore, it is conceptually more elegant (and certainly much easier to make mathematically rigorous) to define phase of matter first in a strictly infinite system, and address finite-size effects afterwards. This is indeed the approach we will take in this paper.

Later in Sec.~\ref{sec:alg_states}, we will review some basic notions in the standard mathematical formalism of infinite quantum many-body systems \cite{bratteli2013operator1,bratteli2013operator2,Naaijkens_2017}. In this formalism, states are generally \emph{defined} through the expectation values they assign to local operators (Def.~\ref{def:states}). Therefore, if we wish to formulate a notion of phase of matter that applies naturally to infinite systems, the definition must refer only to local operators.

As we will see shortly, the existing formulations of strong symmetry and strong-to-weak symmetry breaking (SWSSB)  cannot be expressed purely in terms of local operators and therefore cannot be directly applied to infinite systems.  In this paper we aim to close this gap. Doing so will lead to a clearer conceptual foundation for strong symmetries and strong-to-weak symmetry breaking that does not rely on subtle and hard-to-measure global properties.

One of the lessons that will emerge is that formulating the problem purely in terms of local operators requires a reorganization of how we think about strong symmetries and ``strong-to-weak symmetry breaking'' . To illustrate this point, consider a finite (but large) spin-$1/2$ chain with a $\mathbb{Z}_2$ symmetry generated by the unitary operator
\begin{equation}
    X = \prod_i \sigma_i^x .
\end{equation}
Now consider the following two states\footnote{In this discussion, we suppress the overall normalization of density matrices when it is inessential.
}:
\begin{equation}
\label{eq:example_states}
    \rho_1 = \mathbb{I}, 
    \qquad 
    \rho_2 = \frac{1}{2}(\mathbb{I} + X).
\end{equation}
The first state $\rho_1$ is the maximally mixed state. It has the weak symmetry $X$ (since $X\rho_1 X^\dagger=\rho_1$) but no strong symmetry. The second state $\rho_2$, on the other hand, has strong symmetry $X$ in the usual sense since $\rho_2 X = X \rho_2 = \rho_2$.

From the perspective of local operators, however, $\rho_1$ and $\rho_2$ are completely indistinguishable. Indeed, they produce the same state upon tracing out even a single spin. In this paper we therefore take the viewpoint that after we take the thermodynamic limit, \textit{both} $\rho_1$ and $\rho_2$ should be regarded as \textit{lacking} strong symmetry. Formulating a definition of strong symmetry in the thermodynamic limit in terms of local operators will be one of the main goals of this paper, see Def.~\ref{def:strong_sym_main}.

By contrast, an example of a state that \emph{does} possess strong symmetry in the thermodynamic limit, according to our definition, is the pure product state 
\begin{equation}
    \ket{\Psi} = \ket{+}^{\otimes N},
\end{equation}
where $\ket{\pm }$ is the $\pm 1$ eigenstate of $\sigma^x$. 
For another example that is not pure, we can group the spins into pairs and then form the product state
\begin{equation}
    \rho_3 = \left(p \ket{\mbox{++}} \bra{\mbox{++}} + 
(1-p) \ket{\mbox{-- --}} \bra{\mbox{-- --}}\right)^{\otimes N/2}
\end{equation}

What is the physical distinction between $\rho_1,\rho_2$ on the one hand, and $\ket{\Psi}$ and $\rho_3$ on the other hand? Intuitively, the difference is clear: $\rho_1$ and $\rho_2$ exhibit incoherent local fluctuations of $\mathbb{Z}_2$ charge, whereas $\ket{\Psi}$ and $\rho_3$ do not. The definition of strong symmetry that we develop in this paper will make this intuition precise.

This example also reveals that, from the local perspective, the terminology of SWSSB must be reconsidered. Indeed, the state $\rho_2$ has been presented as the prototypical example of a state exhibiting SWSSB \cite{Lessa_2025}, yet from the local viewpoint it possesses only weak symmetry in the thermodynamic limit. Thus, strictly speaking, it should \textbf{not} be regarded as displaying SWSSB. This is primarily a matter of terminology. Phases of matter that, in the conventional terminology, are distinguished by the presence or absence of SWSSB, will from our perspective instead be distinguished by the presence or absence of strong symmetry in the infinite system limit.

This does not mean that the notion of ``SWSSB'' is irrelevant for us. Rather, in our view SWSSB is more naturally interpreted as a property of \emph{dynamical processes} than of states. For example, ordinary spontaneous symmetry breaking refers to the situation in which a Hamiltonian has a symmetry, while its ground states or thermal equilibrium states do not. By analogy, SWSSB can be understood as the situation where a quantum channel or Lindbladian possesses a strong symmetry, but the steady state exhibits only weak symmetry. This appears to be the notion of SWSSB implicit  in Ref.~\cite{Huang_2407}, where it was shown that in the presence of SWSSB of a $\mathrm{U}(1)$ symmetry, the hydrodynamic mode associated with the $\mathrm{U}(1)$ charge can be interpreted as the ``Goldstone mode'' of the spontaneous symmetry breaking.

Finally, let us highlight one key feature of our definition of strong symmetry. For infinite-size systems, we will show in Sec.~\ref{sec:bath_evo} that a strongly symmetric channel can convert a state with strong symmetry into one in which the symmetry becomes weak, but the reverse process can never occur (Prop.~\ref{prop:weak_to_weak}). In other words, once local charge coherence is destroyed, it cannot be recovered. This reflects a form of macroscopic irreversibility—an emergent ``arrow of time'' that appears only in the thermodynamic limit.

\section{Basic setup and summary of main results}\label{sec:setup}

\subsection{Algebra and states}\label{sec:alg_states}
In this section, we introduce some basic notions in local quantum physics, which are necessary to formulate our main results. Details can be found in Ref.~\cite{Blackadar:2006OA,bratteli2013operator1,bratteli2013operator2,Naaijkens_2017,Liu2024LRLSM}. For convenience, we summarize our notations in a list placed after the references. Readers familiar with operator algebras may safely skip this section and consult it as needed. For simplicity, in the main text we will only consider ``bosonic'' spin systems, as opposed to fermionic lattice systems. However,  our results can be easily extended to fermonic systems once the appropriate formalism has been introduced, see Appendix \ref{sec:fermionic_strong_sym}.

In the present work, we work on an infinite lattice $\Lambda \subseteq \mathbb{R}^d$, where on each lattice site $i\in\Lambda$ is assigned with a Hilbert space $\cH_{i}$ with $\dim_{\bbC}(\cH_{i})={n_i}$ for some $2\leqslant n_i\in\mathbb{N}$. We will work with general spatial dimension $d\geqslant 1$ except when otherwise stated.

It is tempting to declare that this system is described by a Hilbert space $\cH\stackrel{?}{=}\bigotimes_{i\in\Lambda}\cH_{i}$. However, as noted in \cite{vonNeumann1939}, such infinite tensor product of Hilbert spaces is ill-defined. Concretely, the inner product between $|\xi\ra:=\bigotimes_{i\in\Lambda}|\xi_{i}\ra$ and $|\eta\ra:=\bigotimes_{i\in\Lambda}|\eta_{i}\ra$,
\begin{equation}
    \la\xi|\eta\ra\stackrel{?}{=}\prod_{i\in\Lambda}\la\xi_{i}|\eta_{i}\ra
\end{equation}
typically diverges and there is no obvious way to regularize it. Thus, the Hilbert-space formalism is not an appropriate mathematical description of infinite systems. Nevertheless, the Heisenberg picture of quantum mechanics, which focuses on operators rather than states, is still applicable even in the presence of infinitely many degrees of freedom, such as spin systems on infinite lattices.

Specifically, for each site $i\in\Lambda$, we associate it with an algebra $\A_{i}:=\mathrm{M}_{n_i}(\bbC)$ (the $n_{i}\times n_{i}$-matrix algebra over complex numbers). For each finite subset $\Gamma\subseteq\Lambda$, local operators supported on $\Gamma$ are defined as
\beq
\A_{\Gamma}:=\bigotimes_{i\in\Gamma}\A_{i}
\eeq
We then define the algebra of local operators $\A^\ell$ to contain all operators supported on any finite subset $\Gamma$. Formally, we can write
\beq
\A^{\ell}:=(\bigcup_{\Gamma\subseteq\Lambda}\A_{\Gamma})/\cong,
\eeq
where we have identified $a\cong a\otimes 1_{\Gamma'\setminus\Gamma}$ if $\Gamma\subseteq\Gamma'$.

The matrix adjoint map $a \mapsto a^\dagger$ on $\A_\Gamma$ defines an involutive anti-linear anti-automorphism on $\A^\ell$, i.e. an anti-linear map $* : \A^\ell \to \A^\ell, a \mapsto a^*$, such that
\beq\label{eq:*-op}
\begin{split}
    (ab)^*&= b^{*}a^{*},\\
    (a^*)^*&=a.
\end{split}
\eeq
Any algebra with $*$-operation satisfying Eq.~\eqref{eq:*-op} is called a $*$-algebra.

A \emph{state} $\psi$ associated to the $*$-algebra $\A^{\ell}$ is defined as follows:
\begin{definition}\label{def:states}
    A linear functional $\psi:\A^{\ell}\to \bbC$ is called a state if it satisfies:
    \begin{enumerate}
        \item Normalization: $\psi(1)=1$.
        \item Positivity: $\psi(a^{*}a)\geqslant0$ for any $a\in\A^{\ell}$.
    \end{enumerate}
    The space of all states on $\A^{\ell}$ is denoted by $\cS(\A^{\ell})$.
\end{definition}

It is readily checked that $\cS(\A^{\ell})$ is a convex set, \ie a convex combination of states is again a state. This leads to the purity of states
\begin{definition}\label{def:pure_states}
    The extremal points of $\cS(\A^{\ell})$, denoted by $\partial\cS(\A^{\ell})$, are called pure states. A state is mixed if it is not pure.
\end{definition}

We will be interested in purification, or more generally, the extension of a mixed state $\psi$ on $\A^{\ell}$. 
\begin{definition}\label{def:extension}
    Let $\cB$ be the algebra of local operators for another spin system. Then, a state $\psi'$ on $\cB\otimes \A^{\ell}$ is an extension of $\psi$ if $\psi'|_{1_{\cB}\otimes\A^{\ell}}=\psi$. In addition, if $\psi'$ is a pure state on $\cB\otimes \A^{\ell}$, it is called a purification of $\psi$.
\end{definition}
We note that given a state $\psi$, its purification may \textbf{not} exist at all once the locality is imposed. Indeed, Lemma~\ref{lemma:clustering} shows that any pure state in the sense of Def.~\ref{def:pure_states} must satisfy the clustering property (i.e., decay of connected correlation functions with the spatial separation; see Lemma~\ref{lemma:clustering} for the precise definition). If a state fails to be clustering, there is no way to purify it while preserving locality.

\subsection{The symmetry actions}\label{sec:sym_action}
Let us now discuss how to formulate symmetry actions in the local language. Specifically, a physically relevant symmetry should map local operators to local operators. Thus, a symmetry operation will correspond to an \emph{automorphism} $\alpha \in \mathrm{Aut}(\A^\ell)$. This means that $\alpha$ is a linear map $\alpha : \A^\ell \to \A^\ell$ such that $\alpha(ab) = \alpha(a) \alpha(b)$ and $\alpha(a^*) = \alpha(a)^*$, and $\alpha$ has an inverse $\alpha^{-1}$ satisfying the same conditions. For example, one could consider an ``on-site'' symmetry where one chooses a unitary action $u_i$ on each site and then conjugates operators supported on $\Gamma \subseteq \Lambda$ by $\bigotimes_{i \in \Gamma} u_i$. Non-on site symmetries such as translation symmetry also correspond to automorphisms of $\A^\ell$.

For some of our results (Theorem \ref{thm:LSM_strong_main}), it will be helpful to impose a stricter notion of locality 
\cite{Gross_2012, Farrelly_2020}:
\begin{definition}\label{def:QCA}
    Let $\alpha\in\Aut(\A^{\ell})$, one says $\alpha$ is a QCA if there exists $r_{\alpha}>0$  such that for any local operator $x\in \A_{X}$, we have $\alpha(x)\in\A_{B(X,r_{\alpha})}$, where $B(X,r_{\alpha}):=\{p\in \Lambda:\dist(p,X)\leqslant r_{\alpha}\}$. The constant $r_{\alpha}$ is called the radius of $\alpha$.
\end{definition}

Given a symmetry group $G$, by slightly abusing the notations{\footnote{Previously the notation $\alpha$ is used to represent an operation acting on operators, but here we use it to represent a map from the symmetry group $G$ to all possible QCA operations $\G^{\QCA}$.}}, the symmetry action can be represented by a group homomorphism $\alpha: G\to \G^{\QCA}$ (or more generally $\alpha : G \to \mathrm{Aut}(\A^\ell)$).
The image of $g\in G$ under $\alpha$ is written as $\alpha_{g}$.
This symmetry may contain internal and/or translation symmetry, and the internal symmetry, which acts as a finite-depth quantum circuits. This general type of symmetry actions covers many physically relevant cases. 

Let us mention that for 1d quantum spin chains, given $\alpha:G\to\G^{\QCA}$, one can define its \textit{anomaly index} $\omega\in\rH^{3}(G;\U)$ \cite{Else_2014,kapustin2024anomalous}, which we review in Appendix \ref{sec:anomaly_index}. In fact, this index can be defined in a similar manner for more general symmetry actions with tails). Anomaly indices in higher dimensions have also been studied in \cite{kapustin2025highersymmetriesanomaliesquantum,Kawagoe2025} but we will not need such generalizations in the present paper.

\subsection{Taking thermodynamic limit}
\label{subsec:taking_thermo}

In this paper we work directly in the infinite-volume setting. One might question the physical relevance of this, since any realistic system will have finite  size. The point is that, in order to study phases of matter, it is at the very least necessary to consider the limiting behavior as the system size approaches infinity. The infinite-system framework can be interpreted simply as a convenient way to package this limit.

Specifically, suppose that for each $m$ we have a finite-size system that occupies a finite subset $\Gamma_m \subseteq \Lambda$ of the infinite-system lattice $\Lambda \subseteq \mathbb{R}^d$. We demand the following properties of the sequence $\{\Gamma_m \}_{m \in \mathbb{N}^*}$:
\begin{enumerate}
    \item it is increasing: $\Gamma_{m}\subseteq\Gamma_{m+1}$, $\forall\,m\in\N^*$.
    \item it is exhausting: $\bigcup_{m=1}^{\infty}\Gamma_{m}=\Lambda$.
\end{enumerate}
Let $\A_{\Gamma_{m}}$ be the algebra of operators on $\Gamma_{m}$, consider and consider a sequence of states $\{\psi_{m}\}_{m\in\N^*}$, where each $\psi_{m}$ is a state of $\A_{\Gamma_{m}}$. We say that this sequence of states converges to a limit state $\psi$ as $m\to\infty$ if \footnote{Mathematically, we work with weak-* topology of states.}
\beq
\lim_{m\to\infty}\psi_{m}(a)=\psi(a),\quad\forall\,a\in\A^{\ell}
\eeq
Therefore, when we talk about a state in the infinite system framework, the reader is always free, if they choose, to interpret it as a limit of finite-size system states.

Similarly, if we have a sequence of automorphisms $\{\alpha^{(m)}\}_{m\in\N^*}$ where each $\alpha^{(m)}\in\Aut(\A_{\Gamma_{m}})$, then we say they converge to a limit automorphism $\alpha\in\Aut(\A^{\ell})$ if
\beq\label{eq:strong_limit}
\lim_{m\to\infty}\alpha^{(m)}(a)=\alpha(a),\quad\forall\,a\in\A^{\ell}.
\eeq
Mathematically, this amounts to saying that we work with the strong topology of automorphisms, and \eqref{eq:strong_limit} describes ``strong convergence'' of the sequence $\{\alpha^{(m)}\}_{m\in\mathbb{N}}$ to $\alpha$. It will be easy to check that
\begin{lemma}[Lemma 5.6 of \cite{Bachmann2012automorphic}]\label{lemma:strong_weak-*}
    Let $\{\alpha^{(m)}\}_{m\in\mathbb{N}}$ be a strongly convergent sequence of automorphisms and $\{\psi_{m}\}_{m\in\mathbb{N}}$ be a sequence of states. Then $\lim_{m\to\infty}\psi_{m}=\psi$ if and only if $\lim_{m\to\infty}\psi_{m}\circ\alpha^{(m)}=\psi\circ\alpha$.
\end{lemma}

\subsection{Summary of main results}\label{sec:sum_main_results}
In this section, we summarize our main results in terms of the framework described above.
Let us begin with defining strong symmetries in the infinite-volume systems.
 Recall that, in finite systems, a strong symmetry of a state described by a density matrix $\rho$ is defined by the condition $U\rho\propto\rho$. In infinite systems, we instead make the following definition:
\begin{definition}\label{def:strong_sym_main}
For example, in a
    Let $\psi\in\cS(\A^{\ell})$ and $\alpha\in\Aut(\A^{\ell})$. We say that $\psi$ is (weakly) symmetric under $\alpha$ if $\psi\circ\alpha=\psi$. In addition, we say that $\psi$ is \textbf{strongly symmetric} if for every spin system $\cB$, and every extension $\psi'$ of $\psi$ on $\cB\otimes\A^{\ell}$, we have $\psi'\circ(1_{\cB}\otimes \alpha)=\psi'$.
\end{definition}
Recall the definition of ``extension'' from Section \ref{sec:alg_states}.
It is straightforward to show that in \emph{finite} systems, this agrees with the usual definition. For example, the canonical purification of a mixed state $\rho$ is symmetric under $\mathbb{I} \otimes U$ if and only if $\rho$ is strongly symmetric under $U$. Moreover, one can show that for \emph{pure} states, strong and weak symmetry are equivalent; specifically, this follows from the fact that for a pure state $\psi$ on $\A$, any extension $\psi'$ on $\cB \otimes \A$ is simply a tensor product between $\cB$ and $\A$. See Appendix \ref{appendix:extension_of_pure_state} for the proof of this fact.

Next, we discuss how this definition is related to notions of ``charge coherence''.
  Let $\alpha \colon G \to \Aut(\A^{\ell})$ be a group homomorphism of a compact group\footnote{This includes finite groups as special cases.} $G$ into automorphisms of the algebra $\A^{\ell}$. Note that we do \textbf{not} have to assume $G$ is represented by QCA's here.

An operator $0\not=O \in \A^{\ell}$ is called a \textbf{charged local operator} if it satisfies
\beq\label{eq:ord_para}
\int_G \alpha_g(O)\, \rd g = 0
\eeq
where $\rd g$ denotes the normalized Haar measure on $G$. An equivalent way to formulate this is that if one decomposes $O$ into a sum over irreps of $G$, there is no weight on the trivial irrep. For notational simplicity, we will henceforth refer to such operator $O$ simply as a \emph{charged operator}.

Let us recall that in quantum mechanics, the fidelity between two density operators $\rho,\sigma$ is defined as
\[
F(\rho,\sigma):=\!\left(\tr\sqrt{\sqrt{\rho}\sigma\sqrt{\rho}}\right)^{2},
\]
which generalizes the (squared) inner product to mixed states. As shown in \cite{Alberti1983}, the notion of fidelity can be extended to positive linear functionals in infinite-size systems (see Def.~\ref{def:fidelity}). One can concretely compute this fidelity between two states $\psi,\psi'$ on an infinite  system according to
\begin{equation}
    F(\psi,\psi') = \lim_{n \to \infty} F(\rho_{A_n}, \rho_{A_n}'),
\end{equation}
where $\rho_{A_n}$ and $\rho_{A_n}'$ are the density matrices describing the reduced state of $\psi$ and $\psi'$ respectively on the finite subsystem $A_n$, and $\{ A_n\}$ is a sequence of subsystems of increasing size that eventually grow to encompass the entire system (i.e. they satisfy the ``increasing'' and ``exhausting'' conditions from Section \ref{subsec:taking_thermo}).
We then propose the following criterion for diagnosing the strong-symmetry breaking based on fidelity.
\begin{definition}\label{def:charge_coherence}
    A state $\psi$ is said to be \textbf{(locally) charge-coherent} if $F(\psi, \psi_{O})=0$ for any charged operator $O$, where $\psi_{O}$ is defined as $\psi_{O}(a):=\psi(O^*\,a\,O)$.
\end{definition}
We note that the parenthetical ``(locally)'' is arguably redundant, since in this paper we only ever consider (quasi-)local operators; therefore, we will omit it in the rest of this paper. However, we included it here in order to emphasize that we are not referring to \emph{global} charge coherence, which is not even well-defined in infinite systems anyway.
In words, the charge-coherence condition is saying that if we act on the state with an operator that locally creates or destroys charge, then the resulting state has \emph{exactly zero} overlap with the original state. We note that this ``fidelity correlator'' was proposed as a diagnostic of SW-SSB in Ref.~\cite{divi2026localstrongtoweakspontaneoussymmetry}.

It turns out that this charge-coherence condition is equivalent to strong symmetries.

\begin{theorem*}[Theorem \ref{thm:strong_sym=cha_coh}]\label{thm:strong_sym=charge_coherence_main}
    A state $\psi$ is strongly symmetric under a compact symmetry group $G$ if and only if it is charge-coherent.
\end{theorem*}

It is useful to know whether the property of having a strong symmetry is invariant under dynamical processes. To this end,
we introduce a class of unital completely positive (UCP) maps, called bath evolutions, which describe general dynamical processes in quantum many-body systems and generalize finite-depth quantum channels. We then analyze their symmetry properties, showing in particular that a strongly symmetric bath evolution can never map a state which is not strongly symmetric symmetric state to a strongly symmetric state; see Proposition~\ref{prop:weak_to_weak}.

Finally, we will also discuss anomaly constraints for strong symmetries, as well as their mixed strong--weak variants. The generalization to strong-weak mixed anomalies is necessary because, as we show in Proposition~\ref{prop:translation}, if a state is clustering\footnote{That is, it exhibits no long-range order; see Lemma~\ref{lemma:clustering} for a precise definition.} and strongly symmetric under translation, then it must be pure. Therefore, for mixed states, imposing strong translation symmetry is often too restrictive.

Let us now spell out our version of anomaly constraints.
 Recall that a state $\psi$ on 1d spin chain is said to satisfy the area law of mutual information if $I(\Gamma;\Gamma^{c}):=S(\psi\|\psi_{\Gamma}\otimes \psi_{\Gamma^{c}})<\const$ for any finite interval $\Gamma\subseteq\Lambda=\z$, where $S(\cdot\|\cdot)$ stands for the relative entropy \cite{Araki1976entropyI,Araki1977entropyII,ohya2004quantum} and $\psi_{\Gamma}$ means the restriction of $\psi$ on $\Gamma$ (a.k.a the reduced density matrix).
\begin{theorem}[Informal version of Theorem \ref{thm:LSM_strong}]\label{thm:LSM_strong_main}
    Let $\psi$ be a state on a quantum spin chain and $\alpha:G\to\G^{\QCA}$ be a symmetry action, then the following conditions are incompatible:
    \begin{enumerate}
        \item $\psi$ is clustering (a.k.a cluster decomposition).
        \item $\psi$ has $\alpha_{g}$ as strong symmetry for all $g\in G$.
        \item $\psi$ satisfies the area law of mutual information.
        \item $\alpha$ has nontrivial anomaly index $1\not=\omega\in\rH^{3}(G;\U)$.
    \end{enumerate}
\end{theorem}
Later, we will generalize above theorem to a more general symmetry conditions, called the von Neumann symmetry Def.~\ref{def:temp_vN_sym}.

\subsection{Relation with previous definitions of SW-SSB}

In this paper, we have argued that in infinite systems, the property of a state that was previously referred to as SW-SSB should simply be interpreted as the absence of strong symmetry (see Section \ref{sec:intro}). Nevertheless, this still raises the question of how the diagnostics for SW-SSB that have previously been proposed \cite{Lessa_2025, Weinstein2025} relate to our criteria for having strong symmetry in the thermodynamic limit.

An important point is that the previously proposed diagnostics cannot be directly applied in infinite systems. To see this, note that when applied in finite systems, they turn out to be very sensitive to global data that is not reflected in the expectation values of local operators. For example, both the fidelity correlators and Renyi 2-correlators of Ref.~\cite{Lessa_2025} give completely different results when applied to the states $\rho_1$ and $\rho_2$ of \eqref{eq:example_states}, despite their local indistinguishability. Thus, they cannot correspond to any well defined quantity in the thermodynamic limit.

Nevertheless, we will argue that diagnostics such as those of Refs.~\cite{Lessa_2025, Weinstein2025} do correspond to some interesting aspects of how the the thermodynamic limit is achieved starting from finite systems, as we will now describe. The discussion will also clarify how the ``SW-SSB'' defined via such diagnostics is related to the absence of strong symmetry in the thermodynamic limit.

Let us first recall how the thermodynamic limit works for regular spontaneous symmetry breaking (not SW-SSB). We can begin with spin chains of size $m$, and then approach the thermodynamic limit by sending $m \to \infty$ (Recall our precise characterization of this limit as described in Section \ref{subsec:taking_thermo}). Let $G$ be a finite group,
and consider a representation $\alpha^{(m)}$ of $G$ on the system of size $m$, which we assume converges to a homomorphism $\alpha : G \to \Aut(\A^\ell)$ as $m \to \infty$ (in the sense described in Section \ref{subsec:taking_thermo}).
Now consider a family of local Hamiltonians $H_{m}$ on the systems of size $m$ that are symmetric under $\alpha^{(m)}$, and consider the Gibbs state described the density matrix
\begin{equation}
    \rho_m= \frac{1}{Z} e^{-\beta H_m}.
\end{equation}
Under suitable assumptions, one expects that $\rho_m$ will converge (in the sense described in Section \ref{subsec:taking_thermo}) to an infinite-volume state $\psi$ that is symmetric under $\alpha$.
However, when we have SSB, one finds that $\psi$ has long-range order -- that is, it has two-point correlation functions that do not decay with distance (in other words, it fails to satisfy the ``clustering property''). However, generally we can write
\begin{equation}
    \psi = \frac{1}{|G|} \sum_{g \in G} \psi_{\mathrm{SSB} }\circ \alpha_{g}
\end{equation}
where $\psi_{\mathrm{SSB}}$ is a state satisfying the clustering property that, however, fails to be invariant under $\alpha$.

Now we turn to the analogous statements for SW-SSB.
Consider a family of states $\psi_{m}$ with strong symmetry $\alpha^{(m)}$.
Furthermore, let us assume that $\psi_m$ converges to some state $\psi$ as $m \to \infty$. We assume this time that there is no regular SSB, in which case it is reasonable to assume that $\psi$ satisfies the clustering property.

We can also consider the canonical purifications $\widetilde{\psi}_{m}$ of $\psi_{m}$, which by the assumption of strong symmetry are invariant under $\id \otimes \alpha^{(m)}$. We expect under suitable conditions that $\widetilde{\psi}_m$ will converge to a state $\widetilde{\psi}'$ on the doubled version of the infinite chain\footnote{We denote this limiting state $\widetilde{\psi}'$ rather than $\widetilde{\psi}$ to distinguish it from the canonical purification of $\psi$ (using the infinite-system version of canonical purification defined in later sections), with which it does not necessarily coincide.}, in which case $\widetilde{\psi}'$ must be invariant under $\id \otimes \alpha$.  However, in general $\widetilde{\psi}'$ might not obey the clustering property. Thus, by analogy to the SSB case, one might expect that it can be decomposed as a sum
\begin{equation}
    \widetilde{\psi}' = \frac{1}{|G|} \sum_{g \in G} \widetilde{\psi}_{\mathrm{SWSSB}} \circ (\id \otimes \alpha_{g}), \label{eq:swssb_decomposition}
\end{equation}
where the state $\widetilde{\psi}_{\mathrm{SWSSB}}$ obeys the clustering property but might not be invariant under $\id \otimes \alpha$. Note that by restricting to the original system, \eqref{eq:swssb_decomposition} implies that
\begin{equation}
\label{eq:swssb_restricted_sum}
    \psi = \frac{1}{|G|}  \sum_{g \in G} (\widetilde{\psi}_{\mathrm{SWSSB}})_R \circ \alpha_{g},
\end{equation}
where $(\widetilde{\psi}_{\mathrm{SWSSB}})_R$ is the restriction to the original system.
In order for $\psi$ to obey the clustering property, it must be the case that all the states appearing in the sum \eqref{eq:swssb_restricted_sum} are equal\footnote{Technically there could be another possibility, which is that the states appearing in the sum \eqref{eq:swssb_restricted_sum} could differ in a manner that goes to zero at spatial infinity. However, this would not be compatible with translation invariance, for example.}, and hence that
$(\widetilde{\psi}_{\mathrm{SWSSB}})_R$ is invariant under $\alpha$ and $(\widetilde{\psi}_{\mathrm{SWSSB}})_R = \psi$.
Thus, $\widetilde{\psi}_{\mathrm{SWSSB}}$ is an extension of $\psi$. Therefore, invoking the results on extensions outlined in Section \ref{sec:sum_main_results}, we conclude that \emph{if} $\psi$ has the strong symmetry, then $\widetilde{\psi}_{\mathrm{SWSSB}}$ is in fact invariant under $\id \otimes \alpha$, and therefore $\widetilde{\psi}' = \widetilde{\psi}_{\mathrm{SWSSB}}$ obeys the clustering property. Therefore, the violation of the clustering property for $\widetilde{\psi}'$ is a signature of the \emph{absence} of strong symmetry in the state $\psi$. 
Indeed, this is precisely the diagnostic for ``SW-SSB'' proposed in Ref.~\cite{Weinstein2025}. (We remark that Ref.~\cite{Weinstein2025} proved that their criterion for SW-SSB is also equivalent to the original one proposed in Ref.~\cite{Lessa_2025} in terms of a fidelity correlator.) 
Thus we see how the ``SW-SSB'' of the previous literature is related to the absence of strong symmetry in the thermodynamic limit.

Note that we have not given an argument for the converse statement, i.e.\ that the strong symmetry being absent in $\psi$ necessarily implies failure of clustering for $\widetilde{\psi}'$. This converse statement seems to hold in the most physically reasonable cases, but it is false in general, if one allows the finite-size states to be sufficiently ``wild'', as the following counterexample demonstrates.
Suppose that we consider an on-site symmetry, and for each $L$ we choose a ``typical'' pure state (i.e.\ choosen from a Haar-random distribution) within the subspace of the Hilbert space comprising the +1 eigenstates of $U^{(m)}(g)$ (where $U^{(m)}(g)$ is the unitary action of the symmetry on the Hilbert space). Such typical states look like the maximally mixed state with respect to local operators, so in this case $\psi_m$ will converge to the maximally mixed state $\psi$, On the other hand, since each $\psi_m$ is pure, its canonical purification is simply $\psi_m \otimes \psi_m^*$  (where $\psi_m^*$ is the complex conjugation of $\psi_m$), which converges to $\widetilde{\psi}' = \psi \otimes \psi^*$. This state indeed obeys the clustering property, despite the fact that $\psi$ does not have the strong symmetry.

This example demonstrates the limitations of the previous diagnostics of SW-SSB, since in some cases they fail to correctly identify the lack of strong symmetry in the infinite system limit. Thus, in general it is better to use the new criteria put forward in the present paper. 

Finally, let us remark that a similar example to the counter-example discussed above was previously considered in Ref.~\cite{Feng_2025}. However, they drew quite different conclusions from it. They considered a family of states for increasing system sizes $m$ that fails to satisfy SW-SSB according to the diagnostics of Ref.~\cite{Lessa_2025}, and therefore characterized it as being ``in a different phase'' compared to another family of states that does satisfy SW-SSB. However, in our view this is not the best interpretation. Rather, both families in fact converge to the \emph{same} infinite-volume state in the thermodynamic limit, and therefore should be viewed as corresponding to the \emph{same} phase, characterized by the absence of strong symmetry in the thermodynamic limit. With this interpretation, the conclusions of Ref.~\cite{Feng_2025} regarding the computational difficulty of distinguishing between phases of matter are no longer applicable. We feel this example also serves to illustrate the danger of focusing on non-local quantities that do not have any well-defined extension to the thermodynamic limit, as it tends to lead to making distinctions that arguably have very little practical relevance (as illustrated by the results of Ref.~\cite{Feng_2025} regarding the computational difficulty of distinguishing between the two families of states).

\subsection{The general decay of conditional mutual information}
\label{sec:cmi_decay}

In this subsection, we wish to make a further remark about the relation between finite-system ``SWSSB'' and the infinite-size results. One feature of finite-system ``SWSSB'' states is that they typically exhibit infinitely long-ranged behavior of the conditional mutual information (CMI) \cite{Sang_2025,Lessa_2025}. It is clear that this does not survive in the infinite-system limit; or in other words, the long-rangedness of the CMI in finite system is probing very non-local correlations. To see this, note that the states $\rho_1$ and $\rho_2$ in \eqref{eq:example_states} are locally indistinguishable, i.e.\ they give the same infinite-volume state, and $\rho_1$ is a product state in which the CMI is completely trivial.

What we want to argue in this subsection is that more generally, in the infinite-volume limit there is \emph{no such thing} as infinitely long-ranged CMI for \emph{any} state. First let us
 define CMI for infinite-size systems. Given a state $\psi$, consider a partition of the lattice as in Fig.~\ref{fig:partition},

\begin{figure}[h!]
    \centering
    \includegraphics[width=0.5\linewidth]{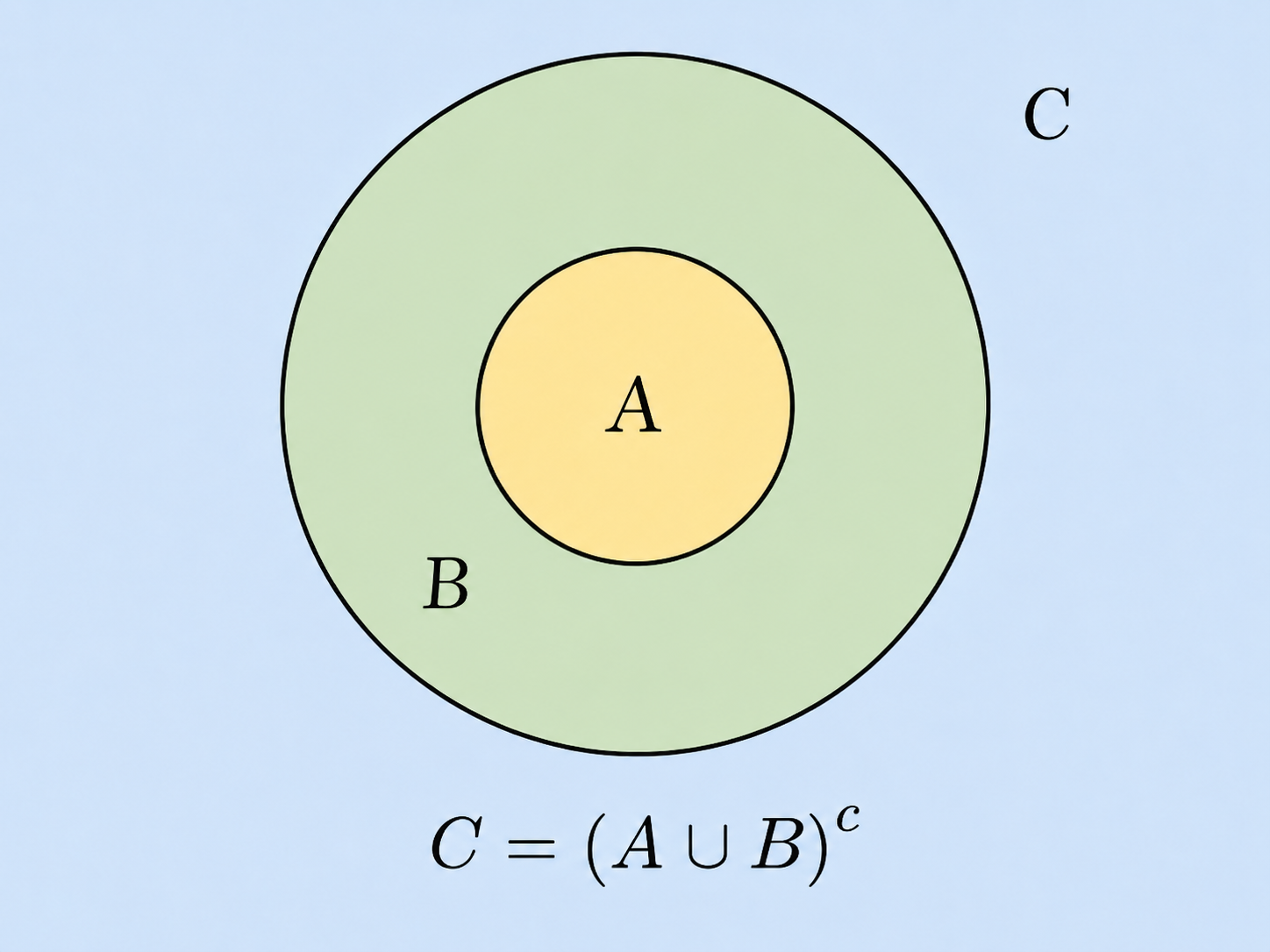}
    \caption{Partition of the lattice}
    \label{fig:partition}
\end{figure}
The CMI of $\psi$ (for this paritition) is defined as
\beq\label{eq:CMI}
I_{\psi}(A:C|B):=S(\psi||\psi_{A}\otimes \psi_{A^{c}})-S(\psi_{AB}||\psi_{A}\otimes \psi_{B})
\eeq
where $\psi_{A}$ denotes the restriction of $\psi$ to the region $A$; the notation $\psi_{B}$, $\psi_{AB}$, etc. is understood similarly. Here $S(\cdot||\cdot)$ denotes the relative entropy. We emphasize that, in our partition, the region $A$ is always bounded, and hence the right-hand side of Eq.~\eqref{eq:CMI} is finite. We will be interested in the behavior of $I_{\psi}(A:C|B)$ in the limit where one increases the size of $B$ while keeping $A$ fixed. We show that this quantity \textit{must vanish in this limit}. Specifically, we have
\begin{lemma}
\label{lemma:cmi_decay}
    Let $B_n$ be any sequence of regions $B_n \subseteq A^c$ such that $B_{0}:=B$, $B_{n} \subseteq B_{n+1}$ and such that every point in $C$ is eventually contained in some $B_n$. Define $C_n = (A \cup B_n)^c$. Then
    \begin{equation}
        \lim_{n \to \infty} I(A:C_n|B_n) = 0.
    \end{equation}
\end{lemma}

This easily follows from the following ``martingale" property of relative entropy:
\begin{lemma}[Prop.~5.23 of \cite{ohya2004quantum}]\label{lemma:martingale}
    Let $\{\Gamma_{n}\}_{n\in\N}$ be any increasing sequence of subsets of $\Lambda$, such that $\bigcup_{n=1}^{\infty}\Gamma_{n}=\Lambda$, then for any two states $\omega,\varphi$ of $\A^{\ell}$,
    \beq
    \lim_{n\to\infty}S(\omega_{n}||\varphi_{n})=S(\omega||\varphi)
    \eeq
    where $\omega_{n}:=\omega|_{\Gamma_{n}}$ and $\varphi_{n}$ is defined similarly.
\end{lemma}

Equivalently, for quantum spin systems, relative entropy may be defined as the limit of relative entropies over increasing finite regions. Lemma~\ref{lemma:martingale} shows that this definition agrees with Araki's relative entropy. Under this formulation, Lemma \ref{lemma:cmi_decay} is automatic.

We emphasize that Lemma \ref{lemma:cmi_decay} also holds for finite-size systems. The only subtlety is that $I_{\psi}(A:C|B)$ may remain nonzero (or does not decay at all) for intermediate choices of $B$, and vanish only when $B=A^{c}$, in which case $C$ is empty. This is the case for the $\mathbb{I}+X$-state, for example. The difference in infinite systems is that one can choose $A$ and $B_n$ to be bounded regions, so that $C_n$ still contains ``most'' of the system (i.e. all except a bounded region) for each value of $n$; nevertheless the zero value of the limit still holds.

We also emphasize that Lemma \ref{lemma:cmi_decay} does not put any bounds on \emph{how fast} the CMI decays; it does not imply, for example, that it must decay exponentially. Generally, non-equilibrium phase transitions between different mixed state phases of matter are expected to have slow (e.g.\ power law) decay of CMI \cite{Sang_2404}.

\section{More background on operator algebras}\label{sec:backgrounds}
In this section we will review some more mathematical properties of operator algebras that we we will require in order to give the rigorous proofs of our results. Readers familiar with operator algebras may skip directly to Sec.~\ref{sec:strong_sym} and refer back to this section as needed.

\subsection{Quasi-local operators}\label{sec:quasi_local}
We will often need to work with operators with tails. To this end, we need to introduce the algebra of quasi-local operators. This algebra is defined to be the completion of $\A^{\ell}$ with respect to the operator norm:
\beq\label{eq:quasi_local_algebra}
\A:=\overline{\A^{\ell}}^{\|\cdot\|}
\eeq
The completion procedure is similar to how one obtains $\R$ from $\mathbb{Q}$ by adding the limits of Cauchy sequences (see \eg Ref.~\cite{Tao2022}). We recall that the operator-norm satisfies:
\beq\label{eq:Banach-inequality}
\| ab\|\leqslant \|a\|\cdot \|b\|
\eeq
Any norm-complete algebra satisfying Eq.~\eqref{eq:Banach-inequality} is called a Banach algebra.

Moreover, $\A$ inherits the $*$-operation from $\A^{\ell}$, and this $*$-operation is compatible with the operator norm:
\beq\label{eq:C*-identity}
\|a^{*}a\|=\|a\|^{2}.
\eeq
In general, a Banach $*$-algebra satisfying Eq.~\eqref{eq:C*-identity} is called a $C^*$-algebra. 
The algebra $\A$ provides an example of such a $C^*$-algebra.

It is important to note that there are many physically relevant objects that are \textbf{not} in $\A$. A prominent example is a finite-range Hamiltonian, which is given by the following formal sum:
\beq
H=\sum_{i\in\Lambda}h_{i}
\eeq
where each $h_{i}$ is supported on a neighborhood $N_{i}$ of $i$, with diameter $\diam(N_{i})<r$ for some constant $r$.
Surely $H$ does not converge in the operator norm. However, we note that $H$ makes sense as a derivation on local operators:
\beq\label{eq:derivation}
\delta_{H}(a)=\sum_{i\in\Lambda}[h_{i},a],\quad\forall a\in\A^{\ell}
\eeq
With the help of Lieb-Robinson bound, it is shown that $\delta_{H}$ can be exponentiated into a $*$-automorphism $\alpha_{t}\in\Aut(\A)$ (called the time evolution of $\delta_{H}$), see \eg Ref.~\cite{Nachtergaele_2006} for details. In general, $\alpha_{t}$ is an \textbf{outer} automorphism of $\A$, meaning that $\alpha_{t}\not=\Ad_{u}$ for any unitary $u\in\A$. Similarly, finite-depth quantum circuits (FDQC) are also \textbf{outer} automorphisms of $\A$.

States on $\A$, as well as pure states, are defined in the same way as in Def.~\ref{def:states}, so we do not repeat the definition here.

Let us also briefly clarify our convention for tensor products of $C^*$-algebras. In general, there are several inequivalent $C^*$-tensor products, depending on the choice of completion procedure (see, e.g., Chap.~6 of \cite{murphy2014c}). In the present paper, however, we work exclusively with spin systems, whose quasi-local algebras are uniformly hyperfinite (UHF). For UHF algebras, the $C^*$-tensor product is unique; see Theorem 6.3.11 of \cite{murphy2014c}. To be specific, throughout this paper, the notation $\otimes$ always refers to the spatial tensor product, both for $C^*$-algebras and for von Neumann algebras (to be introduced in Sec.~\ref{sec:GNS_vN} below).

\subsection{GNS representations and von Neumann algebras}\label{sec:GNS_vN}
In Sec.~\ref{sec:quasi_local}, we have defined the quasi-local algebra $\A$. We also defined the states as positive linear functionals (see Def.~\ref{def:states}) as well as their classical mixture.
However, we are still missing one of the most important ingredients of quantum mechanics, \ie the coherent superposition.
This can be cured by introducing the so-called Gelfand-Naimark-Segal (GNS) construction.

Recall that in perturbative quantum field theory, the Hilbert space is constructed by applying all creation operators to a reference state (i.e.\ the vacuum of the theory). 
In the operator-algebraic setting, the GNS construction replaces the role of creation and annihilation operators with elements of the $C^*$-algebra $\A$.

We first define the GNS ideal associated with a state $\psi$ on $\A$:
\beq
N_{\psi} := \{ a \in \A \mid \psi(a^{*}a) = 0 \}.
\eeq
Elements of $N_{\psi}$ may be regarded as ``annihilation operators'' for $\psi$. 
Accordingly, vector states in the GNS representation can be identified with equivalence classes in the quotient space $\A / N_{\psi}$. If $a\in\A$, we write $[a]$ for its equivalence class.

The quotient $\A / N_{\psi}$ carries a natural inner product given by
\beq\label{eq:GNS_inner_product}
\langle [a], [b] \rangle := \psi(a^{*} b).
\eeq
One can readily verify that this inner product is well defined. 
However, $\A / N_{\psi}$ is not complete with respect to the norm induced by \eqref{eq:GNS_inner_product}; hence, it is only a pre-Hilbert space. The GNS Hilbert space $\cH_{\psi}$ is obtained by completing $\A/N_{\psi}$ with respect to the inner product Eq.~\eqref{eq:GNS_inner_product}. Besides, we have a representation $\pi_{\psi}$ of $\A$ defined by:
\beq\label{eq:inner_prod}
\pi_{\psi}(a)([b]):=[ab], \quad\forall\,a,b\in\A
\eeq
Lastly, $|\psi\ra:=[1_{\A}]$ represents $\psi$ in $\cH_{\psi}$:
\beq
\la\psi|\pi_{\psi}(a)|\psi\ra = \psi(a),\quad\forall\,a\in\A
\eeq
\begin{definition}\label{def:GNS_triple}
    Given a state $\psi$ on $\A$, $(\pi_{\psi},\cH_{\psi},|\psi\ra)$ is called the GNS triple of $\psi$.
\end{definition}
Sometimes, it will be convenient to talk about the GNS representation of a general positive linear functional $\rho$, \ie it may not be normalized. Note that we always have\footnote{To see this, note $a^{*}a\leqslant \|a\|^{2}\times 1_{\A}$, if $\rho(1_{\A})=0$ we would have $\rho(a^{*}a)\leqslant \|a\|^{2}\rho(1_{\A})=0$ hence $\rho$ is identically 0.} $\rho(1)>0$ as long as $\rho\not=0$. The GNS triple of $\rho$ is defined to be the GNS triple of $\rho(1_{\A})^{-1}\rho$.

The vector $|\psi\rangle$ is distinguished in the above construction because $\pi_{\psi}(\A)|\psi\rangle$ is dense in $\cH_{\psi}$. Equivalently, $\cH_{\psi}$ is the completion of $\pi_{\psi}(\A)|\psi\rangle$ with respect to the inner product in Eq.~\eqref{eq:inner_prod}:
\beq
\cH_{\psi}=\overline{\pi_{\psi}(\A)|\psi\rangle}^{\langle\cdot\,|\,\cdot\rangle}.
\eeq
When no confusion can arise, we suppress explicit reference to the inner product from now on. Such a vector is called \emph{cyclic}, and a representation admitting a cyclic vector is called a \emph{cyclic representation}. More generally:
\begin{definition}\label{def:cyclic}
    Let $\pi:\A\to\B(\cH)$ be a representation. A vector $|\phi\rangle\in\cH$ is called \emph{cyclic} if $\pi(\A)|\phi\rangle$ is dense in $\cH$. In this case, $\pi$ is called a \emph{cyclic representation}.
\end{definition}

\begin{remark}\label{remark:uniqueness_GNS}
We remark that the GNS triple is unique only up to unitary equivalence. 
If $(\pi', \cH',|\psi'\ra)$ is another cyclic representation with
\beq
\la\psi'|\pi'(a)|\psi'\ra=\psi(a),\quad\forall\,a\in\A
\eeq
then there exists a unique unitary operator 
$U : \cH_{\psi} \to \cH'$ such that
\beq
\begin{split}
U |\psi\ra &= |\psi'\ra, \\
\pi'(a) &= U \pi_{\psi}(a) U^{\dagger}, \qquad \forall\, a \in \A .
\end{split}
\eeq
\end{remark}

This uniqueness of GNS representations plays an essential role in the discussion of the symmetries of states. 
For instance, let $\alpha \in \Aut(\A)$ be a $*$-automorphism and let $\psi$ be a state satisfying $\psi \circ \alpha = \psi$. 
A natural question is how $\alpha$ acts on the GNS Hilbert space $\cH_{\psi}$. 
By the uniqueness of the GNS triple, there exists a unique unitary operator 
$U_{\alpha} \in \B(\cH_{\psi})$ implementing $\alpha$, i.e.,
\beq
\begin{split}
U_{\alpha} |\psi\ra &= |\psi\ra, \\
\pi_{\psi} \circ \alpha (a) &= U_{\alpha} \pi_{\psi}(a) U_{\alpha}^{\dagger}, 
\qquad \forall\, a \in \A .
\end{split}
\eeq

\begin{remark}\label{remark:inner1}
    We note that although $U_{\alpha}\in\B(\cH_{\psi})$, in general $U_{\alpha}\not\in\pi_{\psi}(\A)$ unless $\alpha$ is an inner-automorphism.
\end{remark}
Given a representation, it is natural to ask if it is irreducible. The reducibility is closely related to the notion of invariant subspace and subrepresentations, which we will define now.
\begin{definition}\label{def:subrep}
    Let $\pi:\A\to\B(\cH)$ be a representation of $\A$ on $\cH$. If $\cK\subseteq\cH$ be a subspace such that $\pi(\A)\cK=\cK$, then $\cK$ is called an invariant subspace while $\pi(\A)|_{\cK}$ is called a subrepresentation of $\pi$. The representation $\pi$ is \textbf{irreducible} if there is no nontrivial invariant subspace (\ie $\cK\not=\{0\},\cH$).
\end{definition}
We have the following classical result
\begin{lemma}[Theorem 2.3.19 of Ref.~\cite{bratteli2013operator1}]\label{lemma:Schur}
    Given a state $\psi$ on $\A$,
 the following three statements are equivalent:
    \begin{enumerate}
        \item $\pi_{\psi}$ is irreducible.
        \item The Schur's lemma holds for $\pi_{\psi}$, \ie $\pi_{\psi}(\A)':=\{x\in\B(\cH_{\psi}):[x,\pi_{\psi}(\A)]=0\}=\bbC\cdot1_{\cH_{\psi}}$.
        \item The state $\psi$ is pure.
    \end{enumerate}
\end{lemma}
We remark that $\pi_{\psi}(\A)'$ in this lemma is called the \emph{commutant} of $\pi_{\psi}(\A)$ in $\mathcal B(\mathcal H_{\psi})$.
In Remark~\ref{remark:inner1}, we observed that $U_{\alpha}\notin \pi_{\psi}(\A)$ in general.
If $\psi$ is a pure state, then the \emph{double} commutant satisfies \footnote{Note that $\pi_{\psi}(\A)'''=\pi_{\psi}(\A)'$ automatically holds, so there will be nothing new by taking further commutants.}:
\[
\pi_{\psi}(\A)''=\mathcal B(\mathcal H_{\psi}).
\]
Consequently, as long as $\psi$ is pure, we always have $U_{\alpha}\in \pi_{\psi}(\A)''$.

This suggests that the double commutant $\pi_{\psi}(\A)''$ may be an interesting object in its own right.
Indeed, as discussed before, the time evolution $\alpha_t$ is in general an outer automorphism of $\A$.
However, for a ground state\footnote{Given a local Hamiltonian $H$ (viewed as a derivation), a state $\psi$ is called a ground state if $-i\,\psi(a^*\delta_H(a))\ge 0$ for all $a$ in the domain of $\delta_H$. Such a state can be mixed in general.} $\psi$, one can show that $\alpha_t$ can always be implemented by a unitary $U_t\in \pi_{\psi}(\A)''$, even when $\psi$ is mixed.
See Proposition~5.3.19 of Ref.~\cite{bratteli2013operator2} for more details.

Based on above discussion,
\begin{definition}\label{def:vN_algebra}
    Given a state $\psi$ of $\A$, we define the double commutant $\cM_{\psi}:=\pi_{\psi}(\A)''$ to be the von Neumann (vN) algebra generated by $\psi$.
\end{definition}
vN algebras play an essential role in the discussion of symmetries of mixed-state phases, as we will see later.

Before ending this subsection, we mention a powerful result in the theory of von Neumann algebras, called Kaplansky's density theorem. To this end, let us define some useful topologies on $\B(\cH)$.
\begin{definition}\label{def:SOT_WOT}
    The strong operator topology (SOT) and weak operator topology (WOT) are defined as follows:
    \begin{enumerate}
        \item A sequence $x_{n}\in\B(\cH)$ SOT-converges to $x\in\B(\cH)$  if for any $|\psi\ra\in\cH$, we have
        \beq
        \|(x-x_{n})|\psi\ra\|\stackrel{n\to\infty}{\longrightarrow}0
        \eeq
        We simply write $x_{n}\stackrel{\SOT}{\longrightarrow}x$ in this case.
        \item A sequence $x_{n}\in\B(\cH)$ WOT-converges to $x\in\B(\cH)$  if for any $|\phi\ra,|\psi\ra\in\cH$, we have
        \beq
        \la\phi|(x-x_{n})|\psi\ra\stackrel{n\to\infty}{\longrightarrow}0
        \eeq
        We simply write $x_{n}\stackrel{\WOT}{\longrightarrow}x$.
    \end{enumerate}
    Given a $*$-subalgebra $\mathscr{C}\subseteq\B(\cH)$, its SOT closure $\ovl{\sC}^{\SOT}$ is defined to be the SOT limits of all SOT-convergent sequences in $\sC$. The WOT closure is defined similarly.
\end{definition}
To relate von Neumann algebras and these two topologies, we have the following famous result, due to von Neumann himself.
\begin{theorem}[von Neumann's double commutant theorem]\label{theorem:double_commutant}
    Let $\sC\subseteq\B(\cH)$ be a *-subalgebra, then
    \beq
    \ovl{\sC}^{\WOT}=\ovl{\sC}^{\SOT} = \sC''
    \eeq
\end{theorem}
See \eg theorem I.9.1.1 of Ref.~\cite{Blackadar:2006OA} for a proof. This result is remarkably convenient because the double commutant is defined in a purely algebraic way.

The power of SOT and WOT can be seen from the following two results.
\begin{lemma}[Ref.~\cite{tao2009notes11}]\label{lemma:WOTcompact}
    The unit ball (of the operator norm) in $\B(\cH)$ is compact in WOT.
\end{lemma}
Although SOT and WOT are very powerful, one has to be careful when dealing with them.
\begin{remark}\label{remark:SOT_WOT}
    In general, operator multiplication in $\B(\cH)$ is \textbf{not} jointly continuous in either the strong or weak operator topology. That is, $x_n \to x$ and $y_n \to y$ (in SOT or WOT) do not in general imply $x_n y_n \to xy$.

    Besides, the $*$-operation is non-continuous in SOT. See Ref.~\cite{tao2009notes11} for more discussions.
\end{remark}

\subsection{Equivalences}
Given two representations of $\A$, it is useful to know if and when they are unitarily equivalent. The following lemma can be helpful in this situation.
\begin{lemma}\label{lemma:unitary_equivalence}
    Consider two \textbf{pure} states $\psi_{1},\psi_{2}$ of $\A$, write $(\pi_{i},\cH_{i},|\psi_{i}\ra),i=1,2$ for their corresponding GNS triples. Then the following conditions are equivalent:
    \begin{enumerate}
        \item The representations $\pi_{1}$ is unitarily equivalent to $\pi_{2}$, \ie there exists a unitary $U:\cH_{1}\to \cH_{2}$ such that
        \beq
        U\pi_{1}(a)U^{\dagger}=\pi_{2}(a),\quad\forall\,a\in\A
        \eeq
        \item (Kadison transitivity) there exists a unitary $u\in\A$ such that $\psi_{1}=\psi_{2}\circ \Ad_{u}$.
        \item For any $\epsilon>0$, there exists a finite region $\Gamma_{\epsilon}\subseteq\Lambda$, such that
        \beq
        |\psi_{1}(a)-\psi_{2}(a)|<\epsilon\|a\|,\quad\forall\,a \in\A_{\Gamma_{\epsilon}^{c}}
        \eeq
        where $\Gamma_{\epsilon}^{c}:=\Lambda\setminus\Gamma_{\epsilon}$. This condition means $\psi_{1}$ and $\psi_{2}$ are asymptotically equal near the spatial infinity.
    \end{enumerate}
    We say $\psi_{1}$ is unitarily equivalent to $\psi_{2}$ and write $\psi_{1}\simeq\psi_{2}$ if any of the above three conditions are satisfied. We also say they fall into the same superselection sector in this case.
\end{lemma}
See Corollary 2.6.11 of Ref.~\cite{bratteli2013operator1} and Theorem 10.2.6 of Ref.~\cite{kadison1997fundamentals2} for a proof.

For two pure states $\psi_{1}\not\simeq\psi_{2}$, we say that they are disjoint or in different superselection sectors. The following result characterizes disjoint states.
\begin{lemma}[Corollary 10.3.8 of Ref.~\cite{kadison1997fundamentals2}]\label{lemma:disjoint}
    For disjoint pure states $\psi_{1},\psi_{2}$, we have
    \beq
    \|\psi_{1}-\psi_{2}\|:=\sup_{a\in\A,\|a\|=1}|\psi_{1}(a)-\psi_{2}(a)|=2
    \eeq
\end{lemma}
Note that $\|\psi_{1}-\psi_{2}\|$ is nothing but twice of the trace-norm distance in the usual quantum mechanics. Thus Lemma \ref{lemma:disjoint} simply says states from different superselection sectors are “orthogonal” to each other.

The unitary equivalence and Lemma \ref{lemma:unitary_equivalence} are useful for pure states. In order to study mixed states, a slightly weaker equivalence, called \textit{quasi-equivalence} can be more helpful. In order to define quasi-equivalences, let us define the notion of subrepresentations.

\begin{definition}\label{def:quasi-equivalence}
    Two states $\psi_{1},\psi_{2}$ of $\A$ are \textbf{quasi-equivalent} if there exists an isomorphism $f:\cM_{1}\to\cM_{2}$ such that
    \beq
    f(\pi_{1}(a))=\pi_{2}(a),\quad\forall\,a\in\A
    \eeq
    where $\cM_{i}$ is the vN algebra generated by $\psi_{i},\,i=1,2$. We write $\psi_{1}\sim\psi_{2}$ in this case.

    Similarly, we say $\psi_{1}$ is disjoint from $\psi_{2}$ if they do not have any quasi-equivalent subrepresentations.
\end{definition}
For pure states, quasi-equivalence is nothing but unitary equivalence defined as in Lemma \ref{lemma:unitary_equivalence}. Thus quasi-equivalence indeed generalizes unitary equivalence between pure states. To give a nontrivial example of quasi-equivalent states which are not unitarily equivalent, consider two equivalent pure states $\omega_{1}\simeq\omega_{2}$, then we have quasi-equivalence $\omega:=\frac{1}{2}(\omega_{1}+\omega_{2})\sim\omega_{1}\sim\omega_{2}$; See Example.~\ref{example:factors} below.

There are other useful characterization for quasi-equivalence between states:
\begin{lemma}[Corollary 10.3.4 of Ref.~\cite{kadison1997fundamentals2}]\label{lemma:quasi_eq_rep}
    Two states $\psi_{1}\sim\psi_{2}$ iff $\pi_{1}$ has \textbf{no} subrepresentation that is disjoint from $\pi_{2}$ and vice versa. In other words, every subrepresentation of $\pi_{1}$ has a subrepresentation which is equivalent to some subrepresentation of $\pi_{2}$ and vice versa.
\end{lemma}

In order to detect subrepresentations, the following lemma can be quite convenient.
\begin{lemma}\label{lemma:majorized}
    Let $\psi:\A\to\bbC$ be a state and $\rho$ be another positive linear functional with $\rho\leqslant\psi$, then
    \begin{enumerate}
        \item There exists a positive operator $T\in\pi_{\psi}(\A)'$ with $\|T\|\leqslant 1$, such that $\rho(a)=\la\psi|T\pi_{\psi}(a)|\psi\ra$.
        \item The GNS representation of $\rho$ is equivalent to a subrepresentation of $\pi_{\psi}$.
    \end{enumerate}
\end{lemma}
\begin{proof}
    The first statement is Theorem 2.3.19 of Ref.~\cite{bratteli2013operator1}, so we will omit its proof here. Below we assume the existence of such positive $T$ and prove the second assertion.

    Defining $|\rho\ra:=T^{1/2}|\psi\ra$. 
    Consider the invariant subspace subspace $\cK_{\rho}:=\ovl{\pi_{\psi}(\A)|\rho\ra}\subseteq\cH_{\psi}$, where the completion is taken with respect to the inner product. This is a subspace in general since $|\rho\ra$ may not be a cyclic vector of $(\pi_{\psi},\cH_{\psi},|\psi\ra)$.
    
    Therefore, we end up with a cyclic subrepresentation on $\cK_{\rho}$, with
    \beq
    \la\rho|\pi_{\psi}(a)|\rho\ra=\rho(a),\quad\forall\,a\in\A
    \eeq
    By Remark.~\ref{remark:uniqueness_GNS}, this subrepresentation must be unitarily equivalent to the GNS representation of $\rho$.
\end{proof}

It is natural to ask whether a GNS representation can be decomposed into simpler constituents (\eg, irreducible representations). We emphasize that, in general, such a naive decomposition is \textbf{impossible}. Nevertheless, there exists a closely related structural result, which we describe below.

To this end, we need the notion of “simple” vN algebras, \ie the factors.
\begin{definition}\label{def:factor}
    A vN algebra $\cM_{\psi}$ is a factor if $\cM_{\psi}\cap\cM_{\psi}'=\bbC\cdot\id$. The state $\psi$ is called a factor state if $\cM_{\psi}$ is a factor.
\end{definition}
The physical importance of factor states can be seen from
\begin{lemma}[Theorem 2.6.10 of Ref.~\cite{bratteli2013operator1}]\label{lemma:clustering}
    Given a state $\psi\in\cS(\A)$, then
    the following two conditions are equivalent:
    \begin{enumerate}
        \item The state $\psi$ is a factor state.
        \item The state $\psi$ is clustering, \ie $\forall\,a\in\A^{\ell},\,\epsilon>0$, there exists a finite region $\Gamma$ (depending on $a$ and $\epsilon$), such that 
        \beq
        |\psi(ab)-\psi(a)\psi(b)|<\epsilon\|a\|\cdot\|b\|,\quad\forall b\in\A_{\Gamma^{c}}
        \eeq
        where $\Gamma^{c}:=\Lambda\setminus\Gamma$ is the complement.
    \end{enumerate}
\end{lemma}
This is one of the most surprising results in the algebraic approach to quantum many-body physics. There is a generalization of Lemma \ref{lemma:unitary_equivalence} for factor states.
\begin{lemma}[Corollary 2.6.11 of Ref.~\cite{bratteli2013operator1}]\label{lemma:quasi_equiv}
    For two \textbf{factor} states $\psi_{1},\psi_{2}$, the following two statements are equivalent:
    \begin{enumerate}
        \item There is a quasi-equivalence $\psi_{1}\sim\psi_{2}$.
        \item $\forall\,\epsilon>0$, there exists a finite region $\Gamma_{\epsilon}\subseteq\Lambda$, such that
        \beq
        \|(\psi_{1}-\psi_{2})|_{\Gamma_{\epsilon}^{c}}\|<\epsilon
        \eeq
        This means they become coincident when approaching the spatial infinity.
    \end{enumerate}
    However, we note that the Kadison transitivity is not true for quasi-equivalence, \ie there may not be a unitary $u\in\A$ such that $\psi_{1}=\psi_{2}\circ\Ad_{u}$.
\end{lemma}

There is a type classification for vN factors, see \eg Ref.~\cite{Blackadar:2006OA}. We briefly recall it below.
\begin{definition}
    Given a vN factor $\cM_{\psi}$, we call it:
    \begin{enumerate}
        \item type $I_{n}$ if $\cM_{\psi}\simeq \B(\cH)$ for some Hilbert space $\cH$ with $\dim_{\bbC}\cH=n\in\mathbb{N}\cup\{\infty\}$.
        \item type $II$ if it is not type $I$ and admits a trace\footnote{A trace $\tr$ on $\cM_{\psi}$ is a positive linear functional on $\cM_{\psi}$ whose value can diverge, satisfying $\tr(a^* a)=\tr(a a^*)$ for any $a\in\cM_{\psi}$. It is called finite if it values in $\R$ and called semi-finite if it is finite only on a nontrivial subset of $\cM_{\psi}$.} (finite or semi-finite).
        \item type $III$ if it is neither type $I$ nor type $II$. 
    \end{enumerate}
    The type of a factor state  $\psi$ and its GNS representation is defined via the type of $\cM_{\psi}$.
\end{definition}

\begin{example}\label{example:factors}
We will not go into the details of this classification. Instead, we will provide some examples for each type.
\begin{enumerate}
    \item Any type $I$ factor $\psi$ can be decomposed into
    \beq
    \psi=\sum_{i=1}^{\infty}\lambda_{i}\psi_{i}
    \eeq
    where $\lambda_{i}\searrow0$ and add up to $1$, each $\psi_{i}$ is pure and $\psi_{i}\simeq\psi_{j}$ for any $i,j$. Equivalently, this means the GNS representation of $\psi$ decomposes into direct sum of equivalent irreducible representations. We also have $\psi\sim\psi_{i}$ for all $i$. For type I factors, the mixture is (approximately) localized in a finite region since it looks like a pure state near the spatial infinity.
    \item The maximally mixed state (a.k.a tracial state) $\tau$, which is uniquely defined by $\tau(ab)=\tau(ba)$ for all $a,b\in\A$. This state is of type $II$.
    \item A finite-temperature state (a.k.a KMS state) is type $III$.
    \item Let $\psi$ be the unique ground state of the spin-$1/2$ anti-ferromagnetic Heisenberg chain. Then its half-chain restriction $\psi|_{\geqslant0}$ is type $III$ \cite{KEYL_2006,KEYL_2008}.
\end{enumerate}
\end{example}
From the last example above, we see that the type of a factor can serve as a measure of entanglement. This is indeed the case; see \eg Ref.~\cite{KEYL_2006,KEYL_2008,matsui2011boundedness} for related discussions. We will return to this point when we discuss the so-called split property.

Mathematically, the relationship between factor states is particularly simple as indicated by the following lemmas.

\begin{lemma}[Proposition 10.3.2 of Ref.~\cite{kadison1997fundamentals2}]\label{lemma:disjoint_quasieq}
    Let $\psi_{1},\psi_{2}$ be two factor states, then they are either disjoint or quasi-equivalent.
\end{lemma}
\begin{lemma}[Theorem 7.3.6 and Proposition 10.3.2 of Ref.~\cite{kadison1997fundamentals2}]\label{lemma:factor_quasieq}
    If $\pi$ is a factor representation of $\A$, then any subrerepsentation of $\pi$ is quasi-equivalent to itself.

    Equivalently, let $\psi$ be a factor state and $\rho$ is a positive linear functional such that $\rho\leqslant\psi$ (\ie $\rho$ is majorized by $\psi$), then their GNS representations are quasi-equivalent.
\end{lemma}

We are ready for the decomposition theorem of von Neumann algebras.
\begin{lemma}[Theorem III.1.4.7 of Ref.~\cite{Blackadar:2006OA}]
    Any von Neumann algebra $\cM_{\psi}$ can be uniquely written as:
    \beq
    \cM_{\psi}=\cM_{I}\oplus\cM_{II}\oplus \cM_{III}
    \eeq
    where each $\cM_{\mathfrak{i}}$ is a direct integral\footnote{This is a continuous version of direct sum.} of type $\mathfrak{i}$ factors, $\mathfrak{i}= I,II,III$.
\end{lemma}
As a corollary,
\begin{corollary}[Theorem III.5.1.7 of Ref.~\cite{Blackadar:2006OA}]
    Any representation $\pi$ of $\A$ can be decomposed as
    \beq
    \pi=\pi_{I}\oplus \pi_{II}\oplus \pi_{III}
    \eeq
    where $\pi_{\mathfrak{i}}$ is a direct integral of type $\mathfrak{i}$ factor representations, $\mathfrak{i}=I,II,III$.
\end{corollary}

\section{Strong symmetries and von Neumann (vN) symmetries of States}\label{sec:strong_sym}
In this section, we show how anomalous symmetries can constrain the mixed states with such symmetries. Given a (not necessarily pure or factor) state $\psi:\A\to\bbC$, consider $\alpha\in\Aut(\A)$ that leaves $\psi$ invariant. By the uniqueness of GNS representation (Remark.~\ref{remark:uniqueness_GNS}), there exists a unitary operator $U_{\alpha}\in\B(\cH_{\psi})$ such that
\beq
U_{\alpha}\pi_{\psi}(a)U_{\alpha}^{\dagger} = \pi_{\psi}(\alpha(a)),\quad\forall\,a\in\A
\eeq
This $U_{\alpha}$ is unique once we impose $U_{\alpha}|\psi\ra=|\psi\ra$ and we say $U_{\alpha}$ implements $\alpha$ on $\cH_{\psi}$.

Below we give several equivalent definitions for strong symmetries.
\begin{definition}[Strong symmetry]\label{def:strong_sym}
    The symmetry $\alpha$ is a \textbf{strong symmetry} of $\psi$ if the unitary implementation $U_{\alpha}\in \cM_{\psi}$, assuming $U_{\alpha}|\psi\ra=|\psi\ra$.
\end{definition}
This definition is quite general and concise, but it is not physically intuitive. Indeed, it is not obvious why Def.~\ref{def:strong_sym} is equivalent to Def.~\ref{def:strong_sym_main} at all. This equivalence is the content of Theorem \ref{thm:equi_def}.

Below, we provide another (equivalent) characterization of strong symmetries based on the purification map introduced by Woronowicz in \cite{Woronowicz:1972fr}. To this end, we introduce $\overline{\A}$, the complex conjugation of $\A$. Note that there is a canonical anti-linear automorphism $j$ on $\oA\otimes\A$ defined as:
\beq\label{eq:jpositive}
j(\ovl{b}\otimes a):=\ovl{a}\otimes b
\eeq
A state $\tilde{\psi}$ on $\oA\otimes\A$ is called \textbf{$j$-positive} if
\beq
\tilde{\psi}(\ovl{a}\otimes a)\geqslant0,\quad\forall\,a\in\A
\eeq
It is easy to show that any $j$-positive state is $j$-invariant. Besides, $\tilde{\psi}$ is called \textbf{exact} \footnote{This condition is often called Haag duality in the literature, our terminology follows Ref.~\cite{Woronowicz:1972fr}.}if 
\beq\label{eq:exact}
\tilde{\pi}(\ovl{1}\otimes \A)' = \tilde{\pi}(\ovl{\A}\otimes 1)''
\eeq
where $\tilde{\pi}$ is the GNS representation of $\tilde{\psi}$ and the commutants are taken inside $\B(\cH_{\tilde{\psi}})$.

\begin{lemma}[Theorem 1.1, 1.2 of \cite{Woronowicz:1972fr}]\label{lemma:purify}
    Let $\psi$ be a factor state of $\A$, then $\psi$ can be uniquely\footnote{The uniqueness is proved in Ref.~\cite{Woronowicz1973}.} purified into a $j$-positive exact pure state $\tilde{\psi}$ on $\ovl{\A}\otimes \A$, \ie $\tilde{\psi}|_{\A} =\psi$. Furthermore, two factor states are quasi-equivalent iff $\tilde{\psi}_{1}\simeq\tilde{\psi}_{2}$.
\end{lemma}

\begin{remark}
    For finite-dimensional $C^*$-algebras, this purification map is nothing but the canonical purification. If $\psi$ is not a factor state, the construction still applies but the resulting state $\tilde{\psi}$ is no longer pure. In this more general case, we will call $\tilde{\psi}$ the canonically doubled state of $\psi$.
\end{remark}

\begin{remark}   
    Remarkably, we note that the canonical purification preserves clustering property for factor states. Following the spirit of Ref.~\cite{Li:2026qhc}, this can be seen as a consequence of the universal vanishing behavior of CMI in infinite systems as we have discussed in Lemma.~\ref{lemma:cmi_decay}. We remark that since pure states in infinite systems necessarily obey the cluster property, it follows that if a state obeys the cluster property (implying it is a factor state), then so does its canonical purification. One can view this result through the lens of the results of Ref.~\cite{Li:2026qhc}, in which it was argued that in finite systems the decay of correlations for the canonical purification is bounded by the slower of the decay of two-point correlations and of CMI. In infinite sytems, CMI \emph{always} decays (see Section \ref{sec:cmi_decay}), so the cluster property (decay of two-point correlations) is the only condition required to ensure the cluster property for the canonical purification.
\end{remark}

The canonical purification is one example of purification as is defined in Def.~\ref{def:extension} with $\cB= \oA$. 

\begin{lemma}[\cite{Woronowicz:1972fr}]
    Any cyclic representation of $\A$ which is quasi-equivalent to $\pi_{\psi}$ is a subrepresentation of $\pi_{\tilde{\psi}}$.
\end{lemma}
 Let us turn to general extensions.
\begin{lemma}\label{lemma:purification}
    Assuming $\psi$ is a state of $\A$, then
    for any extension $\psi'$ on $\cB\otimes \A$ with GNS triple $(\pi',\cH',|\psi'\ra)$, there is a quasi-equivalence $\pi'|_{\A}\sim\pi_{\psi}$.
\end{lemma}
\begin{proof}
    According to Lemma \ref{lemma:quasi_eq_rep}, it amounts to prove that any subrepresentation of $\pi'|_{\A}$ contains a subrepresentation which is equivalent to a subrepresentation of $\pi_{\psi}$ and vice versa. We note that $\pi_{\psi}$ is readily equivalent to a subrepresentation of $\pi'|_{\A}$, so this direction is trivial.

    For the other direction, let $\pi\subseteq\pi'|_{\A}$ be a subrepresentation of $\pi'|_{\A}$ on an invariant subspace on $\cK\subseteq\cH'$. Let $(\pi'(1\otimes\A))'\ni P:\cH'\to\cK$ be the invariant projection to $\pi$. Write $|\rho\ra:=P|\psi'\ra$ and consider the completion $\cK_{\rho}:=\ovl{\pi'(1\otimes\A)|\rho\ra}$,  we then end up with a cyclic subrepresentation $\pi_{\rho}\subseteq\pi$ on $\cK_{\rho}$. Defining
    \beq
    \rho(a):=\la\rho|\pi'(1\otimes a)|\rho\ra,\quad\forall\,a\in\A
    \eeq
    We note $\pi_{\rho}$ is nothing but the GNS representation of $\rho$. On the other hand, for any positive $a\in\A$,
    \beq
    \begin{split}
        \rho(a)&=\|P\pi'(1\otimes a^{1/2})|\psi'\ra\|^{2} \\
        &\leqslant\|\pi'(1\otimes a^{1/2})|\psi'\ra\|^{2}\\
        &= \la\psi'|\pi'(1\otimes a)|\psi'\ra\\
        &=\psi'(1\otimes a)\\
        &=\psi(a)
    \end{split}
    \eeq
    Thus $\pi_{\rho}$ is equivalent to a subrepresentation of $\pi_{\psi}$, by Lemma \ref{lemma:majorized}. This completes the proof.
\end{proof}
\begin{remark}
    For more general extension of the form $\A\subseteq\sC$ instead of $\A\subseteq\cB\otimes \A$, Lemma \ref{lemma:purification} still holds.
\end{remark}

Below we show the equivalence of definitions for strong symmetries in infinite-volume systems.

\begin{theorem}[strong symmetry via extension and purification]\label{thm:equi_def}
Let $\psi$ be a state on $\A$, and let $\alpha\in\Aut(\A)$ be a symmetry of $\psi$. 
Then the following are equivalent:
\begin{enumerate}
    \item $\psi$ is strongly symmetric under $\alpha$ (see Def.~\ref{def:strong_sym}).
    \item The canonically doubled state $\tilde{\psi}$ on $\oA\otimes \A$ is symmetric under $\id_{\oA}\otimes \alpha$.
    \item Every extension $\psi'$ of $\psi$ to $\cB\otimes \A$ is symmetric under $\id_{\cB}\otimes \alpha$.
\end{enumerate}
\end{theorem}
\begin{proof}
    Throughout this proof, we write $(\pi,\cH,|\psi\ra)$ and $(\tilde{\pi},\tilde{\cH},|\tilde{\psi}\ra)$ for the GNS triples of $\psi$ and $\tilde{\psi}$ respectively.
    
    Since $\tilde{\pi}$ is a cyclic representation, we have $(\tilde{\pi}(\ovl{1}\otimes \A),\tilde{\cH}, |\tilde{\psi}\ra)$ is a another GNS triple of $\psi$ (since $\tilde{\psi}|_{\A} = \psi$ by construction). By the uniqueness of GNS representation (see Remark.~\ref{remark:uniqueness_GNS}), there exists an unitary $V:\cH\to\tilde{\cH}$ such that $V|\psi\ra = |\tilde{\psi}\ra$ and
    \beq\label{eq:intertwiner}
    \tilde{\pi}(1_{\oA}\otimes a)&=V\pi(a)V^{\dagger},\quad\forall\,a\in\A
    \eeq
    Importantly, by taking WOT closure on both sides, we have
    \beq\label{eq:intertwiner_WOT}
    \tilde{\pi}(1_{\oA}\otimes \A)''&=V\cM_{\psi}V^{\dagger}
    \eeq

    Below, we show $1\Rightarrow 2$ by assuming $\psi$ is strongly symmetric under $\alpha$, thus $\exists\,U_{\alpha}\in\cM_{\psi}$ implementing $\alpha$ and $U_{\alpha}|\psi\ra = |\psi\ra$.
    Utilizing Eq.~\eqref{eq:intertwiner_WOT}, we have
    \beq
    \tilde{U}_{\alpha}:= VU_{\alpha}V^{\dagger}\in \tilde{\pi}(\ovl{1}\otimes \A)''
    \eeq
    Note that
    \beq
    \tilde{U}_{\alpha}|\tilde{\psi}\ra=VU_{\alpha}V^{\dagger}|\tilde{\psi}\ra=|\tilde{\psi}\ra
    \eeq
    Besides, $\Ad_{\tilde{U}_\alpha}$ acts trivially on $\tilde{\pi}(\ovl{\A}\otimes 1_{\A})$ since $\tilde{U}_{\alpha}\in\tilde{\pi}(1_{\oA}\otimes \A)''=\tilde{\pi}(\ovl{\A}\otimes 1_{\A})'$ by Eq.~\eqref{eq:exact} while on $\tilde{\pi}(1_{\oA}\otimes \A)$ it implements $\alpha$. Thus we conclude $\tilde{\psi}\circ(\id_{\oA}\otimes\alpha)=\tilde{\psi}$ and this symmetry is implemented by $\tilde{U}_{\alpha}$.

    Conversely, to show that $2\Rightarrow 1$, if $\tilde{\psi}\circ(1_{\oA}\otimes\alpha)=\tilde{\psi}$, there exists a unitary $\tilde{U}_{\alpha}'\in \B(\cH_{\tilde{\psi}})$ implementing $1\otimes\alpha$ and $\tilde{U}_{\alpha}'|\tilde{\psi}\ra= |\tilde{\psi}\ra$. Obviously
    \beq
    \tilde{U}'_{\alpha}\tilde{\pi}(\ovl{a}\otimes 1)\tilde{U}_{\alpha}^{'\dagger} = \tilde{\pi}\circ (1\otimes\alpha)(\ovl{a}\otimes 1)=\tilde{\pi}(\ovl{a}\otimes 1)
    \eeq
    Therefore $\tilde{U}_{\alpha}'\in\tilde{\pi}(\ovl{\A}\otimes 1)'=\tilde{\pi}(1\otimes \A)''$ by Eq.~\eqref{eq:exact}. So by Eq.~\eqref{eq:intertwiner_WOT} again, we have
    \beq
    U_{\alpha}' := V^{\dagger}\tilde{U}_{\alpha}'V\in\cM_{\psi}
    \eeq
    It is easy to check that $U_{\alpha}'|\psi\ra=|\psi\ra$ and $U_{\alpha}'$ implements $\alpha$ on $\pi(\A)$.

    It is clear that $3\Rightarrow 2$, we will show $1\Rightarrow 3$ below.
    We adopt all notations from Lemma \ref{lemma:purification}. By quasi-equivalence, there exists $f:\cM_{\psi}\to(\pi'(1_{\cB}\otimes \A))''$. By strong symmetry, there exists $U_{\alpha}\in\cM_{\psi}$ such that $U_{\alpha}|\psi\ra=|\psi\ra$. Note that $f(U_{\alpha})$ implements $\id_{\cB}\otimes \alpha$ on $\cB\otimes \A$ since $\pi'(\cB\otimes 1)\in(\pi'(1_{\cB}\otimes \A))'$. By assumption $\psi'|_{\A}=\psi$ and taking WOT-limit, we have
    \beq
    \la\psi'|f(x)|\psi'\ra= \la\psi|x|\psi\ra, \quad\forall\,x\in \cM_{\psi}
    \eeq
    Setting $x= U_{\alpha}$, we have $\la\psi'|f(U_{\alpha})|\psi'\ra=1$. Using Cauchy-Schwarz inequality, one has $f(U_{\alpha})|\psi'\ra= |\psi'\ra$, \ie $\psi'\circ(\id_{\cB}\otimes \alpha)=\psi'$. This proves $1\Rightarrow 3$.
\end{proof}

In Appendix.~\ref{sec:sym_cano_purif}, we show that any symmetry $\alpha$ of $\psi$ can always be extended as $\ovl{\alpha}\otimes \alpha$ on $\ovl{\A}\otimes\A$, which is always a symmetry of $\tilde{\psi}$. However, in the case that $\alpha$ is a QCA, the anomaly index of $\ovl{\alpha}\otimes\alpha$ always vanishes even if $\alpha$ is anomalous itself.

\begin{theorem}\label{thm:faithful_trivial_str_sym}
    If a state $\psi$ is faithful, i.e. $\psi(a^{*}a)>0$ for any $0\not= a\in\A$, then $\psi$ admits no non-trivial strong symmetry.
\end{theorem}
\begin{proof}
    Note that $\psi$ is faithful implies its GNS vector $|\psi\ra$ is \textbf{separating} for $\cM_{\psi}$, i.e.
    \beq
    ( x|\psi\ra=0,\, x\in\cM_{\psi}) \implies x=0
    \eeq
    To show this, note that $\A$ is SOT dense in $\cM_{\psi}$, we can restrict to $x\in\A$, in that case,
    \beq
   0= \| x|\psi\ra\|^{2}=\psi(x^*x)\implies x=0
    \eeq
    Thus $|\psi\ra$ is separating for $\cM_{\psi}$ if $\psi$ is faithful.

    If $\psi$ is strongly symmetric under $\alpha\in\Aut(\A)$, the unitary implementation $U_{\alpha}\in\cM_{\psi}$ and $U_{\alpha}|\psi\ra=|\psi\ra$. This forces $U_{\alpha}=1$ since $|\psi\ra$ is separating.

\end{proof}
As a corollary,
\begin{corollary}\label{coro:tracial_strong}
    The maximally mixed state (a.k.a the tracial state) $\tau$ does not admit any nontrivial strong symmetry. The same is true for finite-temperature states (a.k.a KMS states) \footnote{Specifically, because the quasi-local algebra for spin systems is simple (admits no non-trivial two-sided ideals), it follows that any KMS state is faithful.}.
\end{corollary}

Later, we will also be interested in a different symmetry condition that is intermediate between weak and strong symmetries, called the von Neumann symmetry (or vN symmetry for short). 

A reasonable symmetry condition is that the symmetry acts \textbf{separately} on the system (\ie $\cM_{\psi}$) and the environment (\ie $\cM_{\psi}'$). This leads to the following definition:

\begin{definition}[vN symmetries]\label{def:temp_vN_sym}
Let $\alpha$ be a symmetry of $\psi$, and let $U_\alpha$ denote the unitary implementation of $\alpha$ on $\cH_{\psi}$ satisfying $U_{\alpha}|\psi\ra=|\psi\ra$. We say a \textbf{linear} $\alpha$ is a von Neumann (vN) symmetry of $\psi$ if it factorizes as
\beq\label{eq:vN_sym_fac}
U_{\alpha}=u_{\alpha}\cdot v_{\alpha}
\eeq
 where $u_{\alpha}\in\cM_{\psi},v_{\alpha}\in\cM_{\psi}'$ are unitary operators.

\end{definition}

\begin{remark}\label{remark:vN=inner}
    The vN symmetry condition can also be formulated as the statement that $\alpha$ can be extended to an inner automorphism of $\cM_{\psi}$. This condition is known as weakly inner (since $\cM_{\psi}$ is a weak operator closure of $\A$) in mathematical literature, see e.g. Ref.~\cite{eilers2018c}.
\end{remark}

Let us first compare vN symmetries with strong symmetries as defined in Def.~\ref{def:strong_sym}. Obviously, any strong symmetry is automatically a vN symmetry with $v_{\alpha}=1$. Therefore, the notion of vN symmetry is a generalization of strong symmetries. We will discuss the Lieb-Schultz-Mattis type constraints for vN symmetries and its possible physical interpretation in later sections, see Sec.~\ref{sec:LSM}.

In fact, vN symmetries beyond strong symmetries are quite common. For example, we have the following:
\begin{lemma}\label{lemma:vN_sym_type_I}
    Let $\psi$ be a type $I$ factor state (see Example \ref{example:factors}), then every symmetry of $\psi$ is a vN symmetry.
\end{lemma}
This follows from the fact that all $*$-automorphisms are inner for $\cM_{\psi}$ if $\psi$ is a type $I$ factor (see Corollary 9.3.5 of Ref.~\cite{kadison1997fundamentals2}). In contrast, it is straightforward to construct type $I$ factors that admit no nontrivial strong symmetries.

In finite systems, every symmetry of a state is automatically a vN symmetry. Indeed, for any finite-dimensional $C^*$-algebra $\sC$, every automorphism is inner, i.e. of the form $\Ad_{u}$ for some unitary $u\in\sC$, which trivially extends to an inner automorphism on $\cM_{\psi}$ for any state $\psi$.

\begin{example}
Returning to the case of $\A$, recall from Corollary~\ref{coro:tracial_strong} that the tracial state $\tau$ has no strong symmetries. What about its vN symmetries?  
Note that $\tau$ always admits some “trivial’’ vN symmetries: any inner automorphism of $\A$ extends to an inner automorphism of $\cM_{\tau}$.  
Can we describe all vN symmetries of $\tau$?

This is indeed possible; precise criteria can be found in
Refs.~\cite{KISHIMOTO1996100,sun2014stronglyouterproducttype}.  
In brief, the following automorphisms are shown \textbf{NOT} to be vN symmetries of $\tau$:
\begin{enumerate}
    \item Lattice translations.
    \item Any nontrivial on-site finite group symmetry that commutes with translation symmetry.
\end{enumerate}

\end{example}

In fact, we will show the following result for lattice translation symmetry:
\begin{proposition}\label{prop:translation}
    Let $\psi$ be a factor state such that the lattice translation $T$ is a vN symmetry for $\psi$, then $\psi$ is a pure state.
\end{proposition}
\begin{proof}
    We note that the lattice translation $T$ is \textit{asymptotically abelian}, i.e.
    \beq\label{eq:asym_abelian}
    \lim_{n\to\infty}\|[T^{n}(a),b]\|=0,\quad\forall\,a,b\in\A
    \eeq
    It follows from Example 4.3.24 of \cite{bratteli2013operator1} that $\psi$ satisfies
    \beq
    \lim_{n\to\infty}|\psi(T^{n}(a)b)-\psi(a)\psi(b)|=0
    \eeq
    i.e. $\psi$ is \textit{weakly clustering}. By Theorem 4.(v) and Theorem 4a.(i) in Ref.~\cite{Doplicher1967}, this implies $(\pi_{\psi}(\A)\cup \{U_{T}\})'=\bbC\cdot \id_{\cH_{\psi}}$ where $U_{T}$ is the unitary implementation for $T$. By assumption $U_{T}=u_{T}\cdot v_{T}$ for $u_{T}\in \cM_{\psi}, v_{T}\in\cM_{\psi}'$, we then have
    \beq\label{eq:centralizer}
    \bbC\cdot\id_{\cH_{\psi}}=(\pi_{\psi}(\A)\cup \{U_{T}\})'=\cM_{\psi}'\cap Z(v_{T})
    \eeq
    where $Z(v_{T}):=\{x\in \B(\cH_{\psi}):[x,v_{T}]=0\}$ is the centralizer of $v_{T}$. We note that $v_{T}\in\cM_{\psi}'\cap Z(v_{T})=\bbC\cdot\id_{\cH_{\psi}}$, thus $v_{T}$ is a phase itself. Therefore Eq.~\eqref{eq:centralizer} implies $\cM_{\psi}'=\bbC\cdot\id_{\cH_{\psi}}$, i.e. $\psi$ must be pure according to Lemma \ref{lemma:Schur}.
\end{proof}
\begin{remark}
    The same proof applies to any asymptotically abelian symmetry actions (i.e. satisfying Eq.~\eqref{eq:asym_abelian}).
\end{remark}
As an easy corollary, we have:
\begin{corollary}
        If a factor state $\psi$ is strongly symmetric under lattice translation $T$, then $\psi$ is pure.
\end{corollary}
This corollary can also be seen from canonical purification, i.e. applying $\id\otimes \tau^{n}$ to $\tilde{\psi}$ for sufficiently large $n$ so one can show $\tilde{\psi}$ fails to be clustering. Note that all symmetries of a type $I$ factor is automatically a vN symmetry, from Prop.~\ref{prop:translation}, we deduce that
\begin{corollary}
    Let $\psi$ be a translationally invariant type I factor state, then $\psi$ is pure.
\end{corollary}

Below we illustrate the generality of vN symmetries with a simple example drawn from $(1+1)$-dimensional symmetry-protected topological (SPT) phases. Consider an on-site symmetry action 
\[
\alpha : G \longrightarrow \G^{\QCA}.
\]
By definition, the action is on-site if
\begin{equation}
\alpha_g = \prod_{i \in \mathbb{Z}} \Ad_{\rho(g)}, 
\qquad \forall\, g \in G,
\end{equation}
where $\rho $ is a (finite-dimensional) representation of $G$ acting on each local Hilbert space.

Let $\psi$ be a pure short-range entangled (SRE) state, 
meaning that it can be transformed into a pure product state 
by a finite-time evolution generated by (almost-)local Hamiltonians. 
Assume that $\psi$ is invariant under $\alpha_g$ for all $g \in G$. 
Then $\psi$ defines a $(1+1)$-dimensional symmetry protected topological (SPT) phase with symmetry $G$ \cite{Chen2010, Ogata2019split,Kapustin2020invertible}.

Let $\psi_R := \psi|_R$ denote the restriction of $\psi$ to the right half-chain, 
i.e.\ the reduced state obtained by tracing out the left half-chain. 
It follows immediately that $\psi_R$ is invariant under 
$\alpha_g^R$ for all $g \in G$, where $\alpha^R$ denotes the truncation 
of $\alpha$ to the right half-chain.

We now ask the following:

\begin{question}
Is $\psi_R$ strongly symmetric under $\alpha^R$? 
Moreover, does it satisfy the vN symmetry condition ?
\end{question}

We answer this question with the following theorem:
\begin{theorem}\label{thm:strongsym_SPT}
    The restriction $\psi_{R}$ always has $\alpha^{R}$ as a vN symmetry (as long as $\psi$ is a symmetric SRE state) while $\alpha^{R}$ can possibly be a strong symmetry only if $\psi$ is a trivial SPT.
\end{theorem}
We remark that the converse of the second assertion does not hold in general: 
the triviality of the SPT phase of $\psi$ does \emph{not} imply that 
$\alpha^{R}$ acts as a strong symmetry on $\psi_{R}$. 

To illustrate this, consider $G=\mathbb{Z}_{2}=\{0,1\}_{+}$, 
with symmetry action generated by the Pauli operator 
$\rho(0)=\id,\,\rho(1)=\sigma^{z}$. Since $\rH^{2}(\mathbb{Z}_{2};\U)=0$, 
all $\mathbb{Z}_{2}$ SPT phases are trivial in this case. 

As a concrete counterexample, let $\psi$ be the product state 
with spin-up at every site except at $i=-1,0$, 
where the two spins form a singlet. 
In this situation, the restricted state $\psi_{R}$ 
fails to be strongly symmetric under $\alpha^{R}$.

\begin{proof}[Proof to Theorem \ref{thm:strongsym_SPT}]
The first assertion follows from the fact that for any short-range 
entangled (SRE) state $\psi$, the restricted state $\psi_{R}$ 
is a type~I factor state (i.e.\ it \emph{splits}); 
see Lemma~4.2 of Ref.~\cite{Kapustin2020invertible}. 
By Lemma~\ref{lemma:vN_sym_type_I}, it then follows that 
$\psi_{R}$ admits a von Neumann (vN) symmetry.

For the second assertion, suppose that $\psi_{R}$ were strongly symmetric. 
By Theorem~\ref{thm:equi_def}, the unitary implementation $U_{g}$ for $\alpha_{g}^{R}$ in the GNS representation of $\psi_{R}$ satisfies $U_{g}|\psi_{R}\ra=|\psi_{R}\ra$ $U_{g}\in \cM_{R}$ (the von Neumann algebra generated by $\psi_{R}$). Since $\psi_{R}$ is type-I, this condition implies that each pure component of $\psi_{R}$ is again symmetric under $\alpha_{g}^{R}$. This implies $\psi$ has trivial SPT index (in the sense introduced by Ref.~\cite{Ogata2019split}).
\end{proof}

More generally, in Eq.~\eqref{eq:vN_sym_fac}, the unitaries 
$u_{\alpha}$ and $v_{\alpha}$ are defined only up to a phase 
whenever $\psi$ is a factor state. Consequently, even if the 
unitary implementation $\{U_{g}\}_{g\in G}$ forms a genuine 
(non-projective) representation of $G$, the corresponding 
unitaries $u_{g}\in \cM_{\psi}$ define a \emph{projective} 
representation of $G$.

The associated degree-2 group cohomology class of this projective 
representation coincides with the SPT index of $\psi$ 
introduced in Ref.~\cite{Ogata2019split} if $\psi$ splits.

\section{Charge coherence}\label{sec:cha_coh_ord_para}
In this section, we explain how our definition of strong symmetry is related to a ``charge coherence'' condition.

First, let us briefly introduce the definition of fidelity between states in the context of operator algebras \cite{Alberti1983}. 
\begin{definition}[Theorem 1 of Ref.~\cite{Alberti1983}]\label{def:fidelity}
    Let $f_{1},f_{2}$ be two positive linear functionals on $\A$, we define a subset of linear functions:
    \beq
    Q_{\A}(f_{1},f_{2}):=\{f:\A\to\bbC \And |f(a^{*}b)|^{2}\leqslant f_{1}(a^{*}a)f_{2}(b^{*}b),\quad\forall\,a,b\in\A\}
    \eeq
    The fidelity\footnote{In Ref.~\cite{Alberti1983}, the notation $P_{\A}(f_{1},f_{2})$ is used for fidelity} $F(f_{1},f_{2})$ is defined as
    \beq
    F(f_{1},f_{2}):=\max_{f\in Q_{\A}(f_{1},f_{2})}\|f(1)\|^{2}
    \eeq
    In particular, the maximum can be attained by some $f_{0}\in Q_{\A}(f_{1},f_{2})$.
\end{definition}
In order to see why it reduces to the usual fidelity, we have the following generalized version of Uhlmann's theorem for fidelity.
\begin{lemma}[Corollary 1 of Ref.~\cite{Alberti1983}]\label{lemma:Uhlmann}
    Given a representation $\pi:\A\to\B(\cH)$, if $f_{i}$ is represented by a vector $|f_{i}\ra\in\cH, i=1,2$ respectively, \ie $f_{i}(a)=\la f_{i}|\pi(a)|f_{i}\ra$ for all $a\in\A$, then
    \beq\label{eq:Uhlmann}
    F(f_{1},f_{2})=\sup_{R\in\pi(\A)'}|\la f_{1}|R|f_{2}\ra|^{2},\quad\|R\|=1
    \eeq
\end{lemma}
\begin{remark}
    It also follows from corollary 2 of Ref.~\cite{Alberti1983} that $F(f_{1},f_{2})=0$ if two states $f_{1},f_{2}$ are disjoint.
\end{remark}
Thus, Def.~\ref{def:fidelity} reduces to usual fidelity for finite dimensional $C^*$ algebras.

The following lemma tells us how the fidelity on the infinite system $\A$ can be approximated by finite subsystems. Write $\{\Gamma_{L}\}$ for increasing (\ie $\Gamma_{1}\subseteq\Gamma_{2}\dots$) and exhausting (\ie $\Lambda=\bigcup_{L}\Gamma_{L}$) collection of finite subsets of $\Lambda$, indexed by their “size” $L$. We have
\begin{lemma}[Eq.~(9) of Ref.~\cite{Alberti1983}]\label{lemma:finite_size_approx}
    The fidelity $F(f_{1},f_{2})$ satisfies
    \beq
    F(f_{1},f_{2})=\lim_{L\to\infty}F(f_{1}|_{\Gamma_{L}},f_{2}|_{\Gamma_{L}})
    \eeq
    where the fidelity on RHS is taken for the finite-dimensional algebra $\A_{\Gamma_{L}}$.
\end{lemma}

Next, we turn to the charge coherence condition and charged operators introduced in Def.~\ref{def:charge_coherence}. We briefly recall the relevant definition below.

Let $G$ be a finite-dimensional compact Lie group, let $\A$ be a $C^*$ algebra, and let $\alpha : G \to \mathrm{Aut}(A)$ be a continuous group homomorphism\footnote{The topology on $\Aut(\A)$ is specified by Eq.~\eqref{eq:strong_limit}. This is called strong topology in mathematical literature.}. Consider a state $\psi$ on $\A$ which is weakly symmetric under $\alpha_{g}$ for each $g \in G$. We say that $\psi$ is \emph{charge-coherent} with respect to $\alpha$ if
\begin{equation}
    F(\psi, \psi_{O}) = 0
\end{equation}
for any charged operator $O \in \A$ such that $\int_{g \in G} \alpha_g(O)\rd g = 0$.
Here $\psi_{O}$ is the (un-normalized) state defined by
\begin{equation}
    \psi_{O}(b) = \psi(O^{*} b \,O).
\end{equation}

\begin{lemma}
    If $\psi$ is charge-coherent with respect to $\alpha$, then it is also charge-coherent with respect to the restriction $\alpha_H$ for any subgroup $H \leq G$.
    \begin{proof}
    Just write $G$ as the disjoint union of left cosets of $H$.
    \end{proof}
\end{lemma}

\begin{theorem}\label{thm:strong_sym=cha_coh}
    A state $\psi$ is charge-coherent with respect to $\alpha$ if and only if it has strong symmetry $\alpha_{g}$ for each $g \in G$.
\end{theorem}
\begin{proof}
        First we prove that if $\psi$ has strong symmetry $\alpha_g$ for each $g \in G$, then $\psi$ is charge coherent. We work in the GNS representation $(\pi_{\psi},\cH_{\psi},|\psi\ra)$ of $\psi$. Consider any $R \in \cM_\psi'$ with $\|R\| = 1$, we construct the following sesquilinear form 
        \beq\label{eq:inv_pairing}
        B_{R}(|\phi_{1}\ra,\,|\phi_{2}\ra):=\la\phi_{1}|R|\phi_{2}\ra,\quad\forall\,|\phi_{1}\ra,|\phi_{2}\ra\in\cH_{\psi}
        \eeq
        Write $U_{g}$ for the unitary implementation of $\alpha_{g}$. Since $\alpha_{g}$ is a strong symmetry, we have $U_{g} \in \cM_\psi$, and hence
        \beq
        \begin{split}
            B_{R}(U_{g}|\phi_{1}\ra,U_{g}|\phi_{2}\ra)&=\la\phi_{1}|U_{g}^{\dagger}RU_{g}|\phi_{2}\ra\\
            &=\la\phi_{1}|R|\phi_{2}\ra\\
            &=B_{R}(|\phi_{1}\ra,\,|\phi_{2}\ra)
        \end{split}
        \eeq
        To derive the second line, we have used $U_{g}\in\cM_{\psi}$ while $R\in\cM_{\psi}'$.
        Thus, $B_{R}$ is a $G$-invariant pairing. For the choice $|\phi_{1}\ra=|\psi\ra,|\phi_{2}\ra=\pi_{\psi}(O)|\psi\ra=|\psi_{O}\ra$for any charged operator $O\in\A$, we have
        \beq
        \begin{split}
            B_{R}(|\psi\ra,|\psi_{O}\ra)&=B_{R}(U_{g}|\psi\ra,U_{g}|\psi_{O}\ra)\\
            &=B_{R}(|\psi\ra,\pi_{\psi}(\alpha_{g}(O))|\psi\ra)
        \end{split}
        \eeq
         Therefore, summing over $g\in G$ shows that if $\int_{g\in G} \alpha_g(O)\rd g= 0$, then $B_{R}(|\psi\ra,|\psi_{O}\ra) = 0$. From Eq.~\eqref{eq:Uhlmann}, it follows that $F(\psi, \psi_{O}) =\sup_{R\in\cM_{\psi}'} |B_{R}(|\psi\ra,|\psi_{O}\ra)|^{2}=0,\,\|R\|=1$.

        Next we prove the converse, \ie we will assume $\la\psi|R\pi_{\psi}(O)|\psi\ra=0$ for any $R\in\cM_{\psi}'$ and charged operator $O$, and our goal is to show $U_{g}\in\cM_{\psi}$. Note that this is equivalent to showing the $G$-invariance of $B_{R}$ defined by Eq.~\eqref{eq:inv_pairing}.

        Define the group average map $\mathrm{Avr}_{G}:\A\to \A^{G}$ as 
        \beq
        \mathrm{Avr}_{G}(a):=\int_{g\in G}\alpha_{g}(a)\mathrm{d}g,\quad a\in\A
        \eeq
        It is easily checked that $\mathrm{Avr}_{G}(\alpha_{g}(a))=\mathrm{Avr}_{G}(a)=\alpha_{g}(\mathrm{Avr}_{G}(a))$ for any $g\in G$.

        It is clear that $|\la\phi_{1}|R|\phi_{2}\ra|\leqslant \|R\|\cdot\|\,|\phi_{1}\ra\|\cdot\|\,|\phi_{2}\ra\|$ so $B_{R}$ is bounded. Thus it is determined by its value on the dense subspace $\pi_{\psi}(\A)|\psi\ra$.
         Taking $|\phi_{1}\ra=\pi_{\psi}(a_{1})|\psi\ra$ and $|\phi_{2}\ra:=\pi_{\psi}(a_{2})|\psi\ra$. Note that 
\beq\label{eq:avg_seslin_form}
\begin{split}
    B_{R}(|\phi_{1}\ra,\,|\phi_{2}\ra)&=\la\psi|\pi_{\psi}(a_{1}^{*})\cdot R\cdot\pi_{\psi}(a_{2})|\psi\ra\\
    &=\la\psi| R\cdot\pi_{\psi}(a_{1}^{*}a_{2})|\psi\ra\\
    &=\la\psi| R\cdot\pi_{\psi}(\mathrm{Avr}_{G}(a_{1}^{*}a_{2}))|\psi\ra
\end{split}
\eeq
where we have used $\la\psi|R\pi_{\psi}(O)|\psi\ra=0$ if $O$ is an charged operator (\ie $\mathrm{Avr}_{G}(O)=0$). Now, note that
\beq
\begin{split}
    B_{R}(U_{g}|\phi_{1}\ra,U_{g}|\phi_{2}\ra)&=\la\psi|\pi_{\psi}(a_{1}^{*})U_{g}^{\dagger}RU_{g}\pi_{\psi}(a_{2})|\psi\ra\\
    &=\la\psi|\pi_{\psi}(\alpha_{g}(a_{1}^{*}))R\cdot\pi_{\psi}(\alpha_{g}(a_{2}))|\psi\ra\\
    &=\la\psi|R\cdot\pi_{\psi}(\alpha_{g}(a_{1}^{*}))\pi_{\psi}(\alpha_{g}(a_{2}))|\psi\ra\\
    &=\la\psi|R\cdot\pi_{\psi}(\alpha_{g}(a_{1}^{*}a_{2}))|\psi\ra\\
    &=\la\psi|R\cdot\pi_{\psi}(\mathrm{Avr}_{G}(\alpha_{g}(a_{1}^{*}a_{2})))|\psi\ra\\
    &=\la\psi|R\cdot\pi_{\psi}(\mathrm{Avr}_{G}((a_{1}^{*}a_{2}))|\psi\ra\\
    &=B_{R}(|\phi_{1}\ra,|\phi_{2}\ra)
\end{split}
\eeq
In the second line, we have used $U_{g}\pi_{\psi}(a)U_{g}^{\dagger}=\pi_{\psi}\circ\alpha_{g}(a)$ for any $a\in\A$ and $U_{g}|\psi\ra=|\psi\ra$. In the third line, we have used $R\in\cM_{\psi}'$ so it commutes with $\pi_{\psi}(\A)$. Then we used Eq.~\eqref{eq:avg_seslin_form} and $\mathrm{Avr}_{G}(\alpha_{g}(a))=\mathrm{Avr}_{G}(a)$ for any $g\in G$ and $a\in\A$.

Therefore, we conclude that $B_{R}$ is $G$-invariant, this shows $U_{g}^{\dagger}R U_{g}=R$ for any $g\in G$. Since $R\in\cM_{\psi}'$ is arbitrary, this means $U_{g}\in \cM_{\psi}$.
\end{proof}

As a corollary, by lemma \ref{lemma:finite_size_approx}, we have
\begin{corollary}
    A state $\psi$ is strongly symmetric under $\alpha$ iff
    \beq
    \lim_{L\to\infty}F(\psi|_{\Gamma_{L}}, \psi_{O}|_{\Gamma_{L}})=0
    \eeq
    whenever $O$ is an charged operator.
\end{corollary}

\begin{remark}
    If $\psi$ is a pure state, $\cM_{\psi}'=\bbC\cdot\id_{\cH_{\psi}}$, Eq.~\eqref{eq:Uhlmann} shows
    \beq
    F(\psi,\psi_{O})=|\psi(O)|^{2}
    \eeq
    In this case, charge coherence reduces to the standard diagnostic of whether a state is (weakly) symmetric, via the expectation values of charged operators. This reflects the fact that for pure states, strong and weak symmetries are equivalent.
\end{remark}

\section{Split property, separability and mutual information}\label{sec:split}
In this section, we introduce certain notions to characterize the entanglement behavior of a state. We specialize to 1d quantum spin chains in this section. \textit{We stress that we work with 1d spin chains in this and the following section.}
\begin{definition}[Split property \cite{matsui2011boundedness}]
    A factor state $\psi$ is said to split if $\psi\sim \psi_{L}\otimes \psi_{R}$, where $\psi_{L}$ and $\psi_{R}$ are the restriction of $\psi$ to left and right half chain respectively.
\end{definition}
Note that if $\psi$ is a factor state, $\psi_{L}$ and $\psi_{R}$ are automatically factor states.
The following lemma is useful for characterizing the split property.
\begin{lemma}\label{lemma:split}
    Let $\psi$ be a factor state, then it splits if and only if its canonical purification $\tilde{\psi}$ splits.
\end{lemma}
\begin{proof}
    If $\psi$ splits, then $\psi\sim \psi_{L}\otimes \psi_{R}=:\phi$. By lemma \ref{lemma:purify}, $\tilde{\psi}\simeq\tilde{\phi}$.  By the uniqueness of canonical purification (see Lemma.~\ref{lemma:purify}), we must have $\tilde{\phi}= \phi_{L}\otimes\phi_{R}$ where $\phi_{L}$ (resp. $\phi_{R}$) is the purification of $\psi_{L}$ (resp. $\psi_{R}$). So we have unitary equivalence $\tilde{\psi}\simeq\phi_{L}\otimes\phi_{R}$, \ie $\tilde{\psi}$ splits.

    The converse can be proved similarly using lemma \ref{lemma:purify}, so we will suppress it here.
\end{proof}
We noticed that a similar discussion has already appeared in \cite{sopenko2025reflectionpositivityrefinedindex}.

We provide another version of the split property.
\begin{definition}[Countably 2-separable]\label{def:countable_2_sep}
    A factor state $\psi$ is said to be countably 2-separable if it can be written as
    \beq
    \psi = \sum_{k=1}^{\infty}\lambda_{k}\psi_{L}^{(k)}\otimes \psi_{R}^{(k)}
    \eeq
    where $\lambda_{k}\searrow0$ adds up to 1, $\psi_{L}^{(k)}, \psi_{R}^{(k)}$ are factor states defined on the left and right half chains, respectively.
\end{definition}
For a pure state $\psi$, this condition is equivalent to the Schmidt decomposition between the left and right half-chains.  
Such a decomposition requires the GNS Hilbert space to factorize as  
\[
\cH_{\psi} \simeq \cH_{L}\otimes \cH_{R}
\]
As shown in Ref.~\cite{matsui2011boundedness}, this factorization holds precisely when a pure state $\psi$ satisfies the split property or, equivalently, when the half-chain restriction of $\psi$ generates a type~$I$ factor.

We prove a mixed analog of this observation below.
\begin{lemma}
    A factor state $\psi$ splits iff it is countably 2-separable.
\end{lemma}
\begin{proof}
    If $\psi$ splits, its purification $\tilde{\psi}:\ovl{\A}\otimes\A\to\bbC$ also splits. Take the Schmidt decomposition of $\tilde{\psi}$ and restrict it to $\A$; we see that $\psi$ is countably 2-separable.

    Conversely, if $\psi$ is countably 2-separable, due to Lemma \ref{lemma:factor_quasieq}, we have $\psi\sim \psi_{L}^{(k)}\otimes \psi_{R}^{(k)}$ for any $k$ so it splits as well. 
\end{proof}

Now, we relate the split property to correlations and mutual information.
\begin{proposition}\label{prop:MI_split}
    If a factor state $\psi$ satisfies the area law of mutual information,
    \beq
    I(\Gamma;\Gamma^{c}):= S(\psi\|(\psi_{\Gamma}\otimes \psi_{\Gamma^{c}})) < C_{\psi}
    \eeq
    where $C_{\psi}$ is a constant independent of $\Gamma$ and $S(\cdot\|\cdot)$ is the relative entropy,then $\psi$ splits.
\end{proposition}

\begin{proof}
    The following proof is based on the weak-* lower-semicontinuity of realtive entropies (see e.g. Prop.~5.23 of Ref.~\cite{ohya2004quantum}), i.e., if $\lim_{n\to\infty}\varphi_{n}=\varphi$ and $\lim_{n\to\infty}\phi_{n}=\phi$ be two weak-* convergent subsequences of states, then
    \beq
    S(\varphi ||\phi)\leqslant \lim_{k\to\infty}\inf_{n\geqslant k}S(\varphi_{n}||\phi_{n})< C_{\psi}
    \eeq
    Consider the sequence $\varphi_{n}=\psi$ and $\phi_{n}=\psi_{[0,n]}\otimes \psi_{[0,n]^{c}}$, note $\lim_{n\to\infty}\phi_{n}=\psi_{L}\otimes \psi_{R}$, so we obtain $S(\psi||\psi_{L}\otimes \psi_{R})<C_{\psi}<\infty$. Now use the fact that two factor states are either disjoint or quasi-equivalent, and $S(\varphi||\phi)=\infty$ if two states are disjoint, we conclude $\psi\sim\psi_{L}\otimes \psi_{R}$.
\end{proof}

\section{Lieb-Schultz-Mattis (LSM) constraints in 1d}\label{sec:LSM}
In this section, we specialize to one-dimensional quantum spin chains. 
We begin with the LSM constraints for strong symmetries, which are technically simpler, and then extend the argument to the more subtle case of vN symmetries. In the main text, we restricted attention to symmetry actions implemented by QCAs. Here, however, we allow the more general class of symmetry actions implemented by \textit{locality-preserving automorphisms} $\G^{lp}$

see Def.~\ref{def:LPA} for the precise definition.

Let us first recall the pure version of the LSM constraint in Ref.~\cite{kapustin2024anomalous,rubio2024classifyingsymmetricsymmetrybrokenspin}.
\begin{lemma}\label{lemma:KS}
        Let $\alpha:G\to \G^{lp}$ be a symmetry action and $\psi$ be a pure state; then the followings are incompatible:
    \begin{enumerate}
        \item $\psi$ is symmetric under $\alpha$.
        \item $\psi$ splits.
        \item $\alpha$ has a non-vanishing anomaly index.
    \end{enumerate}
\end{lemma}

It is straightforward to generalize it to strong symmetries.
\begin{theorem}[LSM for strong symmetries]\label{thm:LSM_strong}
    Let $\alpha:G\to \G^{lp}$ be a symmetry action and $\psi$ be a possibly mixed state. then the followings are incompatible:
    \begin{enumerate}
        \item $\psi$ is strongly symmetric under $\alpha$.
        \item $\psi$ is clustering.
        \item $\psi$ splits.
        \item $\alpha$ has a non-vanishing anomaly index.
    \end{enumerate}
\end{theorem}
\begin{proof}
    Assuming the contrary, since $\psi$ is a factor state with split property, Lemma \ref{lemma:split} shows its canonical purificaion $\tilde{\psi}$ again splits. By strong symmetry condition, $\tilde{\psi}\circ(\id_{\oA}\otimes\alpha)=\tilde{\psi}$. Since $(\id_{\oA}\otimes\alpha)$ has the same anomaly index as $\alpha$, applying lemma \ref{lemma:KS} to $\tilde{\psi}$ leads us to a contradiction.
\end{proof}
Similarly, we can deduce a generalization of the strong-weak mixed version of LSM, reported in Ref.~\cite{wang2024anomalyopenquantumsystems}.
\begin{corollary}\label{coro:strong_weak_mixed}
Let $\psi$ be a state and $\alpha:G\times H\!\to\!\G^{lp}$ an action such that 
$\alpha|_{G}$ is weakly symmetric for $\psi$, while $\alpha|_{H}$ is strongly symmetric for $\psi$.  
Let $\omega\in\rH^{3}(G\times H;\U)$ denote the anomaly of $\alpha$, and assume that its class is nontrivial in the quotient
\[
\rH^{3}(G\times H;\U)\big/ \rH^{3}(G;\U).
\]
Then $\psi$ cannot satisfy both the split property and clustering simultaneously.

\end{corollary}
\begin{proof}
        Consider $\alpha:G\times H\to \G^{lp}$, where $\alpha|_{G}$ is weak while $\alpha|_{H}$ is strong. The symmetry action $\alpha_{(g,h)}$ can be purified into $\ovl{\alpha}_{(g,1)}\otimes\alpha_{(g,h)}$. A straightforward calculation shows that the anomaly index of this symmetry action lies in the following quotient:
    \beq
    \rH^{3}(G\times H;\U)/\rH^{3}(G;\U)
    \eeq
    By applying lemma \ref{lemma:KS} to this case, we have thus derived the desired conclusion.
\end{proof}

Below we prove the LSM type constraints for anomalous vN symmetries.
\begin{theorem}\label{thm:vN_LSM}
    Let $\psi$ be a state and $\alpha:G\to\G^{lp}$ be a symmetry action by LPA, Then the followings are incompatible:
    \begin{enumerate}
        \item The state $\psi$ has $\alpha_{g}$ as its von Neumann symmetry for all $g\in G$.
        \item $\psi$ has split property.
        \item $\psi$ is clustering.
        \item $\alpha$ is anomalous.
    \end{enumerate}
\end{theorem}

We need the following lemma:
\begin{lemma}[Thm.~13.1.16 of Ref.~\cite{kadison1997fundamentals2}]\label{Lemma:inner_aut}
    Let $\alpha_{i},i=1,2$ be automorphisms of the von Neumann algebra $\cM_{i},i=1,2$, then $\alpha_{1}\otimes \alpha_{2}$ on $\cM_{1}\otimes \cM_{2}$ is inner iff each $\alpha_{i}$ is inner.
\end{lemma}
As a corollary, we have
\begin{corollary}\label{coro:vN_sym_splitting}
     Let $\cM_{i},i=1,2$ be two von Neumann algebras, we write $\cM:=\cM_{1}\otimes\cM_{2}$ for their tensor product. Consider $\alpha_{i}\in\Aut(\cM_{i}),i=1,2$ then, $\alpha:=\alpha_{1}\otimes\alpha_{2}\in\Aut(\cM)$ is a vN symmetry iff each $\alpha_{i}$ is a vN symmetry.
\end{corollary}

\begin{proof}[Proof to Theorem \ref{thm:vN_LSM}]
    Since one can cancel the GNVW index of $\alpha$ by stacking with ancillas \cite{kapustin2024anomalous}, without loss of generality, we assume $\alpha_{g}$ has vanishing GNVW index for any $g\in G$. Consider the restriction of $\alpha$ on right-half chain $\alpha^{R}$ (similarly one can define $\alpha^{L}$).

    The $C^{*}$-algebra $\A$ always factorizes as $\A\simeq\A_{L}\otimes\A_{R}$. In addition,
    if $\psi$ splits (\ie $\psi\sim\psi_{L}\otimes\psi_{R}$), we have $\cH_{\psi}\simeq\cH_{L}\otimes\cH_{R}$, where $(\pi_{i},\cH_{i},|\psi_{i}\ra)$ is a GNS triple for $\psi_{i}$, $i=L,R$. Then we have tensor factorization of von Neumann algebras $\cM_{\psi}\simeq\cM_{L}\otimes \cM_{R}$, where $\cM_{i}:=(\pi_{i}(\A_{i}))''$ for $i=L,R$, we will also write $\cM_{i}'$ for their commutant in $\B(\cH_{i}),i=L,R$.
    
    Note that by Corollary \ref{coro:vN_sym_splitting} $\alpha$ is vN in $\cM_{\psi}$ implies $\alpha^{R}$ is vN in $\cM_{R}$. This means $\alpha^{R}$ is an inner automorphism, so we can represent $\alpha^{R}_{g}$ as $\Ad_{U_{g}}$ for some $U_{g}\in\cM_{R}$.

 By Eq.~\eqref{eq:composition_half_chain},
    \beq
    \alpha_{g}^{R}\alpha_{h}^{R}=\Ad_{V_{g,h}}\alpha_{gh}^{R}
    \eeq
    As a result, on $\cH_{R}$ we have
    \beq
    U_{g}U_{h}\pi_{\psi}(a)U_{h}^{-1}U_{g}^{-1} = \pi_{\psi}(V_{g,h})U_{gh}\pi_{\psi}(a)U_{gh}^{-1}\pi_{\psi}(V_{g,h})^{-1},\quad\forall a\in\A_{R}
    \eeq
    This implies:
    \beq\label{eq:eta}
    \eta_{g,h}:=U_{h}^{-1}U_{g}^{-1}\pi_{\psi}(V_{g,h})U_{gh}\in \pi_{R}(\A_{R})'=\cM_{R}'
    \eeq
    On the other hand, we note $\eta_{g,h}\in\cM_{R}$ since $U_{g}\in\cM_{R},\forall\, g\in G$, we have
    \beq\label{eq:trivializing_2_cochain}
    \eta_{g,h}\in \cM_{R}\cap\cM_{R}'=\bbC\cdot\id_{\cH_{R}}
    \eeq
    where we have used that $\psi$ is a factor (so is $\psi_{R}$) and $\eta_{g,h}\in\U$ due to unitarity.
    Using the definition of the anomaly index:
    \beq\label{eq:anomaly_index_2}
    \omega_{g,h,k}=V_{g,h}V_{gh,k}V_{g,hk}^{-1}\alpha_{g}^{R}(V_{h,k})^{-1}
    \eeq
    Using Eq.~\eqref{eq:trivializing_2_cochain} to replace $V_{g,h}$ by $U_{g},U_{h},U_{gh}$ and $\eta_{g,h}$, one finds,
    \beq\label{eq:trivializing_anomaly}
    \omega_{g,h,k}=(\delta\eta)_{g,h,k}
    \eeq

\end{proof}

Let us give another perspective from lifting theory below; this approach will be useful in establishing further generalization of Theorem \ref{thm:vN_LSM} as in Corollary~\ref{coro:mixed_vN}.

For any linear (unnecessarily vN!) symmetry action $\alpha:G\to \G^{lp}$, as before, we cancel GNVW indices and one always has unitary $U_{g}$ on $\cH_{R}$ such that $\Ad_{U_{g}}$ implements $\alpha_{g}^{R}$. Let us claim that $\Ad_{U_{g}}$ gives a group homomorphism $G\to \mathrm{Out(\cM_{R}})$. To this end, we note
\beq
\begin{split}
    \Ad_{U_{g}}\circ\Ad_{U_{h}}(\cM_{R})&= \Ad_{\pi_{\psi}(V_{g,h})}\circ\Ad_{U_{gh}}\circ\Ad_{\eta_{g,h}}(\cM_{R})
\end{split}
\eeq
Note that $\pi_{\psi}(V_{g,h})\in\cM_{R}$ is an inner-automorphism of $\cM_{R}$ and $\eta_{g,h}\in\cM_{R}'$ (hence its conjugation is trivial on $\cM_{R}$). Therefore, $g\to \Ad_{U_{g}}$ gives a well-defined map from $G$ to $\mathrm{Out}(\cM_{R})$. We are interested in the following lifting problem:
\[\begin{tikzcd}
	& {\mathrm{Aut}(\mathscr{M}_{R})} \\
	G & {\mathrm{Out}(\mathscr{M}_{R})}
	\arrow[from=1-2, to=2-2]
	\arrow[dashed, from=2-1, to=1-2]
	\arrow[from=2-1, to=2-2]
\end{tikzcd}\]

\begin{lemma}\label{lemma:lifting}
    If $\psi$ is a factor state and there exists a lifting $G\to\Aut(\cM_{R})$, then the anomaly index $\omega\in\rH^{3}(G;\U)$ of $\alpha$ must vanish.
\end{lemma}

We note that closely related lifting problems have already been studied in
Refs.~\cite{Jones_2021, evington2022anomaloussymmetriesclassifiablecalgebras}.
A classical result states that, if such a lifting exists, then the associated
obstruction class in $\rH^{3}_{\sigma}(G; Z(\cU_{R}))$ must vanish, where $\cU_{R}$ is the group of unitary operators in $\cM_{R}$, $Z(\cU_{R})$ denotes its centralizer, \(\sigma\) denotes the induced \(G\)-action on $Z(\cU_{R}) := \cU_{R} \cap \cM_{R}'$.
See Theorem 7.1.2 of \cite{Azcárraga_Izquierdo_1995_cohomology}.

The new content of our lemma is that, when $\psi$ is a factor state (i.e. $Z(\cU_{R})=\U$), the above obstruction can be identified with the previously defined anomaly index $\omega\in\rH^{3}(G;\U)$; see Eq.~\eqref{eq:anomaly_index}. \emph{A priori}, these two obstruction classes need not coincide.

\begin{proof}[Proof to Lemma \ref{lemma:lifting}]
    Let $\gamma_{g}\in \Aut(\cM_{R})$ be a lifting for $[\Ad_{U_{g}}]$, so $\gamma_{g}\gamma_{h}=\gamma_{gh}$ and $\gamma_{g}=\Ad_{U_{g}}\circ\Ad_{x_{g}}$ for some unitary $x_{g}\in \cM_{R}$. Since $\gamma:G\to\Aut(\cM_{R})$ is a group homomorphism, we obtain
    \beq 
    U_{g}U_{h}U_{gh}^{-1}=\mu_{g,h}\cdot x_{gh}x_{h}^{-1}\Ad_{U_{h}^{-1}}(x_{g})^{-1}\in\cM_{R}
    \eeq
    for some $\mu_{g,h}\in Z(\cM_{R})=\U$ (where we have used $\cM_{R}$ is a factor). Thus, for $\eta_{g,h}$ defined in Eq.~\eqref{eq:eta}:
    \beq
    \eta_{g,h}=U_{h}^{-1}U_{g}^{-1}U_{gh}\pi_{R}((\alpha_{gh}^{R})^{-1}V_{g,h})\in\cM_{R}\cap\cM_{R}'
    \eeq
    Therefore, after substituting $\eta_{g,h}$ and $U_{g}$'s for $V_{g,h}$ according to Eq.~\eqref{eq:eta}, we again recover Eq.~\eqref{eq:trivializing_anomaly}.
\end{proof}
The merit of this reformulation is to simplify the proof to the following corollary, which generalizes Corollary \ref{coro:strong_weak_mixed}.

\begin{corollary}\label{coro:mixed_vN}
    If $\alpha:G\times H\to \G^{lp}$ where $\alpha|_{H}$ is a vN symmetry while $\alpha|_{G}$ is unnecessarily vN, then the result of Corollary \ref{coro:strong_weak_mixed} still holds if the anomaly index $\omega\in\rH^{3}(G\times H;\U)$ survives under
    \beq\label{eq:anomaly_quotient}
    \rH^{3}(G\times H;\U)\to \rH^{3}(G\times H;\U)/\rH^{3}(G;\U)
    \eeq
\end{corollary}
\begin{proof}
    Let $\psi$ be a split factor state which is symmetric under $G\times H$. By assumption, $H$ is a vN symmetry for $\psi$. Then since $\alpha|_{H}$ is vN, there is no obstruction to construct a lifting $H\to\Aut(\cM_{R})$. Thus, the only non-vanishing obstruction for lifting $G\times H$ (i.e. the anomaly index) lives in $\rH^{3}(G;\U)$. Contradicting the fact Eq.~\eqref{eq:anomaly_quotient}.
\end{proof}
This corollary is useful since Prop.~\ref{prop:translation} shows translation can never be a vN symmetry unless the factor state $\psi$ is actually pure. In application, $G$ is often taken to be the lattice translation symmetry while $H$ is some internal symmetry that is vN.
\begin{remark}
    A similar result holds if one considers $\tilde{H}$ a nontrivial group extension\footnote{This means they fit into the following short exact sequence $1\hookrightarrow H\to \tilde{H}\twoheadrightarrow G\to 1$ and we do not assume this is a central extension.} of $G$ by $H$ rather than $G\times H$: let $q:\tilde{H}\to G$ be the quotient map. The conclusion of Corollary \ref{coro:mixed_vN} follows if one replaces Eq.~\eqref{eq:anomaly_quotient} by $\omega\not\in \im(q^{*})$ where $q^{*}$ is the pullback
    \beq
    q^{*}:\rH^{3}(G;\U)\to \rH^{3}(\tilde{H};\U)
    \eeq
\end{remark}

\section{Bath evolution and locally generated channels}\label{sec:bath_evo}

In quantum mechanics, a general dynamical process is not necessarily given by a unitary evolution, but more generally by a quantum channel \cite{QChannelLecture.pdf}; see that reference for the precise definition. The Stinespring theorem states that every quantum channel on a finite system can be realized as a unitary evolution after enlarging the system by ancillary degrees of freedom.

In infinite systems, the analog of quantum channels are unital, completely positive (UCP) maps on the quasi-local algebra \cite{QChannelLecture.pdf}. However, general UCP maps do not necessarily have a clear physical interpretation as physically achievable dynamical processes, and will be too general for our purposes. Instead we will consider a subset of UCP maps, which we refer to as ``bath evolutions''.

\begin{definition}\label{def:bath_evo}
A map $\E:\A\to\A$ (not necessarily a $*$-homomorphism) is said to be a \textbf{bath evolution} if there exists another $C^*$-algebra $\cB$ (not necessarily a spin system) such that\footnote{As noted before, $\A$ is a UHF algebra and therefore nuclear, so the tensor product $\A\otimes\cB$ is well defined even when $\cB$ is a general $C^*$ algebra \cite{murphy2014c}.}
\beq
\E=(\id_{\A}\otimes \Phi_{\cB})\circ \beta \circ \iota,
\eeq
where $\iota:\A\to\A\otimes\cB$ is the inclusion $a\mapsto a\otimes 1_{\cB}$, $\beta\in\Aut(\A\otimes\cB)$, and $\Phi_{\cB}$ is a state on $\cB$. We call $(\cB,\beta)$ (or simply $\beta$) a purification of $\E$ (on $\A\otimes \cB$).
\end{definition}
Physically, a bath evolution describes a process in which the system is coupled to a ``bath'', initially in the state $\Phi_\cB$, then the two systems evolve unitarily together, and finally the bath is ``traced out''. Observe that the Stinespring dilation theorem implies that every UCP map on a finite system is a bath evolution; however the same result is not expected to hold in infinite systems.

The following proposition shows that the composition of bath evolutions is again a bath evolution.
\begin{proposition}\label{prop:compose_bath_evo}
    Let $\E_{1},\E_{2}$ be two bath evolutions on $\A$, then $\E_{1}\circ\E_{2}$ is again a bath evolution.
\end{proposition}
\begin{proof}
    By Def.~\ref{def:bath_evo}, for $i=1,2$ there exists $\cB_{i}$ such that
    \beq
    \E_{i}=(\id_{\A}\otimes \Phi_{\cB_{i}})\circ \beta_{i}\circ\iota_{i}
    \eeq
    for some $\beta_{i}\in\Aut(\A\otimes\cB_{i})$. Consider $\A\otimes \cB_{1}\otimes\cB_{2}$, one can easily check that 
    \beq
    \E_{1}=(\id_{\A}\otimes \Phi_{\cB_{1}}\otimes \id_{\cB_{2}})\circ \tilde{\beta}_{1}\circ\iota
    \eeq
    where $\tilde{\beta}_{1}\in\Aut(\A\otimes \cB_{1}\otimes \cB_{2})$ is obtained by tensoring $\id_{\cB_{2}}$ with $\beta_{1}$. One can similarly represent $\E_{2}$ on $\A\otimes \cB_{1}\otimes\cB_{2}$.
    
    It follows that
    \beq
    \begin{split}
        \E_{1}\circ\E_{2}&=(\id_{\A}\otimes \Phi_{\cB_{1}}\otimes \id_{\cB_{2}})\circ \tilde{\beta_{1}}\circ\iota_{1}\circ(\id_{\A}\otimes \id_{\cB_{1}}\otimes \Phi_{\cB_{2}})\circ\tilde{\beta_{2}}\circ\iota_{2}
        \\&=(\id_{\A}\otimes\Phi_{\cB_{1}}\otimes\Phi_{\cB_{2}})\circ\tilde{\beta_{1}}\circ\tilde{\beta_{2}}\circ\iota
    \end{split}
    \eeq
   where we have used that $\tilde{\beta}_{1}$ commutes with $\id_{\A}\otimes \id_{\cB_{1}}\otimes\Phi_{\cB_{2}}$.
\end{proof}
Next, we discuss the symmetry property, especially the strong symmetry of a bath evolution.

\begin{definition}
    Let $\E$ be a bath evolution on $\A$ and $\alpha\in\Aut(\A)$, then $\E$ is (weakly) symmetric under $\alpha$ if it commutes with $\alpha$. Furthermore, it is said to be strongly symmetric under $\alpha$ if there \textbf{exists} a purification $(\cB,\beta)$ for $\E$ such that
    \beq
    \beta\circ(\alpha\otimes \id_{\cB})=(\alpha\otimes \id_{\cB})\circ\beta
    \eeq
\end{definition}
Physically, a strongly symmetric bath evolution describes a process in which the system does not exchange charge with the bath.

The following proposition describes how the symmetry properties of a state transform under a strongly symmetric bath evolution.

\begin{proposition}\label{prop:weak_to_weak}
    Let $\psi$ be a state which is not strongly symmetric under $\alpha \in \Aut(\A)$, and suppose that the bath evolution $\E$ is strongly symmetric under $\alpha$, then $\psi\circ\E$ is also not strongly symmetric under $\alpha$. 
\end{proposition}
\begin{proof}
    Since $\psi$ is not strongly symmetric, there exists an extension $\Psi$ on $\A\otimes \cB$ such that:
    \beq\label{eq:asymmetry}
    \Psi\circ(\alpha\otimes \id_{\cB})\not= \Psi
    \eeq
    We fix this choice for $\cB$ and $\Psi$ from now on.

     We then note $\Psi\circ(\E\otimes \id_{\cB}) $ is an extension of $\psi\circ\E$ on $\A\otimes\cB$. Below we show $\Psi\circ(\E\otimes \id_{\cB})$ fails to be symmetric under $\alpha\otimes \id_{\cB}$ and hence $\psi\circ\E$ fails to be strongly symmetric under $\alpha$.
     
    Since $\E$ is strongly symmetric under $\alpha$, there exists a purification $(\sC,\beta)$ for $\E\otimes \id_{\cB}$ such that $\beta$ commutes with $\alpha\otimes\id_{\cB}\otimes \id_{\sC}$. Then we have
    \beq
    \Psi\circ(\E\otimes \id_{\cB})=((\Psi\otimes \Phi_{\sC})\circ\beta)|_{\A\otimes\cB\otimes 1_{\sC}}
    \eeq
    Thus $(\Psi\otimes \Phi_{\sC})\circ\beta$ is an extension of $\Psi\circ(\E\otimes \id_{\cB})$ (and hence an extension of $\psi\circ\E$). However, since $\beta$ commutes with $\alpha\otimes\id_{\cB}\otimes\id_{\sC}$ and is invertible, $(\Psi\otimes \Phi_{\sC})\circ\beta$ is invariant under $\alpha\otimes\id_{\cB}\otimes\id_{\sC}$ if and only if $\Psi\otimes \Phi_{\sC}$ is invariant, which contradicts Eq.~\eqref{eq:asymmetry}.
\end{proof}

In other words, a strongly symmetric bath evolution cannot take a state without strong symmetry to one with strong symmetry.

However, it is worth noting that the reverse process \emph{can} occur. That is, applying a strongly symmetric bath evolution $\E$ to a strongly symmetric state $\psi$, the resulting state $\psi\circ\E$ must be weakly symmetric but can fail to be strongly symmetric.
\begin{figure}[h!]
\begin{tikzpicture}[scale=1.25]
  \def\n{6}      
  \def\dx{0.6}   
  \def\amp{0.3}  
  \def\rectheight{0.2}

  \fill[blue!25] (0,{\amp+\rectheight}) rectangle ({(2*\n+1)*\dx},{\amp-\rectheight});

  \draw (0,0) node[left, yshift=0.17ex] {$\cdots$};

    \draw ({(2*\n+1)*\dx}, 0) node[right, yshift=0.17ex] {$\cdots$};

  \fill[red!25] (0,{-\amp+\rectheight}) rectangle ({(2*\n+1)*\dx},{-\amp-\rectheight});

  \draw ({(2*\n+1)*\dx+0.5},\amp) node[right] {\textcolor{blue}{System}};

    \draw ({(2*\n+1)*\dx+0.5},-\amp) node[right] {\textcolor{red}{Bath}};
  
  \draw[thick]
    (0,0)
    -- ({\dx/2}, {\amp})
    \foreach \i in {1,...,\n} {
      -- ({(2*\i-1/2)*\dx}, {-\amp})
      -- ({(2*\i+1/2)*\dx}, {\amp})
    }
    -- ({(2*\n+1)*\dx}, 0);

    \draw[draw opacity=0] ({6.5*\dx},{\amp}) -- ({7.5*\dx}, {-\amp}) node[midway, above, sloped] {\footnotesize CZ};

    \foreach \i in {1,...,\n} {
        \fill ({(2*\i-1/2)*\dx}, {-\amp}) circle(0.1cm);
        \fill ({(2*\i+1/2)*\dx}, {\amp}) circle(0.1cm);
    }
    \fill ({0.5*\dx}, {\amp}) circle(0.1cm);
\end{tikzpicture}
\caption{\label{fig:bath_evolution}The graphical construction of a bath evolution that sends a state with strong symmetry to one with only weak symmetry, see the text.}
\end{figure}

As an example of this phenomenon, consider the case of a 1-D infinite chain of spin-1/2's, which initially in the state where each spin is in the +1 eigenstate of the Pauli $\sigma^x$. Now we consider a bath evolution where the ``bath'' $\cB$ is also a 1-D infinite chain of spin-1/2's, offset by half a lattice spacing, as shown in Figure \ref{fig:bath_evolution}, and is also initially taken to be in the state where each eigenstate is in the +1 eigenstate of $\sigma^x$. The system-bath coupling involves applying a controlled-Z gate between each system spin and its two nearest neighbors in the bath, as also shown in Figure \ref{fig:bath_evolution}. The combined state of the system and bath is now a 1-D cluster state, and tracing out the bath gives (with respect to local operators) the maximally mixed state on the system. The bath evolution and the initial state of the system are both strongly symmetric under the $\mathbb{Z}_2$ symmetry generated by $\prod_i \sigma_x^i$, but the bath evolution takes this initial state to one with only weak symmetry. This can be viewed as a special case of the strongly symmetric dephasing channel discussed in Ref.~\cite{Lessa_2025}.

We remark that the 1-D cluster state shared between system and bath serves an example of a purification of the maximally mixed state which is different from the canonical one. In fact, a bath evolution which sends a strongly symmetric state to a state without the strong symmetry necessarily requires that any purification of the state of the combined system-bath just before tracing out the bath be a \emph{non}-canonical purification of the final state of the system (with the bath traced out). In \emph{finite} systems, all purifications are unitarily equivalent to the canonical purification, so it follows that in finite systems, unlike infinite systems, strongly symmetric channels always preserve strong symmetry.

As an easy but important corollary of Proposition \ref{prop:weak_to_weak}:
\begin{corollary}\label{coro:irreversible}
    Let $\alpha\in \Aut(\A)$, if two states $\psi_{1},\psi_{2}$ are two-way connected by strongly $\alpha$-symmetric bath evolutions, that is, 
    \beq
    \begin{split}
        \psi_{1}&=\psi_{2}\circ \E\\
        \psi_{2}&=\psi_{1}\circ\E'
    \end{split}
    \eeq
    Then, $\psi_{1}$ is strongly $\alpha$-symmetric if and only if $\psi_{2}$ is strongly $\alpha$-symmetric.
\end{corollary}

For many practical purposes, such as the definition of mixed-state phases, bath evolutions are still too general. To obtain better control over locality, we introduce the notion of \textit{locally generated channels}. To this end, we first recall the concept of locally generated automorphisms (LGAs), introduced in Refs.~\cite{Kapustin2020invertible}, which are used in the definition of gapped phases for pure states.
\begin{definition}
    Let $\mathbb{B}_{d}:=\{\prod_{i=1}^{d}[m_{i},m_{i}']\subseteq\R^{d}:m_{i},m_{i}'\in \z,m_{i}\leqslant m_{i}'\}$, an element in $\mathbb{B}_{d}$ is called a brick in $\R^{d}$. A Hamiltonian $H=\sum_{Y\in \mathbb{B}_{d}}h^{Y}$ (viewed as a derivation) \textbf{on a spin system} is \textbf{almost-local} if:
    \begin{enumerate}
        \item     \beq
    \tau(a^{*}h^{Y})=0
    \eeq
    where $\tau$ is the tracial state and $a\in\A_{Z}$ for any $Z\in\mathbb{B}_{d},Z\subsetneq Y$.
    \item there exists a super-polynomial function\footnote{This means $f(r)$ is non-negative, decreasing and $\lim_{r\to\infty}r^{n}f(r)=0$ for any $n\in\mathbb{N}$.} $f(r)$ such that
    \beq
    \|h^{Y}\|\leqslant f(\diam(Y))
    \eeq
    \end{enumerate}
    We denote the set of all almost-local Hamiltonians by ${\frak{D}}^{al}$.
\end{definition}
In Ref.~\cite{Kapustin2022Noether}, it was shown that any continuous path
$F:[0,1]\to \frak{D}^{al}$, namely a continuously time-dependent almost-local Hamiltonian,
can be exponentiated, using the Lieb-Robinson bound, to a strongly continuous family
of $*$-automorphisms $\{\gamma_{s}\}_{s\in[0,1]}$, which describes the corresponding time evolution.
This leads to the notion of locally generated automorphisms (LGA), as well as its counterpart in open quantum systems, namely locally generated channels (LGC):
\begin{definition}
    An automorphism $\alpha\in\Aut(\A)$ is locally generated if there exists a continuous path $F:[0,1]\to{\frak{D}}^{al}$ such that $\alpha=\gamma_{s}|_{s=1}$. Similarly, a bath evolution $\E$ is locally generated if there is a purification $\beta$ for $\E$ on $\A\otimes \cB$, such that $\beta$ is an LGA.
\end{definition}
One major application of LGC is defining phases of mixed-states.
\begin{definition}

    For two states $\psi_{1},\psi_{2}$, we say $\psi_{1}$ is locally generated from $\psi_{2}$ if there exists an LGC $\E$ such that 
    \beq
    \psi_{1}=\psi_{2}\circ\E
    \eeq
    Two states are said to be two-way connected or in the same phase if they can be locally generated from each other. 
\end{definition}
\begin{remark}
One may also consider the case where $\psi_{i}$ carries certain symmetries, either strong or weak. Correspondingly, one may impose the requirement that the LGCs preserve these symmetries in the strong or weak sense.
\end{remark}

\begin{definition}
    A state $\Omega$ is called a product state if it satisfies $\Omega(ab)=\Omega(a)\Omega(b)$ whenever the supports of $a,b$ do not overlap. Then:
    \begin{enumerate}
        \item A state $\psi$ is called short-range entangled (SRE) if $\psi=\Omega\circ\beta$ for some pure product state $\Omega$ and an LGC $\beta$.
        \item A state $\psi$ is said to be in the trivial phase if it is two-way connected to a (possibly mixed) product state $\Omega$ .
    \end{enumerate}
\end{definition}

Below, we show that, in \textbf{1d spin chains}, if a state $\psi$ has some anomalous symmetry $\alpha:G\to\G^{lp}$ as its vN symmetry, then $\psi$ cannot be locally generated from any product state. In particular, such a state cannot fall into the trivial phase.
\begin{proposition}
    On a quantum spin chain, if $\psi$ is vN symmetric under an anomalous symmetry $\alpha:G\to\tilde{\G}^{lp}$, then it cannot be locally generated from a product state.
\end{proposition}

\begin{proof}
    Let us assume $\psi=\Omega\circ\E$ for some product state $\Omega$ and an LGC $\E$.
    By Theorem \ref{thm:vN_LSM} and Proposition.~\ref{prop:MI_split}, we only need to show $\psi$ must be a factor state and it satisfies the area law of mutual information.

    Since $\Omega$ is manifestly a factor state, we take its canonical purification $\tilde{\Omega}$ on $\A\otimes\ovl{\A}$ (this is again a product state). We also purify $\E$ into an LGA $\beta$ on $\A\otimes \cB$. Choosing a pure product state $\Omega_{\cB}$ on $\cB$ , it is easy to verify $\omega:=(\tilde{\Omega}\otimes\Omega_{\cB})\circ\beta$ is a purification of $\psi$, where we have extended $\beta$ on $\A\otimes\ovl{\A}\otimes\cB$ by tensoring identity on $\ovl{\A}$. Thus we have shown that if $\psi$ is locally generated from a product state, then it admits an SRE purification $\omega$. Thus $\Omega\circ\E$ must be a factor state for any LGC $\E$. Besides, it is known that SRE states satisfy the area law of entanglement entropy (see e.g. Lemma 4.2 of \cite{Kapustin2020invertible}):
    \beq
    S(\omega||\omega_{I}\otimes \omega_{I^{c}})<\const
    \eeq
    for any finite interval $I$,
    where $S(\cdot\|\cdot)$ is the relative entropy. By the monotonicity, restriction on subsystems does not increase the relative entropy. Therefore:
    \beq
    S(\psi\|\psi_{I}\otimes \psi_{I^{c}})<\const
    \eeq
    i.e. the mutual information of $\psi$ satisfies the area law of mutual information as well.

\end{proof}

\section{Discussion}

What is a phase of matter? In this paper we have put forward the view that it should always be characterized via the expectation values of \emph{local} operators in the thermodynamic limit. In particular, we investigated the symmetry properties of mixed states from the perspective of (quasi-)local operators. This perspective is advantageous both practically, since local operators are more accessible experimentally and numerically, and theoretically, since phases of matter are properly defined only in the thermodynamic limit. In particular, the charge coherence condition (Def.~\ref{def:charge_coherence}) yields a more robust and mathematically rigorous formulation for distinguishing and diagnosing phases of matter previously described in terms of ``SW-SSB''.

Note that, while we chose to focus our presentation on spin systems for concreteness, many of our results are in fact much more general. The definition of strong symmetry, as well as most of the core results (such as the relation with charge coherence, and the constraints on strongly symmetric bath evolutions) in fact hold in \emph{any} quantum system in which states can be characterized as positive linear functionals on a (separable) $C^*$-algebra. This can include, for example: classical systems (which correspond to commutative $C^*$-algebras); quantum field theories, when formulated in the framework of algebraic quantum field theory (AQFT) \cite{haag2012local}; constrained spin systems where the operator algebra fails to be a tensor product over sites (e.g.\ lattice gauge theories); lattice systems where the Hilbert space on each site is infinite (e.g. lattices of harmonic oscillators); and systems with non-local interactions such as the Sachdev-Ye-Kitaev (SYK) model (as long as there is still a well-defined thermodynamic limit).

In addition, mutual information plays an important role in our derivation of Lieb--Schultz--Mattis-type constraints for mixed states. It would therefore be interesting to explore how other information-theoretic quantities, such as conditional mutual information (CMI) can be used to characterize mixed states in the thermodynamic limit \cite{Sang_2025,Yi_2026, Li:2026qhc}. 

\emph{Note added.}-- in the course of preparing this work, we became aware of the parallel work Ref.~\cite{divi2026localstrongtoweakspontaneoussymmetry}. While several of our results and definitions of strong symmetry were arrived at  independently of Ref.~\cite{divi2026localstrongtoweakspontaneoussymmetry}, we learned of the ``fidelity correlator'' discussed in Sections \ref{sec:sum_main_results} and \ref{sec:cha_coh_ord_para} from the authors of Ref. \cite{divi2026localstrongtoweakspontaneoussymmetry}, after which we were able to prove equivalence with our other definitions. We also learned of the parallel work Ref.~\cite{zhang2026localdiagnosticsstrongtoweakspontaneous} which has some overlap with our results.

\section*{Acknowledgments}
We thank Chong Wang, Leonardo A. Lessa and Francisco Divi for helpful discussions. Research at Perimeter Institute is supported in part by the Government of Canada through the Department of Innovation, Science and Economic Development and by the Province of Ontario through the Ministry of Colleges, Universities and Research Excellence. RL is also supported by the Simons Collaboration on Global Categorical Symmetries through Simons Foundation grant 888996. JY is also supported by the Natural Sciences and Engineering Research Council (NSERC) of Canada.

\appendix

\section{Extension of a pure state}
\label{appendix:extension_of_pure_state}
In this appendix, we present a proof to the following claim we made in the main text:
\begin{proposition}
    Let $\psi$ be a \textbf{pure} state on a quasi-local algebra $\A$ and $\Psi$ be its extension on $\A\otimes\cB$ where $\cB$ is another $C^*$-algebra, then there exists a state $\varphi$ on $\cB$ such that
    \beq
    \Psi= \psi \otimes\varphi
    \eeq
\end{proposition}
\begin{proof}
    By Lemma \ref{lemma:purification}, we have 
    \beq
    \B(\cH_\psi)\simeq \pi_{\Psi}(\A\otimes 1)''\subseteq \B(\cH_{\Psi})
    \eeq
    Thus there exists a Hilbert space $\cK$ such that $\cH_{\Psi}\simeq\cH_{\psi}\otimes \cK$. Since $\Psi|_{\A\otimes 1}=\psi$ is pure, the Schmidt rank of $|\Psi\ra\in\cH_{\psi}\otimes \cK$ must be 1. The desired result thus follows.
\end{proof}

\section{Symmetry under canonical purification}\label{sec:sym_cano_purif}
Let $\A$ be a unital $C^*$-algebra\footnote{For us the only relevant case is the quasi-local algebra. However, the procedure and results described in this appendix works for any unital $C^*$-algebras.} , let
$\alpha\in\operatorname{Aut}(\A)$, and suppose that $\psi$ is a invariant factor state.
Denote by $\overline{\A}$ the complex conjugate algebra and define
$\bar\alpha\in\operatorname{Aut}(\overline{\A})$ by
\[
    \bar\alpha(\bar a)=\overline{\alpha(a)},
    \qquad a\in\A.
\]

The goal of this appendix is to show
\begin{proposition}
    The canonical purification of $\psi\circ\alpha$ is $\tilde{\psi}\circ(\ovl{\alpha}\otimes\alpha)$. In particular, if $\psi\circ\alpha=\psi$, then $\tilde{\psi}\circ(\ovl{\alpha}\otimes \alpha)=\tilde{\psi}$.
\end{proposition}
\begin{proof}
By uniqueness of the canonical purification, it suffices to verify that
$\tilde{\psi}\circ(\ovl{\alpha}\otimes\alpha)$ is $j$-positive and exact.
Its $j$-positivity follows immediately from that of $\tilde{\psi}$:
\beq
\tilde{\psi}\bigl(\ovl{\alpha}(\bar{a})\otimes\alpha(a)\bigr)\geqslant 0.
\eeq
Moreover,
\beq
\pi_{\tilde{\psi}\circ(\ovl{\alpha}\otimes\alpha)}(\bar{a}\otimes b)
=
\pi_{\tilde{\psi}}\bigl(\ovl{\alpha}(\bar{a})\otimes\alpha(b)\bigr),
\eeq
and hence, since $\alpha(\A)=\A$,
\beq
\pi_{\tilde{\psi}\circ(\ovl{\alpha}\otimes\alpha)}(\bar{1}\otimes\A)
=
\pi_{\tilde{\psi}}(\bar{1}\otimes\A).
\eeq
Thus $\tilde{\psi}\circ(\ovl{\alpha}\otimes\alpha)$ is also exact, and therefore
is the canonical purification of $\psi\circ\alpha$.

\end{proof}

\section{Strong symmetries in fermionic systems}\label{sec:fermionic_strong_sym}
\subsection{$\z_{2}$-graded $C^*$ algebras}
In this section, we extend our characterization of strong symmetries to lattice system with fermions. These systems are described by quasi-local algebra with a $\z_{2}$-grading (i.e. the conjugation by fermion parity $(-1)^{F}$). We will briefly recall basic concepts about $\z_{2}$-graded $C^*$-algebras and interested readers are referred to Ref.~\cite{Crismale_2020} for more details.

\begin{definition}
    Let $A$ be a $*$-algebra, $\theta$ be a $*$-automorphism of $A$ satisfying $\theta^{2}=\id_{A}$. The pair $(A,\theta)$ (or simply $A$) is called a $\z_{2}$-graded $*$-algebra.
\end{definition}
An element $a\in A$ is called \textbf{even} (resp. \textbf{odd}) if $\theta(a)=a$ (resp. $\theta(a)=-a$). The degree $\deg(a)$ is defined as
\beq
\deg(a):=\begin{cases}
    0,\quad a\text{ is even}\\
    1,\quad a\text{ is odd}
\end{cases}
\eeq
The degree of $a$ is also written as $|a|$. We will use both of them interchangeably.

For $\z_{2}$-graded algebras, the ordinary tensor product $\otimes$ is replaced by the \textbf{graded} tensor product $\hat{\otimes}$. In more details, let $A,B$ be two $\z_{2}$-graded algebras, $A\hat{\otimes} B$ is the same as $A\otimes B$ as a linear space but has a different multiplication:
\beq
(a_{1}\hat{\otimes}b_{1})\cdot(a_{2}\hat{\otimes}b_{2})=(-1)^{\deg(b_{1})\deg(a_{2})}(a_{1}a_{2})\hat{\otimes}(b_{1}b_{2})
\eeq
In order to obtain a $\z_{2}$-graded $C^*$-algebra, we have to complete it with respect to suitable norms and such norm may not be unique in general. However, we will show that for nice algebras (which we will focus on later), such uniqueness can be shown, see Sec.~\ref{sec:uniqueness_gr_tensor_prod}.

We now introduce states and automorphisms compatible with the grading.

\begin{definition}\label{def:Z2_graded_state}
Let $(A,\theta)$ be a unital $\mathbb{Z}_{2}$-graded $C^{*}$-algebra. A state on $A$ is a positive linear functional $\psi\colon A\to\mathbb{C}$ satisfying
$
\psi(1_A)=1.
$
It is called \emph{parity-invariant}\footnote{In Ref.~\cite{Griseta2024Z2_gradedC*} such states are called \textbf{even} states.} if $
\psi\circ\theta=\psi.
$
\end{definition}
\begin{definition}
Let $(A,\theta)$ be a $\mathbb{Z}_{2}$-graded $C^{*}$-algebra. A $*$-automorphism $\alpha\in\operatorname{Aut}(A)$ is called \emph{graded} if $
\alpha\circ\theta=\theta\circ\alpha.
$

\end{definition}

Physically, states of different fermion parity cannot be coherently superposed. Consequently, one should restrict to parity-invariant states\footnote{One manifestation of this is we allowed non-parity-invariant states, in general it would not be possible to take the tensor product of states. Suppose that $\omega_i$ is a non-parity-invariant state on $A_i$ for $i=1,2$. Then there exist odd self-adjoint elements $a_i\in A_i$ such that $\omega_i(a_i)\neq 0$. If the product state existed, it would satisfy
\beq
(\omega_1\hat{\otimes}\omega_2)(a_1\hat{\otimes}a_2)
=\omega_1(a_1)\omega_2(a_2).
\eeq
However, $a_1\hat{\otimes}a_2$ is skew-adjoint, so the left-hand side is purely imaginary, whereas the right-hand side is real and nonzero. Hence $\omega_1\hat{\otimes}\omega_2$ does not define a state; in particular, a graded product state exists only if there is at most one non-parity-invariant tensor factor.}.
\textit{For this reason, states will always mean parity-invariant states in the following discussion.}

States induce a natural grading on their GNS representations, with respect to which state-preserving graded automorphisms are implemented by even unitaries.

\begin{lemma}\label{lemma:implementer_fermionic}
    Let $\psi$ be a state, then the GNS Hilbert space $\cH_{\psi}$ is a $\z_{2}$-graded. Besides, if $\alpha$ is graded $*$-automorphism which preserves $\psi$, then $\alpha$ can be implemented by an even unitary operator $U_{\alpha}$ on $\cH_{\psi}$.
\end{lemma}
\begin{proof}
The first assertion follows directly from the uniqueness of the ungraded GNS representation. Let $\gamma$ denote the self-adjoint unitary implementing $\theta$ on $\cH_{\psi}$, which also determines the $\z_2$-grading of $\cH_{\psi}$. It remains to prove the second assertion.

By the uniqueness of the GNS representation, there exists a unitary operator $U_{\alpha}$ implementing $\alpha$. Since $\alpha$ commutes with $\theta$, i.e. $[\Ad_{U_{\alpha}},\Ad_{\gamma}]=0$, this implies $\gamma U_{\alpha}\gamma=\pm U_{\alpha}$. On the other hand, since $\gamma|\psi\ra=U_{\alpha}|\psi\ra=|\psi\ra$, only $\gamma U_{\alpha}\gamma=U_{\alpha}$ is possible. Therefore, $U_{\alpha}$ is even.
\end{proof}

\subsection{The graded commutant and the graded von Neumann algebra}

\subsubsection{The graded commutant}

\begin{definition}
The \emph{graded commutant} of the GNS representation of a state $\psi$ is
\[
\pi_\psi(A)^{\mathrm{gr}\prime}
   :=
   \left\{
      T\in\mathcal B(\mathcal H_\psi):
      T\pi_\psi(a)
      =
      (-1)^{\deg(T)\deg(a)}\pi_\psi(a)T
      \text{ for all homogeneous }a\in A
   \right\},
\]
where the condition is first imposed on homogeneous operators $T$ and then
extended linearly.
\end{definition}

Thus, an even operator in the graded commutant commutes with every
$\pi_\psi(a)$, whereas an odd operator commutes with the even elements and
anticommutes with the odd elements.

To compare the graded commutant with the ordinary commutant, introduce the
Klein transformation \cite{Crismale_2020}:
\[
\kappa_\psi=\operatorname{Ad}_{U_\psi},
\qquad
U_\psi:=\frac{1+i\gamma}{\sqrt{2}}.
\]
Hence $\kappa_\psi$ is a $*$-automorphism of
$\mathcal B(\mathcal H_\psi)$. It acts on homogeneous operators as
\[
\kappa_\psi(T)
=
\begin{cases}
T, & \deg(T)=0,\\[2mm]
i\gamma T, & \deg(T)=1.
\end{cases}
\]
Moreover,
\[
\kappa_\psi^2=\operatorname{Ad}_{\gamma}.
\]

If $T$ is odd, then the relation
\[
T\pi_\psi(a)=(-1)^{\deg(a)}\pi_\psi(a)T
\]
is equivalent to
\[
(i\gamma T)\pi_\psi(a)
   =\pi_\psi(a)(i\gamma T).
\]
For even $T$, graded commutation is simply ordinary commutation.
Consequently,
\[
\kappa_\psi\left(
\pi_\psi(A)^{\mathrm{gr}\prime}
\right)
=
\pi_\psi(A)'.
\]
Equivalently,
\[
\pi_\psi(A)^{\mathrm{gr}\prime}
=
\kappa_\psi^{-1}\left(\pi_\psi(A)'\right).
\]
Thus, the graded commutant and the ordinary commutant are canonically
$*$-isomorphic through the Klein transformation.

Since the Klein transformation identifies the graded and ordinary
commutants, the usual irreducibility criterion for pure states immediately
gives the following.

\begin{proposition}
A state $\psi$ is pure if and only if $\pi_\psi(A)^{\mathrm{gr}\prime}=\bbC 1.$

\end{proposition}

\subsubsection{The graded von Neumann algebra generated by the state}

\begin{definition}
The \emph{graded von Neumann algebra generated by $\psi$} is the double
graded commutant
\[
\cM_\psi^{\mathrm{gr}}
   :=
   \pi_\psi(A)^{\mathrm{gr}\prime
   \mathrm{gr}\prime}.
\]
\end{definition}

Let $\mathcal S\subseteq\mathcal B(\mathcal H_\psi)$ be any
$\gamma$-invariant set. The Klein transformation gives
\[
\mathcal S^{\mathrm{gr}\prime}
=
\kappa_\psi^{-1}(\mathcal S').
\]
Applying this identity twice, we obtain
\begin{align*}
\mathcal S^{\mathrm{gr}\prime\mathrm{gr}\prime}
&=
\kappa_\psi^{-1}
\left(
   \left(\mathcal S^{\mathrm{gr}\prime}\right)'
\right)\\
&=
\kappa_\psi^{-1}
\left(
   \left(\kappa_\psi^{-1}(\mathcal S')\right)'
\right)\\
&=
\kappa_\psi^{-2}(\mathcal S'').
\end{align*}
Since
\[
\kappa_\psi^2=\operatorname{Ad}_{\gamma}
\]
and $\mathcal S''$ is invariant under
$\operatorname{Ad}_{\gamma}$, it follows that
\[
\mathcal S^{\mathrm{gr}\prime\mathrm{gr}\prime}
=
\mathcal S''.
\]
Taking $\mathcal S=\pi_\psi(A)$ gives
\[
\cM_\psi^{\mathrm{gr}}
=
\pi_\psi(A)''
\]

Therefore, the graded and ungraded bicommutant constructions produce the
same von Neumann algebra as a concrete operator algebra. The difference
lies not in the resulting algebra, but in the appropriate notion of its
center.

\subsubsection{Graded factors}

Let
$\cM_\psi:=\pi_\psi(A)''$
with grading $\widehat\theta_\psi
   :=\operatorname{Ad}_{\gamma}\big|_{\cM_\psi}.$

\begin{definition}
The \emph{graded center} of $\cM_\psi$ is
\[
Z_{\mathrm{gr}}(\cM_\psi)
   :=
   \cM_\psi
   \cap
   \cM_\psi^{\mathrm{gr}\prime}.
\]
The graded von Neumann algebra $\cM_\psi$ is called a
\emph{graded factor} if
\[
Z_{\mathrm{gr}}(\cM_\psi)=\mathbb C1.
\]
\end{definition}

An even element belongs to the graded center precisely when it belongs to
the ordinary center. On the other hand, there are no nonzero odd elements
in the graded center. Indeed, if $x$ were odd and graded-central, then,
since $x^*$ is also odd,
\[
xx^*=-x^*x.
\]
The left-hand side is positive and the right-hand side is negative, so
both must vanish. Hence $x=0$.

It follows that
\[
Z_{\mathrm{gr}}(\cM_\psi)
=
Z(\cM_\psi)_{\mathrm{even}}
=
Z(\cM_\psi)^{\widehat\theta_\psi},
\]
where
\[
Z(\cM_\psi)^{\widehat\theta_\psi}
=
\left\{
z\in Z(\cM_\psi):
\gamma z\gamma=z
\right\}.
\]
Consequently,
\[
\cM_\psi\text{ is a graded factor}
\quad\Longleftrightarrow\quad
Z(\cM_\psi)^{\widehat\theta_\psi}
=
\mathbb C1.
\]

By comparison, $\cM_\psi$ is an ordinary factor if and only if
\[
Z(\cM_\psi)=\mathbb C1.
\]
Thus,
\[
\cM_\psi\text{ is an ordinary factor}
\quad\Longrightarrow\quad
\cM_\psi\text{ is a graded factor},
\]
but the converse need not hold. Below we give an example of graded factor which fails to be an ordinary factor.

\begin{example}
    Let
\[
\cM=\mathcal B(\mathcal K)\oplus\mathcal B(\mathcal K)
\]
and let the grading interchange the two summands:
\[
\widehat\theta(x\oplus y)=y\oplus x.
\]
Then
\[
Z(\cM)
=
\mathbb C1\oplus\mathbb C1,
\]
so $\cM$ is not an ordinary factor. However,
\[
Z(\cM)^{\widehat\theta}
=
\left\{
\lambda1\oplus\lambda1:\lambda\in\mathbb C
\right\}
=
\mathbb C1,
\]
and hence $\cM$ is a graded factor.
\end{example}

A graded factor has ordinary center either $\bbC $ or
$\bbC\oplus\bbC$, with the grading exchanging the two summands. The second
case occurs in the Kitaev chain, the nontrivial one-dimensional fermionic
invertible phase without any symmetry except fermion parity
\cite{matsui2020splitpropertyfermionicstring,Bourne2020fermion}.

\subsection{Canonical purification for graded case}

Let $(\A,\theta)$ be a $\mathbb Z_2$-graded $C^*$-algebra and let
$\psi$ be an even factor state:
\[
    \psi\circ\theta=\psi.
\]
We write $\widehat\otimes$ for the graded tensor product. Thus, for
homogeneous elements,
\[
    (\bar b\widehat\otimes a)(\bar d\widehat\otimes c)
    =
    (-1)^{|a||d|}\overline{bd}\widehat\otimes ac.
\]

Choose a standard form
\[
    (M,\mathcal H,J,\mathcal P),
    \qquad
    M=\pi(\A)'',
\]
and let $|\psi\ra\in\mathcal P$ be the canonical vector representing
$\psi$. Let $\gamma$ be the standard unitary implementing the grading:
\[
    \gamma=\gamma^*,
    \qquad
    \gamma^2=1,
    \qquad
    \gamma\pi(a)\gamma=\pi(\theta(a)).
\]
Since $\psi$ is even,
\[
    \gamma|\psi\ra=|\psi\ra,
    \qquad
    \gamma J=J\gamma.
\]

For $x=x_0+x_1\in B(\mathcal H)$, where
\[
    \gamma x_j\gamma=(-1)^j x_j,
\]
define the Klein-twist automorphism
\[
    \eta_\gamma(x)=x_0+i\gamma x_1.
\]
The graded commutant of $M$ is
\[
    M^{\mathrm{gr}\prime}
    :=
    \eta_\gamma(M').
\]
For homogeneous $x\in M$ and
$y\in M^{\mathrm{gr}\prime}$, one has
\[
    xy=(-1)^{|x||y|}yx.
\]

The graded purification of $\psi$ is the state on
$\overline{\A}\widehat\otimes_{\max}\A$ defined by
\[
    \widetilde{\psi}(\bar b\widehat\otimes a)
    =
    \la
        \psi|
        \eta_\gamma\!(J\pi(b)J)
        \pi(a)|\psi\ra
\]

Indeed, the maps
\[
    \bar b\longmapsto
    \eta_\gamma\!\left(J\pi(b)J\right),
    \qquad
    a\longmapsto\pi(a),
\]
have graded-commuting ranges and therefore define a representation of
the graded tensor product. Hence the displayed formula defines a
state. Its marginals are
\[
    \widetilde{\psi}(\bar 1\widehat\otimes a)=\psi(a),
    \qquad
    \widetilde{\psi}(\bar b\widehat\otimes 1)
    =\overline{\psi(b)}.
\]
Moreover, since $\gamma|\psi\ra=|\psi\ra$,
\[
    \widetilde{\psi}\circ
    (\bar\theta\widehat\otimes\theta)
    =
    \widetilde{\psi},
\]
so the purified state is even.

To prove purity, observe that the von Neumann algebra generated by the
two graded copies is
\[
    N=M\vee\eta_\gamma(M'),
\]
and
\[
    N'
    =
    M'\cap\eta_\gamma(M).
\]
Its even part is $\mathbb C1$ because $M$ is a factor. A nonzero odd
element in this intersection would give an odd unitary
$u\in M$ satisfying
\[
    uy=\theta(y)u,
    \qquad y\in M.
\]
It would therefore implement $\theta$, implying
$\theta(u)=u$, in contradiction with the oddness condition
$\theta(u)=-u$. Hence
\[
    N'=\mathbb C1,
\]
so the doubled representation is irreducible and
$\widetilde{\psi}$ is pure.

Finally, if $\alpha\in\operatorname{Aut}(\A)$ is graded and preserves
$\psi$, then its standard implementer $U_\alpha$ satisfies
\[
    U_\alpha J=JU_\alpha,
    \qquad
    U_\alpha\gamma=\gamma U_\alpha,
    \qquad
    U_\alpha|\psi\ra=|\psi\ra.
\]
It follows immediately from the defining formula that
\[
    \widetilde{\psi}\circ
    (\bar\alpha\widehat\otimes\alpha)
    =
    \widetilde{\psi}
\]

\subsection{Strong symmetry for fermionic systems}
We now extend strong symmetry (Def.~\ref{def:strong_sym}) to fermionic systems. The fermionic quasi-local algebra is similarly defined as in the spin case, except that each site carries a possibly trivial $\z_2$-grading and all tensor products are graded. We again write $(\A,\theta)$ for the fermionic quasi-local algebra.

Let $\psi$ be a state invariant under a grading-preserving $*$-automorphism $\alpha$, and let $U_\alpha$ denote its unitary implementation in the GNS representation of $\psi$.

\begin{definition}
The state $\psi$ is \emph{strongly symmetric} under $\alpha$ if
\[
U_\alpha\in\cM_\psi.
\]
\end{definition}
The goal is to establish the fermionic version of Theorem.~\ref{thm:equi_def}, i.e.

\begin{theorem}\label{thm:fermionic_equivalence}
    For a state $ \psi$ and $\z_{2}$-graded quasi-local algebra $(\A,\theta)$, the following statements are equivalent:
    \begin{enumerate}
        \item $\psi$ is strongly symmetric under $\alpha$.
        \item The canonically doubled state $\tilde{\psi}$ is symmetric under $1\hat{\otimes}\alpha$.
        \item For any graded $C^*$ algebra $\cB$ and any extension $\Psi$ of $\psi$ on $\cB\hat{\otimes}\A$, $\Psi$ is symmetric under $1\hat{\otimes }\alpha$.
    \end{enumerate}
\end{theorem}
By Lemma.~\ref{lemma:implementer_fermionic}, the unitary operator $U_{\alpha}$ must be even, all graded commutations involving $U_{\alpha}$ reduces to ordinary commutations. Thus, the proof is identical to the spin case. We will omit it here.

We now formulate the charge oherence criterion for fermionic systems.
Consider a continuous map
\[
\alpha\colon G\to\Aut(\A)
\]
be an action of a compact group by grading-preserving automorphisms. Recall here $\Aut(\A)$ is equipped with the strong topology.

\begin{proposition}
The following are equivalent:
\begin{enumerate}
    \item $\psi$ is strongly symmetric under $\alpha$.
    \item $\psi$ satisfies charge coherence: for every charged operator
    $O\in\A$,
    \[
    F(\psi,\psi_O)=0,
    \qquad
    \psi_O(\,\cdot\,):=\psi(O^*\,\cdot\,O).
    \]
\end{enumerate}
\end{proposition}
For the same reason as above, the proof is identical to the spin case so we will omit it.

\subsection{On the uniqueness of graded tensor product}\label{sec:uniqueness_gr_tensor_prod}
We show that the graded tensor product of $(A,\theta_A)$ and $(B,\theta_B)$ admits a unique $C^*$-norm whenever either $A$ or $B$ is nuclear, in particular, when either is UHF. This applies to the fermionic quasi-local algebra $\A$. Until this uniqueness is established, however, the choice of tensor-product completion must be specified.
\begin{definition}[Algebraic graded tensor product]
As a vector space, $A\odgr B$ is the ordinary algebraic tensor product
$A\odot B$. Its multiplication and involution are prescribed on homogeneous
elementary tensors by the Koszul sign rules
\begin{align*}
 (a\odgr b)(a'\odgr b')
   &=(-1)^{|b||a'|}\,aa'\odgr bb',\\
 (a\odgr b)^*
   &=(-1)^{|a||b|}\,a^*\odgr b^*.
\end{align*}
They are then extended linearly. The tensor $a\odgr b$ has degree
$|a|+|b|$, and the grading automorphism is
\[
  \Theta_0=\theta_A\odot\theta_B,
  \qquad \Theta_0(a\odgr b)=\theta_A(a)\odgr\theta_B(b).
\]
\end{definition}
However, the algebra $A\odgr B$ above is incomplete in general, one has to specify the norm on $A\odgr B$ so that one can complete it.

\begin{definition}[Graded $C^*$-tensor norm]
A graded $C^*$-tensor norm $\alpha$ on $A\odgr B$ is a $C^*$-norm such that
\[
  \|a\odgr b\|_\lambda=\|a\|\,\|b\|
  \quad(a\in A,\ b\in B),
\]
and $\Theta_0$ extends to an isometric order-two automorphism of the
completion $A\grten_\lambda B$, where $\lambda$ labels different tensor-product norms.
Equivalently, the completion is a graded
$C^*$-algebra containing the two factors with their prescribed norms and
graded commutation relation.
\end{definition}

The two extremal norms are the following \cite{murphy2014c}.

\begin{itemize}
\item The \emph{minimal} (or spatial) graded norm $\|\cdot\|_{\Min}$ is
obtained from faithful grading-covariant representations of $A$ and $B$ on
graded Hilbert spaces and their graded Hilbert-space tensor product.

More concretely, if $\pi_A$ and $\pi_B$ are faithful graded representations
and $\gamma_A$ is the grading operator for $\pi_A$, then for homogeneous
$a,b$ the spatial representation can be written
\[
  a\grten b\longmapsto
  \pi_A(a)\gamma_A^{|b|}\otimes\pi_B(b).
\]
The formula is independent of the chosen faithful graded representations up
to equality of the resulting norm.

\item The \emph{maximal} graded norm is the universal norm
\[
  \|x\|_{\Max}=\sup\{\|\pi(x)\|:\pi
  \text{ is a }*\text{-representation of }A\odgr B\}.
\]
Equivalently, one may take the supremum over pairs of representations of
$A$ and $B$ whose homogeneous ranges supercommute:
\[
  \pi_B(b)\pi_A(a)=(-1)^{|a||b|}\pi_A(a)\pi_B(b).
\]
\end{itemize}

Any other tensor-product norm $\|\cdot\|_{\lambda}$ satisfy:
\beq
\|\cdot\|_{\Min}\leqslant \|\cdot\|_{\lambda}\leqslant \|\cdot\|_{\Max}
\eeq
Thus, we only need to show the maximal norm agrees with the minimal norm when either $A$ or $B$ is nuclear.

The key point is to consider the crossed-product algebra
$C:=A\rtimes_{\theta_{A}}\z_{2}$ generated by a
copy of $A$ and a self-adjoint unitary $u$ satisfying $ua=\theta_A(a)u$. We consider another $\z_{2}$-action $\Theta$ on $C$ defined by $\Theta(a)=a$ for $a\in A$ while $\Theta(u)=-u$.

The uniqueness follows by establishing the following lemma:
\begin{lemma}\label{lemma:fixed_point}
    For $\lambda\in \{\max,\min\}$,
    \beq
    A\grten_{\mathrm{gr},\lambda}B
  \simeq
  \bigl(A\rtimes_{\theta_A}\z_{2}\otimes_\lambda B\bigr)^{
  \Theta\otimes\theta_B}
    \eeq
    Note that the tensor product on RHS is ungraded.
\end{lemma}
The desired uniqueness follows immediately from the preceding lemma. Indeed, if $A$ is nuclear, then so is $A\rtimes_{\theta_A}\z_2$, since $\z_2$ is finite (see Corollary 7.18 of Ref.~\cite{WilliamsCrossedProducts}). Hence the right-hand side of the isomorphism in the lemma is independent of the tensor norm $\lambda$ (hence so is the left-hand side). The case where $B$ is nuclear follows similarly.

\begin{proof}[Proof to Lemma \ref{lemma:fixed_point}]
    For homogeneous $b\in B$, define the Klein transformation
\[
  \kappa(a\grten b)=a\,u^{|b|}\otimes b
\]
and extend linearly to $A\odgr B$. Then we first check that $\kappa$ is an injective $*$-homomorphism from $A\odgr B$ into
$C\odot B$. Indeed, let $a,a'\in A,b,b'\in B$ be homogeneous, then
\[
  u^{|b|}a'=\theta_A^{|b|}(a')u^{|b|}
  =(-1)^{|b||a'|}a'u^{|b|}.
\]
Therefore
\begin{align*}
 \kappa(a\grten b)\kappa(a'\grten b')
 &=a u^{|b|}a'u^{|b'|}\otimes bb'\\
 &=(-1)^{|b||a'|}aa'u^{|b|+|b'|}\otimes bb'\\
 &=\kappa\bigl((a\grten b)(a'\grten b')\bigr).
\end{align*}
Likewise,
\[
  (a u^{|b|})^*=u^{|b|}a^*
  =(-1)^{|a||b|}a^*u^{|b|},
\]
which gives $\kappa((a\grten b)^*)=\kappa(a\grten b)^*$. Finally,
the vector-space decomposition $C=A\oplus Au$ shows that $\kappa$ is
injective.
In fact,
\beq
\kappa(A\odgr B)&=(A\odot B^{(0)})\oplus(Au\odot B^{(1)})\\
  &=(C\odot B)^{\Theta\odot\theta_B}
\eeq
note that $\kappa$ sends $A\odgr B^{(0)}$ into the first term and $A\odgr B^{(1)}$ to the second term.

For the statement after completion, use the contractive conditional
expectation
\[
  \E_\lambda(z)=\tfrac12\bigl(z+(\Theta\otimes_{\lambda}\theta_{B})(z))
  :C\otimes_\lambda B\longrightarrow (C\otimes_\lambda B)^{\Theta\otimes_{\lambda}\theta_{B}}.
\]
The algebraic tensor product is dense and invariant under $\Theta\otimes_{\lambda}\theta_{B}$;
hence applying $\E_\lambda$ to algebraic approximants proves that its
algebraic fixed points are dense in the completed fixed-point algebra.
\end{proof}

\section{A review on group cohomology}\label{sec:group_cohomology}

In this section, we review the basics of group cohomology and differentiable group cohomology. For group cohomology, there are many materials in the literature \cite{brown2012cohomology,Weibel_1994_group,Chen2010,Yang_2017,kobayashi2025projectiverepresentationsbogomolovmultiplier}. See also appendix A.1 of Ref. \cite{kapustin2024anomalous} and Ref. \cite{brylinski2000differentiable} for the version for Lie groups. We will only cover the motivations and basics here.

\subsection{Projective representations in quantum mechanics}

To motivate group cohomology, we start with projective representations in quantum mechanics. Suppose we have a symmetry group $G$ (assumed to be unitary and discrete for simplicity) acting on a Hilbert space $\cH$. Usually this symmetry action is given by a homomorphism $\rho:G\to U(\cH)$, \ie a unitary representation of $\cH$. More explicitly, for each $g\in G$, we assign a unitary operator $\rho(g)$ such that
\beq
\rho(g)\rho(h)=\rho(gh),\quad\forall\,g,h\in G
\eeq
However, in quantum mechanics, states are {\it not} really a vector in $\cH$, but a {\it{ray}}. That means a state $|\psi\ra$ is the same as $e^{i\theta}|\psi\ra$ as a quantum state. Thus, the space of states is not literally $\cH$, but the projective space $P(\cH)$. This for allows more general symmetry actions as
\beq
\rho(g)\rho(h)=\omega(g,h)\rho(gh)
\eeq
where\footnote{In principle, one has to show that the phase $\omega(g,h)$ is the same on each quantum state. This relies the coherence of these states and one can find the proof in Sec. 2.2 of Ref. \cite{weinberg2005quantum}.} $\omega(g,h)\in \U$. This $\rho$ is a representation up to a phase $\omega$ and is called a projective representation.
Moreover, the matrix multiplication is associative, so 
\beq
(\rho(g)\rho(h))\rho(k)=\rho(g)(\rho(h)\rho(k))
\eeq
This imposes the following constraint on $\omega$,
\beq\label{eq:2-cocycle_condition}
\omega(g,h)\omega(gh,k)=\omega(g,hk)\omega(h,k)
\eeq
Any function $G\times G\to \U$ satisfying Eq. \eqref{eq:2-cocycle_condition} is called a 2-cocycle. Furthermore, one can redefine the phase of $\rho(g)\to \tilde{\rho}(g)=\rho(g)\eta(g),\eta(g)\in\U$ (we do not require $\eta:G\to \U$ to be a homomorphism), and the resulting 2-cocycle is
\beq\label{eq:shift_2-coboundary}
\tilde{\omega}(g,h)=\omega(g,h)\eta(g)\eta(h)\eta(gh)^{-1}
\eeq
One can easily check that $\tilde{\omega}$ again satisfies the 2-cocycle condition, Eq. \eqref{eq:2-cocycle_condition}. If there exists $\eta(g)$ such that $\tilde{\omega}(g,h)=1$ for all $g,h\in G$, then we say that $\omega$ is a 2-coboundary or trivial. Any two 2-cocycles $\omega$ and $\tilde{\omega}$ related by Eq. \eqref{eq:shift_2-coboundary} are viewed as equivalent, since they differ only by the artificial choice of phase factors $\eta(g)$ of representation matrix $\rho(g)$. We write $\omega\sim \tilde{\omega}$ if $\omega$ and $\tilde\omega$ are equivalent. The space of 2-cocycles modulo this equivalence $\sim$ is the so-called the degree 2 group cohomology of $G$, denoted by $\rH^{2}(G;\U)$.
\begin{example}
    Let us consider $G=\z_{2}\times\z_{2}$. We write its elements as $(a,b)$ where $a,b=0,1 \mod{2}$. Then we define a projective representation $\rho$ as follows:
    \beq
    \begin{split}
        \rho(0,0)=I,\,\rho(1,0)=\sigma_{x}\\
        \rho(0,1)=\sigma_{y},\,\rho(1,1)=\sigma_{z}
    \end{split}
    \eeq
    Note that $\rho(0,1)\rho(1,0)=-i\rho(1,1)$ hence $\omega((0,1),(1,0))=-i$. Similarly, $\omega((1,0),(0,1))=i$. One can show that this 2-cocycle is not a 2-coboundary and hence defines the nontrivial class in $\rH^{2}(\z_{2}\times\z_{2};\U)\simeq \z_{2}$. In the context of symmetry-protected topological phases, this projective representation describes the boundary of the cluster state \cite{Son_2011}.
     
\end{example}
\begin{example}\label{example:0+1d_SO3}
    Consider the case where $G=SO(3)$, the spin rotation symmetry \footnote{Actually, this a subtler case because $SO(3)$ is a Lie group so it requires more careful treatment, see e.g. Ref.~\cite{brylinski2000differentiable} for a more careful treatment. We omit this subtlety for now.}. One can show that $\rH^{2}(SO(3);\U)\simeq \rHom(\pi_{1}(SO(3)),\U)\simeq \z_{2}$, and this class is trivial if the (total) spin quantum number $S\in \z$ and it is nontrivial if $S\in\z+\frac{1}{2}$.
\end{example}

A projective representation provides the following constraint on quantum states.
\begin{proposition}
        If $G$ acts on the Hilbert space $\cH$ via a projective representation $\rho$ whose associated 2-cocycle $\omega\not =1\in\rH^{2}(G;\U)$, then there cannot be a nonzero $G$-symmetric state.
\end{proposition}
\begin{proof}
    Suppose $|\psi\ra$ is a $G$-symmetric state, that is 
    \beq
    \rho(g)|\psi\ra=\eta(g)^{-1}|\psi\ra
    \eeq
    where $\eta(g)\in \U$ is any $\U$-valued function on $G$. Then one redefines $\tilde{\rho}(g)=\rho(g)\eta(g)$, this shifts $\omega$ by a 2-coboundary and the resulting $\tilde{\omega}$ (see Eq. \eqref{eq:shift_2-coboundary}) is nontrivial, \ie there exists $g,h\in G$ such that $\tilde{\omega}(g,h)\not=1$. Now
    \beq
    \tilde{\rho}(g)|\psi\ra=|\psi\ra,\,\forall \,g\in G
    \eeq
    One can calculate $\tilde{\rho}(g)\tilde{\rho}(h)|\psi\ra$ in 2 different ways
    \beq
    \begin{split}
        \tilde{\rho}(g)(\tilde{\rho}(h)|\psi\ra)&=\tilde{\rho}(g)|\psi\ra=|\psi\ra\\
        (\tilde{\rho}(g)\tilde{\rho}(h))|\psi\ra&=\tilde{\omega}(g,h)\tilde{\rho}(gh)|\psi\ra=\tilde{\omega}(g,h)|\psi\ra
    \end{split}
    \eeq
    By assumption, $\tilde{\omega}(g,h)\not =1$ for some $g,h\in G$. Hence $|\psi\ra=0$, which shows that there is no nonzero $G$-symmetric state.
\end{proof}

As a corollary, consider a $G$-symmetric Hamiltonian $H$ which has a $G$ symmetry that acts projectively. We have

\begin{corollary}
    If $\rho$ is nontrivial projective representation, then a $G$-symmetric Hamiltonian must have degenerate ground states which break the $G$-symmetry.
\end{corollary}

This can be viewed as $(0+1)d$ version of anomaly constraints.

\begin{example}
    Consider a system made of $N$ qubits (or equivalently, spin $\frac{1}{2}$'s), whose Hamiltonian $H$ has a $G=SO(3)$ symmetry encountered in example \ref{example:0+1d_SO3}. If $N=1\mod{2}$, then this system must be at least 2-fold degenerate. For example, consider $N=1$, for the Hamiltonian $H$ to be $SO(3)$-symmetric, it has to commute with all Pauli operators. It is easy to check that $H$ must be $\lambda I$ for some $\lambda\in \bbC$ and $I$ is the identity operator. Hence the ground states are trivially 2-fold degenerate. However, for $N=2$ where the total spin is an integer, one can take 
    \beq
    H=J\vec{S}_{1}\cdot \vec{S}_{2},J>0
    \eeq
    where the ground state is non-degnerate.
\end{example}

\subsection{Group cohomology}

Now we present the definition of group cohomology in general. Let $G$ be a discrete group, one defines a space $BG$ which is a collection of spaces $\{G^{n}\}_{n=1,2,...}$ equipped with a collection of maps $d_{k}:G^{n}\to G^{n-1},\,k=0,1,...,n$ (called face maps). Explicitly,
\beq\label{eq:face_maps}
d_{k}(g_{1},g_{2},...,g_{n})=\begin{cases}
    (g_2,...,g_n),\,k=0\\
    (g_1,...,g_{k}g_{k+1},...,g_{n}),\,0<k<n\\
    (g_1,...,g_{n-1}),k=n
\end{cases}
\eeq
One can check that if $d=\sum_{k=0}^{n}(-1)^{k}d_{k}$, then $d^{2}=0$.
Let $A$ be an Abelian group (with {\it discrete topology}). For example $A$ can be $\z_{2}$, $\z$, $\R$ or $\U$. We denote all $A$-valued functions on $BG$ as $C^{\bullet}(BG,A)$. For example, one writes $\omega\in C^{2}(BG,A)$ if $\omega:G^{2}\to A$.
Consider an $A$-valued function $\omega$ on $G^{n-1}$. The maps $d_{k}:G^{n}\to G^{n-1}$ induces a pullback of $\omega$, \ie $d_{k}^{*}\omega:=\omega\circ d_{k}$ on $G^{n}$. We denote $\delta=d^{*}$ (it follows that $\delta^{2}=0$), thus $C^{\bullet}(BG,A)$ together with $\delta$ becomes a cochain complex.

\begin{definition}\label{def:group_cohomology}
    A function $\omega:G^{n}\to A$ is said to be an $n$-cocycle if $\delta \omega=0$. We denote the space of all $n$-cocycles by $\mathrm{Z}^{n}(G;A)$. Besides, if an $n$-cocycle $\omega$ satisfies $\omega=\delta \eta$ for some $\eta\in C^{n-1}(G;A)$, it is called an $n$-coboundary. The space of all $n$-coboundary is denoted as $\mathrm{B}^{n}(G;A),n>1$. Besides, $\mathrm{B}^{1}(G;A)$ is defined to be 0.
\end{definition}

\begin{definition}
    The degree $n$ group cohomology of $G$ is defined to be
    \beq\label{eq:group_coho}
    \rH^{n}(G;A)=\frac{\mathrm{Z}^{n}(G;A)}{\mathrm{B}^{n}(G;A)}
    \eeq
    In more details, $\rH^{n}(G;A)$ are defined to be  equivalence classes of $n$-cocycles under the equivalence relation $\omega\simeq \omega+\delta\eta$ where $\omega\in\mathrm{Z}^{n}(G;A)$ and $\delta\eta\in \mathrm{B}^{n}(G;A)$.
\end{definition}

\begin{example}
    Let us consider a function $\omega:G\to A$ (where $G$ acts trivially on $A$) or equivalently $\omega$ here is a 1-cochain. Now we compute $\delta\omega$
    \beq
    \delta\omega(g_{1},g_{2})=(d_{0}^{*}\omega-d_{1}^{*}\omega+d^{*}_{2}\omega)(g_{1},g_{2})=\omega(g_{1})+\omega(g_{2})-\omega(g_{1}g_{2})
    \eeq
    where we have used Eq. \eqref{eq:face_maps}, \eg
    \beq
    d_{1}^{*}\omega(g_{1},g_{2})=\omega(d_{1}(g_{1},g_2))=\omega(g_{1}g_{2})
    \eeq
    Then $\omega$ is a 1-cocycle iff it is a homomorphism, \ie $\omega(g_1 g_2)=\omega(g_{1})+\omega(g_{2})$. We conclude
    \beq
    \rH^{1}(G;A)=\rHom(G,A)
    \eeq
\end{example}

\begin{example}
    Now we consider a 2-cochain, again denoted by $\omega:G^{2}\to A$. Then one calculates $\delta\omega$ as follows
    \beq
    \delta\omega(g_{1},g_{2},g_{3})=\omega(g_{2},g_{3})-\omega(g_{1}g_{2},g_{3})+\omega(g_{1},g_{2}g_{3})-\omega(g_{1},g_{2})
    \eeq
    If one writes the group action in $A$ as multiplication rather than addition, one immediately recognizes $\delta\omega=0$ is exactly the 2-cocycle condition Eq. \eqref{eq:2-cocycle_condition} in projective representations. One can shift $\omega$ by a 2-coboundary $\delta\eta$. As we computed in the last example, this corresponds to
    \beq
    \omega(g_{1},g_{2})\to \tilde{\omega}(g_{1},g_{2})=\omega(g_{1},g_{2})+\eta(g_{1})+\eta(g_{2})-\eta(g_{1}g_{2})
    \eeq
    In the context of projective representation, this amounts to redefining our representation matrices by a phase Eq. \eqref{eq:shift_2-coboundary}.
\end{example}
Group cohomology of higher degrees are used to classify 't Hooft anomalies in physics. We will explain this in some more details in Sec. \ref{sec:sym_action}.

\begin{remark}
    The geometry behind Eqs.
    \eqref{eq:face_maps} and \eqref{eq:group_coho} is that we are doing simplicial cohomology on the space $BG$ (which is known as classifying space in mathematics), see, \eg, Ref. \cite{Weibel_1994_simplicial} for more details.
\end{remark}

\section{Defining anomaly indices}\label{sec:anomaly_index}
In this section,  we outline the construction of anomaly index of a symmetry action $\alpha:G\to\G^{\QCA}$ in spin chains therefore \textit{we exclusively work with $\Lambda\simeq\z$ in this appendix.}

Given a symmetry group $G$, by slightly abusing the notations{\footnote{Previously the notation $\alpha$ is used to represent an operation acting on operators, but here we use it to represent a map from the symmetry group $G$ to all possible QCA opetations $\G^{\QCA}$.}}, the symmetry action can be represented by a group homomorphism $\alpha: G\to \G^{\QCA}$. This symmetry may contain internal and/or translation symmetry, and the internal symmetry, which acts as a finite-depth quantum circuit, may be discrete or continuous, on-site or non-on-site. This general type of symmetry actions covers many physically relevant cases.

First, suppose $\alpha$ is an internal symmetry action (\ie it contains no translation). For an arbitrary site, say, the origin, it can be shown that $\alpha$ can be decomposed as
\beq \label{eq: decomposition}
    \alpha=\alpha^{L} \, \alpha_0 \, \alpha^{R}
\eeq
where $\alpha^{R}$ (resp. $\alpha^{L}$) is an operation of local operators supported on $[0, \infty)$ (resp. $(-\infty, 0)$), and $\alpha_0$ is the conjugation by a local unitary. Although $\alpha$ is a group homomorphism, in general $\alpha^{R}$ is not. In fact, for any $g, h\in G$,
\beq\label{eq:composition_half_chain}
\alpha^{R}_{g} \, \alpha^{R}_{h}=\Ad_{V_{g,h}} \, \alpha^{R}_{gh}
\eeq
where $V:G\times G\to \cU^{\ell}$ with $\cU^\ell$ the group of local unitaries is \textbf{not} necessarily a homomorphism, and $\Ad_V(a):=VaV^{*}$ for any $a\in\A$. The associativity of $\alpha^{R}$, \ie $\left( \alpha^{R}_{g} \, \alpha^{R}_{h} \right) \, \alpha^{R}_{k} = \alpha^{R}_{g} \, \left( \alpha^{R}_{h} \, \alpha^{R}_{k} \right)$ with $g, h, k\in G$, puts further constraints on $V$
\beq\label{eq:anomaly_index}
\Ad_{\omega_{g,h,k}}=1,\quad
    \omega_{g,h,k}
    :=V_{g,h}V_{gh,k}V_{g,hk}^{-1}\alpha^{R}_{g}(V_{h,k})^{-1}
\eeq
This means the above $\omega$ is actually a phase since it commutes with all local operators. It can be checked that $\omega$ satisfies the 3-cocycle condition, and multiplying $V_{g,h}$ by a phase $\rho_{g,h}\in \U$ shifts $\omega$ by a 3-coboundary $\delta\rho$. Therefore, $\omega$ specifies an element in $\rH^3(G, \U)$, and this element is defined as the anomaly index associated with the symmetry action $\alpha$, see appendix \ref{sec:group_cohomology} for a review of group cohomology. See also Ref.~\cite{Else_2014} for more detailed discussions.

If $\alpha$ contains translation, one can stack the system with another copy on which the translation acts oppositely \cite{kapustin2024anomalous}. The symmetry action on this composite system (denoted by $\alpha_{\otimes}$) contains no translation, and the anomaly index of $\alpha$ is defined to be the anomaly index of $\alpha_{\otimes}$.

We comment on some possible generalizations of the anomaly index. Firstly, one can allow tails by considering the locality-preserving automorphisms introduced in Ref.~\cite{Ranard_2022,kapustin2024anomalous}. In more details,
\begin{definition}\label{def:LPA}
    An automorphism $\alpha\in\Aut(\A)$ is called a locality-preserving automorphism (LPA) if there exists a non-negative decreasing function $f_{\alpha}\searrow 0$ such that for any local operator $x\in\A_{X}$, there exists $x^{(r)}\in\A_{B(X,r)}$ such that
    \beq
    \|\alpha(x)-x^{(r)}\|\leqslant f_{\alpha}(r)\|x\|
    \eeq
    The function $f_{\alpha}$ is called the tail of $\alpha$. The group of LPA is denoted by $\G^{lp}$.
\end{definition}
Clearly, QCA's are special cases of LPA's, corresponding to the strictly local situation in which the tail function satisfies $f_{\alpha}(r)=0$, $\text{for all } r>r_{\alpha}$.
For this reason, we work exclusively with LPAs in the remainder of this paper.

When $G$ is a Lie group, it is natural to require the symmetry action to be smooth while $\G^{lp}$ does not carry a smooth structure. This requirement can be met by restricting to \emph{almost-local} LPAs, namely those whose tail functions satisfy $f_{\alpha}(r)=O(r^{-\infty})$.
The subgroup of almost local LPA can be equipped with a smooth structure, see Ref.~\cite{kapustin2024anomalous} for a detailed discussion. Throughout this work, whenever $G$ is a Lie group, we implicitly assume we are using almost local LPA without further comment.

In all of the above generalizations, the anomaly indices are defined in a similar manner.

\bibliography{lib}

@article{Alberti1983,
  journal={Lett. Math. Phys.},
  pages={25},
  year={1983},
  volume={7},
  doi={10.1007/BF00398708},
  title={{A note on the transition probability over C*-algebras}},
  author={Alberti, Peter M.}
}

@article{Weinstein2025,
  title = {Efficient Detection of Strong-to-Weak Spontaneous Symmetry Breaking via the R\'enyi-1 Correlator},
  author = {Weinstein, Zack},
  journal = {Phys. Rev. Lett.},
  volume = {134},
  issue = {15},
  pages = {150405},
  numpages = {7},
  year = {2025},
  month = {Apr},
  publisher = {American Physical Society},
  doi = {10.1103/PhysRevLett.134.150405},
  url = {https://link.aps.org/doi/10.1103/PhysRevLett.134.150405}
}

@article{Sang_2404,
  archiveprefix={arXiv},
  eprint={2404.07251},
  journal={Phys. Rev. Lett.},
  pages={070403},
  year={2025},
  volume={134},
  doi={10.1103/PhysRevLett.134.070403},
  title={Stability of Mixed-State Quantum Phases via Finite Markov Length},
  author={Sang, Shengqi and Hsieh, Timothy H.}
}

@ARTICLE{Tu2025,
       author = {{Tu}, Yi-Ting and {Long}, David M. and {Else}, Dominic V.},
        title = "{Anomalies of global symmetries on the lattice}",
      journal = {arXiv e-prints},
         year = 2025,
        month = jul,
          eid = {arXiv:2507.21209},
        pages = {arXiv:2507.21209},
          doi = {10.48550/arXiv.2507.21209},
archivePrefix = {arXiv},
       eprint = {2507.21209},
 primaryClass = {cond-mat.str-el},
       adsurl = {https://ui.adsabs.harvard.edu/abs/2025arXiv250721209T}
}

@ARTICLE{Kawagoe2025,
       author = {{Kawagoe}, Kyle and {Shirley}, Wilbur},
        title = "{Anomaly diagnosis via symmetry restriction in two-dimensional lattice systems}",
      journal = {arXiv e-prints},
         year = 2025,
        month = jul,
          eid = {arXiv:2507.07430},
        pages = {arXiv:2507.07430},
          doi = {10.48550/arXiv.2507.07430},
archivePrefix = {arXiv},
       eprint = {2507.07430},
 primaryClass = {cond-mat.str-el},
       adsurl = {https://ui.adsabs.harvard.edu/abs/2025arXiv250707430K}
}

@ARTICLE{Aksoy2023,
       author = {{Aksoy}, {\"O}mer Mert and {Mudry}, Christopher and {Furusaki}, Akira and {Tiwari}, Apoorv},
        title = "{Lieb-Schultz-Mattis anomalies and web of dualities induced by gauging in quantum spin chains}",
      journal = {SciPost Physics},
         year = 2024,
        month = jan,
       volume = {16},
       number = {1},
          eid = {022},
        pages = {022},
          doi = {10.21468/SciPostPhys.16.1.022},
archivePrefix = {arXiv},
       eprint = {2308.00743},
 primaryClass = {cond-mat.str-el},
       adsurl = {https://ui.adsabs.harvard.edu/abs/2024ScPP...16...22A}
}

@ARTICLE{Aksoy2021,
       author = {{Aksoy}, {\"O}mer M. and {Tiwari}, Apoorv and {Mudry}, Christopher},
        title = "{Lieb-Schultz-Mattis type theorems for Majorana models with discrete symmetries}",
      journal = {\prb},
         year = 2021,
        month = aug,
       volume = {104},
       number = {7},
          eid = {075146},
        pages = {075146},
          doi = {10.1103/PhysRevB.104.075146},
archivePrefix = {arXiv},
       eprint = {2102.08389},
 primaryClass = {cond-mat.str-el},
       adsurl = {https://ui.adsabs.harvard.edu/abs/2021PhRvB.104g5146A}
}

@ARTICLE{Pace2024,
       author = {{Pace}, Salvatore D. and {Lam}, Ho Tat and {Aksoy}, {\"O}mer Mert},
        title = "{(SPT-)LSM theorems from projective non-invertible symmetries}",
      journal = {SciPost Physics},
         year = 2025,
        month = jan,
       volume = {18},
       number = {1},
          eid = {028},
        pages = {028},
          doi = {10.21468/SciPostPhys.18.1.028},
archivePrefix = {arXiv},
       eprint = {2409.18113},
 primaryClass = {cond-mat.str-el},
       adsurl = {https://ui.adsabs.harvard.edu/abs/2025ScPP...18...28P}
}

@ARTICLE{Seifnashri2023,
       author = {{Seifnashri}, Sahand},
        title = "{Lieb-Schultz-Mattis anomalies as obstructions to gauging (non-on-site) symmetries}",
      journal = {SciPost Physics},
         year = 2024,
        month = apr,
       volume = {16},
       number = {4},
          eid = {098},
        pages = {098},
          doi = {10.21468/SciPostPhys.16.4.098},
archivePrefix = {arXiv},
       eprint = {2308.05151},
 primaryClass = {cond-mat.str-el},
       adsurl = {https://ui.adsabs.harvard.edu/abs/2024ScPP...16...98S}
}

@ARTICLE{Cho2017,
       author = {{Cho}, Gil Young and {Hsieh}, Chang-Tse and {Ryu}, Shinsei},
        title = "{Anomaly manifestation of Lieb-Schultz-Mattis theorem and topological phases}",
      journal = {\prb},
         year = 2017,
        month = nov,
       volume = {96},
       number = {19},
          eid = {195105},
        pages = {195105},
          doi = {10.1103/PhysRevB.96.195105},
archivePrefix = {arXiv},
       eprint = {1705.03892},
 primaryClass = {cond-mat.str-el},
       adsurl = {https://ui.adsabs.harvard.edu/abs/2017PhRvB..96s5105C}
}

@ARTICLE{kapustin2024anomalous,
       author = {{Kapustin}, Anton and {Sopenko}, Nikita},
        title = "{Anomalous symmetries of quantum spin chains and a generalization of the Lieb-Schultz-Mattis theorem}",
      journal = {arXiv e-prints},
         year = 2024,
        month = jan,
          eid = {arXiv:2401.02533},
        pages = {arXiv:2401.02533},
          doi = {10.48550/arXiv.2401.02533},
archivePrefix = {arXiv},
       eprint = {2401.02533},
 primaryClass = {math-ph},
       adsurl = {https://ui.adsabs.harvard.edu/abs/2024arXiv240102533K}
}

@book{brown2012cohomology,
  title={Cohomology of Groups},
  author={Brown, K.S.},
  isbn={9781468493276},
  lccn={82000733},
  series={Graduate Texts in Mathematics},
  url={https://books.google.ca/books?id=2fzlBwAAQBAJ},
  year={2012},
  publisher={Springer New York}
}

@ARTICLE{Chen2010,
       author = {{Chen}, Xie and {Gu}, Zheng-Cheng and {Wen}, Xiao-Gang},
        title = "{Classification of gapped symmetric phases in one-dimensional spin systems}",
      journal = {\prb},
         year = 2011,
        month = jan,
       volume = {83},
       number = {3},
          eid = {035107},
        pages = {035107},
          doi = {10.1103/PhysRevB.83.035107},
archivePrefix = {arXiv},
       eprint = {1008.3745},
 primaryClass = {cond-mat.str-el},
       adsurl = {https://ui.adsabs.harvard.edu/abs/2011PhRvB..83c5107C}
}

@ARTICLE{Hastings2005,
       author = {{Hastings}, M.~B. and {Wen}, Xiao-Gang},
        title = "{Quasiadiabatic continuation of quantum states: The stability of topological ground-state degeneracy and emergent gauge invariance}",
      journal = {\prb},
         year = 2005,
        month = jul,
       volume = {72},
       number = {4},
          eid = {045141},
        pages = {045141},
          doi = {10.1103/PhysRevB.72.045141},
archivePrefix = {arXiv},
       eprint = {cond-mat/0503554},
 primaryClass = {cond-mat.str-el},
       adsurl = {https://ui.adsabs.harvard.edu/abs/2005PhRvB..72d5141H}
}

@ARTICLE{Chen2013,
       author = {{Chen}, Xie and {Gu}, Zheng-Cheng and {Liu}, Zheng-Xin and
         {Wen}, Xiao-Gang},
        title = "{Symmetry protected topological orders and the group cohomology of their symmetry group}",
      journal = {\prb},
         year = 2013,
        month = apr,
       volume = {87},
       number = {15},
          eid = {155114},
        pages = {155114},
          doi = {10.1103/PhysRevB.87.155114},
archivePrefix = {arXiv},
       eprint = {1106.4772},
 primaryClass = {cond-mat.str-el},
       adsurl = {https://ui.adsabs.harvard.edu/abs/2013PhRvB..87o5114C}
}

@ARTICLE{Gaiotto2017,
       author = {{Gaiotto}, Davide and {Johnson-Freyd}, Theo},
        title = "{Symmetry protected topological phases and generalized cohomology}",
      journal = {Journal of High Energy Physics},
         year = 2019,
        month = may,
       volume = {2019},
       number = {5},
          eid = {7},
        pages = {7},
          doi = {10.1007/JHEP05(2019)007},
archivePrefix = {arXiv},
       eprint = {1712.07950},
 primaryClass = {hep-th},
       adsurl = {https://ui.adsabs.harvard.edu/abs/2019JHEP...05..007G}
}

@ARTICLE{Cheng2018a,
       author = {{Cheng}, Meng},
        title = "{Fermionic Lieb-Schultz-Mattis theorems and weak symmetry-protected phases}",
      journal = {\prb},
         year = 2019,
        month = feb,
       volume = {99},
       number = {7},
          eid = {075143},
        pages = {075143},
          doi = {10.1103/PhysRevB.99.075143},
archivePrefix = {arXiv},
       eprint = {1804.10122},
 primaryClass = {cond-mat.str-el},
       adsurl = {https://ui.adsabs.harvard.edu/abs/2019PhRvB..99g5143C}
}

@ARTICLE{Metlitski2018,
       author = {{Metlitski}, Max A. and {Thorngren}, Ryan},
        title = "{Intrinsic and emergent anomalies at deconfined critical points}",
      journal = {\prb},
         year = 2018,
        month = aug,
       volume = {98},
       number = {8},
          eid = {085140},
        pages = {085140},
          doi = {10.1103/PhysRevB.98.085140},
archivePrefix = {arXiv},
       eprint = {1707.07686},
 primaryClass = {cond-mat.str-el},
       adsurl = {https://ui.adsabs.harvard.edu/abs/2018PhRvB..98h5140M}
}

@ARTICLE{Cheng2015,
       author = {{Cheng}, Meng and {Zaletel}, Michael and {Barkeshli}, Maissam and
         {Vishwanath}, Ashvin and {Bonderson}, Parsa},
        title = "{Translational Symmetry and Microscopic Constraints on Symmetry-Enriched Topological Phases: A View from the Surface}",
      journal = {Physical Review X},
         year = 2016,
        month = dec,
       volume = {6},
       number = {4},
          eid = {041068},
        pages = {041068},
          doi = {10.1103/PhysRevX.6.041068},
archivePrefix = {arXiv},
       eprint = {1511.02263},
 primaryClass = {cond-mat.str-el},
       adsurl = {https://ui.adsabs.harvard.edu/abs/2016PhRvX...6d1068C}
}

@ARTICLE{Hastings2003,
       author = {{Hastings}, M.~B.},
        title = "{Lieb-Schultz-Mattis in higher dimensions}",
      journal = {\prb},
         year = "2004",
        month = "Mar",
       volume = {69},
       number = {10},
          eid = {104431},
        pages = {104431},
          doi = {10.1103/PhysRevB.69.104431},
archivePrefix = {arXiv},
       eprint = {cond-mat/0305505},
 primaryClass = {cond-mat.str-el},
       adsurl = {https://ui.adsabs.harvard.edu/abs/2004PhRvB..69j4431H}
}

@ARTICLE{Oshikawa1999,
       author = {{Oshikawa}, Masaki},
        title = "{Commensurability, Excitation Gap, and Topology in Quantum Many-Particle Systems on a Periodic Lattice}",
      journal = {\prl},
         year = "2000",
        month = "Feb",
       volume = {84},
       number = {7},
        pages = {1535-1538},
          doi = {10.1103/PhysRevLett.84.1535},
archivePrefix = {arXiv},
       eprint = {cond-mat/9911137},
 primaryClass = {cond-mat.str-el},
       adsurl = {https://ui.adsabs.harvard.edu/abs/2000PhRvL..84.1535O}
}

@article{Lieb1961,
title = "Two soluble models of an antiferromagnetic chain",
journal = "Annals of Physics",
volume = "16",
number = "3",
pages = "407 - 466",
year = "1961",
issn = "0003-4916",
doi = "https://doi.org/10.1016/0003-4916(61)90115-4",
url = "http://www.sciencedirect.com/science/article/pii/0003491661901154",
author = "Elliott Lieb and Theodore Schultz and Daniel Mattis"
}

@ARTICLE{Po2017,
       author = {{Po}, Hoi Chun and {Watanabe}, Haruki and {Jian}, Chao-Ming and {Zaletel}, Michael P.},
        title = "{Lattice Homotopy Constraints on Phases of Quantum Magnets}",
      journal = {\prl},
         year = 2017,
        month = sep,
       volume = {119},
       number = {12},
          eid = {127202},
        pages = {127202},
          doi = {10.1103/PhysRevLett.119.127202},
archivePrefix = {arXiv},
       eprint = {1703.06882},
 primaryClass = {cond-mat.str-el},
       adsurl = {https://ui.adsabs.harvard.edu/abs/2017PhRvL.119l7202P}
}

@ARTICLE{Else2020,
       author = {{Else}, Dominic V. and {Thorngren}, Ryan},
        title = "{Topological theory of Lieb-Schultz-Mattis theorems in quantum spin systems}",
      journal = {\prb},
         year = 2020,
        month = jun,
       volume = {101},
       number = {22},
          eid = {224437},
        pages = {224437},
          doi = {10.1103/PhysRevB.101.224437},
archivePrefix = {arXiv},
       eprint = {1907.08204},
 primaryClass = {cond-mat.str-el},
       adsurl = {https://ui.adsabs.harvard.edu/abs/2020PhRvB.101v4437E}
}

@ARTICLE{Ye2021a,
	author = {{Ye}, Weicheng and {Guo}, Meng and {He}, Yin-Chen and {Wang}, Chong and {Zou}, Liujun},
	title = "{Topological characterization of Lieb-Schultz-Mattis constraints and applications to symmetry-enriched quantum criticality}",
	journal = {SciPost Physics},
	year = 2022,
	month = sep,
	volume = {13},
	number = {3},
	eid = {066},
	pages = {066},
	doi = {10.21468/SciPostPhys.13.3.066},
	archivePrefix = {arXiv},
	eprint = {2111.12097},
	primaryClass = {cond-mat.str-el},
	adsurl = {https://ui.adsabs.harvard.edu/abs/2021arXiv211112097Y}
}

@book{weinberg2005quantum,
  title={The Quantum Theory of Fields: Volume 1, Foundations},
  author={Weinberg, S.},
  isbn={9780521670531},
  url={https://books.google.ca/books?id=MSdywQEACAAJ},
  year={2005},
  publisher={Cambridge University Press}
}

@ARTICLE{Kobayashi2018,
       author = {{Kobayashi}, Ryohei and {Shiozaki}, Ken and {Kikuchi}, Yuta and {Ryu}, Shinsei},
        title = "{Lieb-Schultz-Mattis type theorem with higher-form symmetry and the quantum dimer models}",
      journal = {\prb},
         year = 2019,
        month = jan,
       volume = {99},
       number = {1},
          eid = {014402},
        pages = {014402},
          doi = {10.1103/PhysRevB.99.014402},
archivePrefix = {arXiv},
       eprint = {1805.05367},
 primaryClass = {cond-mat.stat-mech},
       adsurl = {https://ui.adsabs.harvard.edu/abs/2019PhRvB..99a4402K}
}

@ARTICLE{Cheng2022,
       author = {{Cheng}, Meng and {Seiberg}, Nathan},
        title = "{Lieb-Schultz-Mattis, Luttinger, and 't Hooft - anomaly matching in lattice systems}",
      journal = {SciPost Physics},
         year = 2023,
        month = aug,
       volume = {15},
       number = {2},
          eid = {051},
        pages = {051},
          doi = {10.21468/SciPostPhys.15.2.051},
archivePrefix = {arXiv},
       eprint = {2211.12543},
 primaryClass = {cond-mat.str-el},
       adsurl = {https://ui.adsabs.harvard.edu/abs/2023ScPP...15...51C}
}

@Article{Jiang2019,
	title={{Generalized Lieb-Schultz-Mattis theorem on bosonic symmetry protected topological phases}},
	author={Shenghan Jiang and Meng Cheng and Yang Qi and Yuan-Ming Lu},
	journal={SciPost Phys.},
	volume={11},
	pages={024},
	year={2021},
	publisher={SciPost},
	doi={10.21468/SciPostPhys.11.2.024},
	url={https://scipost.org/10.21468/SciPostPhys.11.2.024},
archivePrefix = {arXiv},
       eprint = {1907.08596},
 primaryClass = {cond-mat.str-el},
       adsurl = {https://ui.adsabs.harvard.edu/abs/2019arXiv190708596J}
}

@article{Li:2026qhc,
    author = "Li, Zhi and Firanko, Raz and Hsieh, Timothy H.",
    title = "{A Unified Framework for Locally Stable Phases}",
    eprint = "2605.00088",
    archivePrefix = "arXiv",
    primaryClass = "quant-ph",
    month = "4",
    year = "2026"
}

@article{Yi_2026,
   title={Universal Decay of Mutual Information and Conditional Mutual Information in Gapped Pure- and Mixed-State Quantum Matter},
   volume={136},
   ISSN={1079-7114},
   url={http://dx.doi.org/10.1103/mqp8-y1m7},
   DOI={10.1103/mqp8-y1m7},
   number={11},
   journal={Physical Review Letters},
   publisher={American Physical Society (APS)},
   author={Yi, Jinmin and Li, Kangle and Liu, Chuan and Li, Zixuan and Zou, Liujun},
   year={2026},
   month=Mar }

@article{Huang_2407,
  archiveprefix={arXiv},
  eprint={2407.08760},
  journal={Phys. Rev. B},
  pages={125147},
  year={2025},
  volume={111},
  doi={10.1103/PhysRevB.111.125147},
  title={{Hydrodynamics as the effective field theory of strong-to-weak spontaneous symmetry breaking}},
  author={Huang, Xiaoyang and Qi, Marvin and Zhang, Jian-Hao and Lucas, Andrew}
}

@article{matsui2011boundedness,
author = {Matsui, Taku},
title = {BOUNDEDNESS OF ENTANGLEMENT ENTROPY AND SPLIT PROPERTY OF QUANTUM SPIN CHAINS},
journal = {Reviews in Mathematical Physics},
volume = {25},
number = {09},
pages = {1350017},
year = {2013},
doi = {10.1142/S0129055X13500177},

URL = { 
    
        https://doi.org/10.1142/S0129055X13500177
    
    

},
archivePrefix = {arXiv},
       eprint = {1109.5778},
 primaryClass = {math-ph},
       adsurl = {https://ui.adsabs.harvard.edu/abs/2011arXiv1109.5778M}
}

@ARTICLE{Tasaki2022topological,
       author = {{Tasaki}, Hal},
        title = "{The Lieb-Schultz-Mattis Theorem: A Topological Point of View}",
      journal = {arXiv e-prints},
         year = 2022,
        month = feb,
          eid = {arXiv:2202.06243},
        pages = {arXiv:2202.06243},
          doi = {10.48550/arXiv.2202.06243},
archivePrefix = {arXiv},
       eprint = {2202.06243},
 primaryClass = {cond-mat.stat-mech},
       adsurl = {https://ui.adsabs.harvard.edu/abs/2022arXiv220206243T}
}

@ARTICLE{Ogata_2021,
       author = {{Ogata}, Yoshiko and {Tachikawa}, Yuji and {Tasaki}, Hal},
        title = "{General Lieb-Schultz-Mattis Type Theorems for Quantum Spin Chains}",
      journal = {Communications in Mathematical Physics},
         year = 2021,
        month = jul,
       volume = {385},
       number = {1},
        pages = {79-99},
          doi = {10.1007/s00220-021-04116-9},
archivePrefix = {arXiv},
       eprint = {2004.06458},
 primaryClass = {math-ph},
       adsurl = {https://ui.adsabs.harvard.edu/abs/2021CMaPh.385...79O}
}

@ARTICLE{brylinski2000differentiable,
       author = {{Brylinski}, Jean-Luc},
        title = "{Differentiable Cohomology of Gauge Groups}",
      journal = {arXiv Mathematics e-prints},
         year = 2000,
        month = nov,
          eid = {math/0011069},
        pages = {math/0011069},
          doi = {10.48550/arXiv.math/0011069},
archivePrefix = {arXiv},
       eprint = {math/0011069},
 primaryClass = {math.DG},
       adsurl = {https://ui.adsabs.harvard.edu/abs/2000math.....11069B}
}

@inbook{Weibel_1994_group, place={Cambridge}, series={Cambridge Studies in Advanced Mathematics}, title={Group Homology and Cohomology}, booktitle={An Introduction to Homological Algebra}, publisher={Cambridge University Press}, author={Weibel, Charles A.}, year={1994}, pages={160–215}, collection={Cambridge Studies in Advanced Mathematics}}

@ARTICLE{Kawabata2023,
       author = {{Kawabata}, Kohei and {Sohal}, Ramanjit and {Ryu}, Shinsei},
        title = "{Lieb-Schultz-Mattis Theorem in Open Quantum Systems}",
      journal = {\prl},
         year = 2024,
        month = feb,
       volume = {132},
       number = {7},
          eid = {070402},
        pages = {070402},
          doi = {10.1103/PhysRevLett.132.070402},
archivePrefix = {arXiv},
       eprint = {2305.16496},
 primaryClass = {cond-mat.stat-mech},
       adsurl = {https://ui.adsabs.harvard.edu/abs/2024PhRvL.132g0402K}
}

@ARTICLE{Zhou2023,
       author = {{Zhou}, Yi-Neng and {Li}, Xingyu and {Zhai}, Hui and {Li}, Chengshu and {Gu}, Yingfei},
        title = "{Reviving the Lieb-Schultz-Mattis Theorem in Open Quantum Systems}",
      journal = {arXiv e-prints},
         year = 2023,
        month = oct,
          eid = {arXiv:2310.01475},
        pages = {arXiv:2310.01475},
          doi = {10.48550/arXiv.2310.01475},
archivePrefix = {arXiv},
       eprint = {2310.01475},
 primaryClass = {cond-mat.str-el},
       adsurl = {https://ui.adsabs.harvard.edu/abs/2023arXiv231001475Z}
}

@ARTICLE{Ma2022a,
       author = {{Ma}, Ruochen and {Wang}, Chong},
        title = "{Average Symmetry-Protected Topological Phases}",
      journal = {Physical Review X},
         year = 2023,
        month = jul,
       volume = {13},
       number = {3},
          eid = {031016},
        pages = {031016},
          doi = {10.1103/PhysRevX.13.031016},
archivePrefix = {arXiv},
       eprint = {2209.02723},
 primaryClass = {cond-mat.str-el},
       adsurl = {https://ui.adsabs.harvard.edu/abs/2023PhRvX..13c1016M}
}

@ARTICLE{Jian2017,
       author = {{Jian}, Chao-Ming and {Bi}, Zhen and {Xu}, Cenke},
        title = "{Lieb-Schultz-Mattis theorem and its generalizations from the perspective of the symmetry-protected topological phase}",
      journal = {\prb},
         year = 2018,
        month = feb,
       volume = {97},
       number = {5},
          eid = {054412},
        pages = {054412},
          doi = {10.1103/PhysRevB.97.054412},
archivePrefix = {arXiv},
       eprint = {1705.00012},
 primaryClass = {cond-mat.str-el},
       adsurl = {https://ui.adsabs.harvard.edu/abs/2018PhRvB..97e4412J}
}

@book{bratteli2013operator1,
  title={Operator Algebras and Quantum Statistical Mechanics 1: C*- and W*-Algebras. Symmetry Groups. Decomposition of States},
  author={Bratteli, O. and Robinson, D.W.},
  isbn={9783662025215},
  series={Theoretical and Mathematical Physics},
  url={https://books.google.ca/books?id=3R8nswEACAAJ},
  year={2013},
  publisher={Springer Berlin Heidelberg}
}

@book{bratteli2013operator2,
  title={Operator Algebras and Quantum Statistical Mechanics 2: Equilibrium States. Models in Quantum Statistical Mechanics},
  author={Bratteli, O. and Robinson, D.W.},
  isbn={9783662034446},
  series={Theoretical and Mathematical Physics},
  url={https://books.google.ca/books?id=SV3oCAAAQBAJ},
  year={2013},
  publisher={Springer Berlin Heidelberg}
}

@book{Naaijkens_2017,
   title={Quantum Spin Systems on Infinite Lattices: A Concise Introduction},
   ISBN={9783319514581},
   ISSN={1616-6361},
   url={http://dx.doi.org/10.1007/978-3-319-51458-1},
   DOI={10.1007/978-3-319-51458-1},
archivePrefix = {arXiv},
       eprint = {1311.2717},
 primaryClass = {math-ph},
   journal={Lecture Notes in Physics},
   publisher={Springer International Publishing},
   author={Naaijkens, Pieter},
   year={2017} }

@Inbook{Tao2022,
author="Tao, Terence",
title="The Real Numbers",
bookTitle="Analysis I",
year="2022",
publisher="Springer Nature Singapore",
address="Singapore",
pages="81--107",
isbn="978-981-19-7261-4",
doi="10.1007/978-981-19-7261-4_5",
url="https://doi.org/10.1007/978-981-19-7261-4_5"
}

@book{Landsman:2017hpa,
    author = "Landsman, Klaas",
    title = "{Foundations of Quantum Theory}: {From Classical Concepts to Operator Algebras}",
    doi = "10.1007/978-3-319-51777-3",
    isbn = "978-3-319-51777-3",
    publisher = "Springer International Publishing",
    address = "Cham",
    year = "2017"
}

@book{murphy2014c,
  title={C*-Algebras and Operator Theory},
  author={Murphy, G.J.},
  isbn={9780080924960},
  url={https://books.google.ca/books?id=omviBQAAQBAJ},
  year={2014},
  publisher={Elsevier Science}
}

@ARTICLE{Lessa:2024wcw,
       author = {{Lessa}, Leonardo A. and {Cheng}, Meng and {Wang}, Chong},
        title = "{Mixed-state quantum anomaly and multipartite entanglement}",
      journal = {arXiv e-prints},
         year = 2024,
        month = jan,
          eid = {arXiv:2401.17357},
        pages = {arXiv:2401.17357},
          doi = {10.48550/arXiv.2401.17357},
archivePrefix = {arXiv},
       eprint = {2401.17357},
 primaryClass = {cond-mat.str-el},
       adsurl = {https://ui.adsabs.harvard.edu/abs/2024arXiv240117357L}
}

@ARTICLE{Else_2014,
       author = {{Else}, Dominic V. and {Nayak}, Chetan},
        title = "{Classifying symmetry-protected topological phases through the anomalous action of the symmetry on the edge}",
      journal = {\prb},
         year = 2014,
        month = dec,
       volume = {90},
       number = {23},
          eid = {235137},
        pages = {235137},
          doi = {10.1103/PhysRevB.90.235137},
archivePrefix = {arXiv},
       eprint = {1409.5436},
 primaryClass = {cond-mat.str-el},
       adsurl = {https://ui.adsabs.harvard.edu/abs/2014PhRvB..90w5137E}
}

@article{Farrelly_2020,
   title={A review of Quantum Cellular Automata},
   volume={4},
   ISSN={2521-327X},
   url={http://dx.doi.org/10.22331/q-2020-11-30-368},
   DOI={10.22331/q-2020-11-30-368},
   journal={Quantum},
   publisher={Verein zur Forderung des Open Access Publizierens in den Quantenwissenschaften},
   author={Farrelly, Terry},
   year={2020},
   month=nov, pages={368},
archivePrefix = {arXiv},
       eprint = {1904.13318},
 primaryClass = {quant-ph},
       adsurl = {https://ui.adsabs.harvard.edu/abs/2020Quant...4..368F} }

@ARTICLE{Son_2011,
       author = {{Son}, W. and {Amico}, L. and {Vedral}, V.},
        title = "{Topological order in 1D Cluster state protected by symmetry}",
      journal = {Quantum Information Processing},
         year = 2012,
        month = dec,
       volume = {11},
       number = {6},
        pages = {1961-1968},
          doi = {10.1007/s11128-011-0346-7},
archivePrefix = {arXiv},
       eprint = {1111.7173},
 primaryClass = {quant-ph},
       adsurl = {https://ui.adsabs.harvard.edu/abs/2012QuIP...11.1961S}
}

@article{Gross_2012,
	author = {Gross, D. and Nesme, V. and Vogts, H. and Werner, R. F.},
	date = {2012/03/01},
	doi = {10.1007/s00220-012-1423-1},
	id = {Gross2012},
	isbn = {1432-0916},
	journal = {Communications in Mathematical Physics},
	number = {2},
	pages = {419--454},
	title = {Index Theory of One Dimensional Quantum Walks and Cellular Automata},
	url = {https://doi.org/10.1007/s00220-012-1423-1},
	volume = {310},
	year = {2012},
archivePrefix = {arXiv},
       eprint = {0910.3675},
 primaryClass = {quant-ph},
       adsurl = {https://ui.adsabs.harvard.edu/abs/2009arXiv0910.3675G}}

@inbook{Weibel_1994_simplicial, place={Cambridge}, series={Cambridge Studies in Advanced Mathematics}, title={Simplicial Methods in Homological Algebra}, booktitle={An Introduction to Homological Algebra}, publisher={Cambridge University Press}, author={Weibel, Charles A.}, year={1994}, pages={254–299}, collection={Cambridge Studies in Advanced Mathematics}}

@ARTICLE{Nachtergaele_2006,
       author = {{Nachtergaele}, Bruno and {Ogata}, Yoshiko and {Sims}, Robert},
        title = "{Propagation of Correlations in Quantum Lattice Systems}",
      journal = {Journal of Statistical Physics},
         year = 2006,
        month = jul,
       volume = {124},
       number = {1},
        pages = {1-13},
          doi = {10.1007/s10955-006-9143-6},
archivePrefix = {arXiv},
       eprint = {math-ph/0603064},
 primaryClass = {math-ph},
       adsurl = {https://ui.adsabs.harvard.edu/abs/2006JSP...124....1N}
}

@ARTICLE{Ranard_2022,
       author = {{Ranard}, Daniel and {Walter}, Michael and {Witteveen}, Freek},
        title = "{A Converse to Lieb-Robinson Bounds in One Dimension Using Index Theory}",
      journal = {Annales Henri Poincar\&eacute;},
         year = 2022,
        month = nov,
       volume = {23},
       number = {11},
        pages = {3905-3979},
          doi = {10.1007/s00023-022-01193-x},
archivePrefix = {arXiv},
       eprint = {2012.00741},
 primaryClass = {quant-ph},
       adsurl = {https://ui.adsabs.harvard.edu/abs/2022AnHP...23.3905R}
}

@ARTICLE{Kapustin2022Noether,
       author = {{Kapustin}, Anton and {Sopenko}, Nikita},
        title = "{Local Noether theorem for quantum lattice systems and topological invariants of gapped states}",
      journal = {Journal of Mathematical Physics},
         year = 2022,
        month = sep,
       volume = {63},
       number = {9},
          eid = {091903},
        pages = {091903},
          doi = {10.1063/5.0085964},
archivePrefix = {arXiv},
       eprint = {2201.01327},
 primaryClass = {math-ph},
       adsurl = {https://ui.adsabs.harvard.edu/abs/2022JMP....63i1903K}
}

@ARTICLE{Kapustin2020invertible,
       author = {{Kapustin}, Anton and {Sopenko}, Nikita and {Yang}, Bowen},
        title = "{A classification of invertible phases of bosonic quantum lattice systems in one dimension}",
      journal = {Journal of Mathematical Physics},
         year = 2021,
        month = aug,
       volume = {62},
       number = {8},
          eid = {081901},
        pages = {081901},
          doi = {10.1063/5.0055996},
archivePrefix = {arXiv},
       eprint = {2012.15491},
 primaryClass = {quant-ph},
       adsurl = {https://ui.adsabs.harvard.edu/abs/2021JMP....62h1901K}
}

@ARTICLE{Ogata2019split,
       author = {{Ogata}, Yoshiko},
        title = "{A classification of pure states on quantum spin chains satisfying the split property with on-site finite group symmetries}",
      journal = {arXiv e-prints},
         year = 2019,
        month = aug,
          eid = {arXiv:1908.08621},
        pages = {arXiv:1908.08621},
          doi = {10.48550/arXiv.1908.08621},
archivePrefix = {arXiv},
       eprint = {1908.08621},
 primaryClass = {math.OA},
       adsurl = {https://ui.adsabs.harvard.edu/abs/2019arXiv190808621O}
}

@ARTICLE{Garre_Rubio2024anomalous,
       author = {{Garre Rubio}, Jose and {Molnar}, Andras and {Ogata}, Yoshiko},
        title = "{Classifying symmetric and symmetry-broken spin chain phases with anomalous group actions}",
      journal = {arXiv e-prints},
         year = 2024,
        month = mar,
          eid = {arXiv:2403.18573},
        pages = {arXiv:2403.18573},
          doi = {10.48550/arXiv.2403.18573},
archivePrefix = {arXiv},
       eprint = {2403.18573},
 primaryClass = {quant-ph},
       adsurl = {https://ui.adsabs.harvard.edu/abs/2024arXiv240318573G}
}

@ARTICLE{Yao2021twisted,
       author = {{Yao}, Yuan and {Oshikawa}, Masaki},
        title = "{Twisted Boundary Condition and Lieb-Schultz-Mattis Ingappability for Discrete Symmetries}",
      journal = {\prl},
         year = 2021,
        month = may,
       volume = {126},
       number = {21},
          eid = {217201},
        pages = {217201},
          doi = {10.1103/PhysRevLett.126.217201},
archivePrefix = {arXiv},
       eprint = {2010.09244},
 primaryClass = {cond-mat.str-el},
       adsurl = {https://ui.adsabs.harvard.edu/abs/2021PhRvL.126u7201Y}
}

@ARTICLE{Ogata2019LSM,
       author = {{Ogata}, Yoshiko and {Tasaki}, Hal},
        title = "{Lieb-Schultz-Mattis Type Theorems for Quantum Spin Chains Without Continuous Symmetry}",
      journal = {Communications in Mathematical Physics},
         year = 2019,
        month = dec,
       volume = {372},
       number = {3},
        pages = {951-962},
          doi = {10.1007/s00220-019-03343-5},
archivePrefix = {arXiv},
       eprint = {1808.08740},
 primaryClass = {math-ph},
       adsurl = {https://ui.adsabs.harvard.edu/abs/2019CMaPh.372..951O}
}

@book{kadison1997fundamentals2,
  title={Fundamentals of the Theory of Operator Algebras. Volume II},
  author={Kadison, R.V. and Ringrose, J.R.},
  isbn={9780821808207},
  lccn={97020916},
  series={Fundamentals of the Theory of Operator Algebras},
  url={https://books.google.ca/books?id=h5bMkZTnowAC},
  year={1997},
  publisher={American Mathematical Society}
}

@ARTICLE{Bachmann2012automorphic,
       author = {{Bachmann}, Sven and {Michalakis}, Spyridon and {Nachtergaele}, Bruno and {Sims}, Robert},
        title = "{Automorphic Equivalence within Gapped Phases of Quantum Lattice Systems}",
      journal = {Communications in Mathematical Physics},
         year = 2012,
        month = feb,
       volume = {309},
       number = {3},
        pages = {835-871},
          doi = {10.1007/s00220-011-1380-0},
archivePrefix = {arXiv},
       eprint = {1102.0842},
 primaryClass = {math-ph},
       adsurl = {https://ui.adsabs.harvard.edu/abs/2012CMaPh.309..835B}
}

@book{Blackadar:2006OA,
    author = "Blackadar, Bruce",
    title = "{Operator algebras: Theory of C*-algebras and von Neumann algebras}",
    year = "2006"
}

@book{ohya2004quantum,
  title={Quantum Entropy and Its Use},
  author={Ohya, M. and Petz, D.},
  isbn={9783540208068},
  lccn={2004042924},
  series={Theoretical and Mathematical Physics},
  url={https://books.google.ca/books?id=r2ullNVyESQC},
  year={2004},
  publisher={Springer Berlin Heidelberg}
}

@article{Araki1976entropyI,
    author = "Araki, H.",
    title = "{Relative Entropy of States of Von Neumann Algebras}",
    journal = "Publ. Res. Inst. Math. Sci. Kyoto",
    volume = "1976",
    pages = "809--833",
    year = "1976"
}

@article{Araki1977entropyII,
    author = "Araki, H.",
    title = "{Relative Entropy of States of Von Neumann Algebras II}",
    journal = "Publ. Res. Inst. Math. Sci. Kyoto",
    volume = "1977",
    pages = "173-192",
    year = "1977"
}

@ARTICLE{Kapustin2025higher,
       author = {{Kapustin}, Anton},
        title = "{Higher symmetries and anomalies in quantum lattice systems}",
      journal = {arXiv e-prints},
         year = 2025,
        month = may,
          eid = {arXiv:2505.04719},
        pages = {arXiv:2505.04719},
          doi = {10.48550/arXiv.2505.04719},
archivePrefix = {arXiv},
       eprint = {2505.04719},
 primaryClass = {math-ph},
       adsurl = {https://ui.adsabs.harvard.edu/abs/2025arXiv250504719K}
}

@ARTICLE{Ogata2019reflection,
       author = {{Ogata}, Yoshiko},
        title = "{A ${\mathbb Z}_2$-index of symmetry protected topological phases with reflection symmetry for quantum spin chains}",
      journal = {arXiv e-prints},
         year = 2019,
        month = apr,
          eid = {arXiv:1904.01669},
        pages = {arXiv:1904.01669},
          doi = {10.48550/arXiv.1904.01669},
archivePrefix = {arXiv},
       eprint = {1904.01669},
 primaryClass = {math-ph},
       adsurl = {https://ui.adsabs.harvard.edu/abs/2019arXiv190401669O}
}

@ARTICLE{Liu2024LRLSM,
       author = {{Liu}, Ruizhi and {Yi}, Jinmin and {Zhou}, Shiyu and {Zou}, Liujun},
        title = "{Entanglement area law and Lieb-Schultz-Mattis theorem in long-range interacting systems, and symmetry-enforced long-range entanglement}",
      journal = {arXiv e-prints},
         year = 2024,
        month = may,
          eid = {arXiv:2405.14929},
        pages = {arXiv:2405.14929},
          doi = {10.48550/arXiv.2405.14929},
archivePrefix = {arXiv},
       eprint = {2405.14929},
 primaryClass = {cond-mat.str-el},
       adsurl = {https://ui.adsabs.harvard.edu/abs/2024arXiv240514929L},
      DOI={10.1103/2jgh-nrj1},
   number={21},
   journal={Physical Review B},
   publisher={American Physical Society (APS)},
   author={Liu, Ruizhi and Yi, Jinmin and Zhou, Shiyu and Zou, Liujun},
   year={2025}
}

@ARTICLE{Ogata2019TRS,
       author = {{Ogata}, Yoshiko},
        title = "{A Z$_{2}$-Index of Symmetry Protected Topological Phases with Time Reversal Symmetry for Quantum Spin Chains}",
      journal = {Communications in Mathematical Physics},
         year = 2019,
        month = jul,
       volume = {374},
       number = {2},
        pages = {705-734},
          doi = {10.1007/s00220-019-03521-5},
archivePrefix = {arXiv},
       eprint = {1810.01045},
 primaryClass = {math-ph},
       adsurl = {https://ui.adsabs.harvard.edu/abs/2019CMaPh.374..705O}
}

@ARTICLE{Watanabe2018LSM,
       author = {{Watanabe}, Haruki},
        title = "{Lieb-Schultz-Mattis-type filling constraints in the 1651 magnetic space groups}",
      journal = {\prb},
         year = 2018,
        month = apr,
       volume = {97},
       number = {16},
          eid = {165117},
        pages = {165117},
          doi = {10.1103/PhysRevB.97.165117},
archivePrefix = {arXiv},
       eprint = {1802.00587},
 primaryClass = {cond-mat.str-el},
       adsurl = {https://ui.adsabs.harvard.edu/abs/2018PhRvB..97p5117W}
}

@ARTICLE{Bourne2020fermion,
       author = {{Bourne}, Chris and {Ogata}, Yoshiko},
        title = "{The classification of symmetry protected topological phases of one-dimensional fermion systems}",
      journal = {arXiv e-prints},
         year = 2020,
        month = jun,
          eid = {arXiv:2006.15232},
        pages = {arXiv:2006.15232},
          doi = {10.48550/arXiv.2006.15232},
archivePrefix = {arXiv},
       eprint = {2006.15232},
 primaryClass = {math-ph},
       adsurl = {https://ui.adsabs.harvard.edu/abs/2020arXiv200615232B}
}

@ARTICLE{Griseta2024Z2_gradedC*,
       author = {{Griseta}, Maria Elena and {Zurlo}, Paola},
        title = "{$C^*$-independence for $\mathbb{Z}_2$-graded $C^*$-algebras}",
      journal = {arXiv e-prints},
         year = 2024,
        month = jan,
          eid = {arXiv:2401.08293},
        pages = {arXiv:2401.08293},
          doi = {10.48550/arXiv.2401.08293},
archivePrefix = {arXiv},
       eprint = {2401.08293},
 primaryClass = {math.OA},
       adsurl = {https://ui.adsabs.harvard.edu/abs/2024arXiv240108293G}
}

@misc{liu2025twistedlocalitypreservingautomorphismsanomaly,
      title={Twisted locality-preserving automorphisms, anomaly index, and generalized Lieb-Schultz-Mattis theorems with anti-unitary symmetries}, 
      author={Ruizhi Liu and Jinmin Yi and Liujun Zou},
      year={2025},
      eprint={2510.06555},
      archivePrefix={arXiv},
      primaryClass={cond-mat.str-el},
      url={https://arxiv.org/abs/2510.06555}, 
}

@misc{kapustin2025highersymmetriesanomaliesquantum,
      title={Higher symmetries and anomalies in quantum lattice systems}, 
      author={Anton Kapustin and Shixiong Xu},
      year={2025},
      eprint={2505.04719},
      archivePrefix={arXiv},
      primaryClass={math-ph},
      url={https://arxiv.org/abs/2505.04719}, 
}

@article{KISHIMOTO1996100,
title = {The Rohlin Property for Shifts on UHF Algebras and Automorphisms of Cuntz Algebras},
journal = {Journal of Functional Analysis},
volume = {140},
number = {1},
pages = {100-123},
year = {1996},
issn = {0022-1236},
doi = {https://doi.org/10.1006/jfan.1996.0100},
url = {https://www.sciencedirect.com/science/article/pii/S0022123696901007},
author = {Akitaka Kishimoto}
}

@misc{sun2014stronglyouterproducttype,
      title={Strongly outer product type actions}, 
      author={Michael Y. Sun},
      year={2014},
      eprint={1403.5357},
      archivePrefix={arXiv},
      primaryClass={math.OA},
      url={https://arxiv.org/abs/1403.5357}, 
}

@article{Woronowicz:1972fr,
    author = "Woronowicz, S. L.",
    title = "{On the purification of factor states}",
    doi = "10.1007/BF01645776",
    journal = "Commun. Math. Phys.",
    volume = "28",
    pages = "221--235",
    year = "1972"
}

@misc{wang2024anomalyopenquantumsystems,
      title={Anomaly in open quantum systems and its implications on mixed-state quantum phases}, 
      author={Zijian Wang and Linhao Li},
      year={2024},
      eprint={2403.14533},
      archivePrefix={arXiv},
      primaryClass={quant-ph},
      url={https://arxiv.org/abs/2403.14533}, 
}

@article{Woronowicz1973,
  author    = {S. L. Woronowicz},
  title     = {On the purification map},
  journal   = {Communications in Mathematical Physics},
  year      = {1973},
  volume    = {30},
  number    = {1},
  pages     = {55--67},
  doi       = {10.1007/BF01646688},
  url       = {https://doi.org/10.1007/BF01646688},
  issn      = {1432-0916}
}

@misc{sopenko2025reflectionpositivityrefinedindex,
      title={Reflection positivity and a refined index for 2d invertible phases}, 
      author={Nikita Sopenko},
      year={2025},
      eprint={2509.01711},
      archivePrefix={arXiv},
      primaryClass={math-ph},
      url={https://arxiv.org/abs/2509.01711}, 
}

@book{Azcárraga_Izquierdo_1995_cohomology, place={Cambridge}, series={Cambridge Monographs on Mathematical Physics}, title={Lie Groups, Lie Algebras, Cohomology and some Applications in Physics}, publisher={Cambridge University Press}, author={Azcárraga, Josi A. de and Izquierdo, Josi M.}, year={1995}, collection={Cambridge Monographs on Mathematical Physics}}

@article{vonNeumann1939,
author = {von Neumann, J.},
journal = {Compositio Mathematica},
pages = {1-77},
publisher = {Johnson Reprint Corporation},
title = {On infinite direct products},
url = {http://eudml.org/doc/88704},
volume = {6},
year = {1939},
}

@article{KEYL_2008,
   title={ON HAAG DUALITY FOR PURE STATES OF QUANTUM SPIN CHAINS},
   volume={20},
   ISSN={1793-6659},
   url={http://dx.doi.org/10.1142/S0129055X08003377},
   DOI={10.1142/s0129055x08003377},
   number={06},
   journal={Reviews in Mathematical Physics},
   publisher={World Scientific Pub Co Pte Lt},
   author={KEYL, M. and MATSUI, TAKU and SCHLINGEMANN, D. and WERNER, R. F.},
   year={2008},
   month=jul, pages={707–724} }

@article{KEYL_2006,
   title={ENTANGLEMENT, HAAG-DUALITY AND TYPE PROPERTIES OF INFINITE QUANTUM SPIN CHAINS},
   volume={18},
   ISSN={1793-6659},
   url={http://dx.doi.org/10.1142/S0129055X0600284X},
   DOI={10.1142/s0129055x0600284x},
   number={09},
   journal={Reviews in Mathematical Physics},
   publisher={World Scientific Pub Co Pte Lt},
   author={KEYL, M. and MATSUI, T. and SCHLINGEMANN, D. and WERNER, R. F.},
   year={2006},
   month=oct, pages={935–970} }

@misc{tao2009notes11,
  author = {Tao, Terence},
  title = {{245B}, Notes 11: The strong and weak topologies},
  year = {2009},
  month = feb,
  howpublished = {\url{https://terrytao.wordpress.com/2009/02/21/245b-notes-11-the-strong-and-weak-topologies/}},
  note = {Accessed: 2024-05-21}
}

@book{haag2012local,
  title={Local Quantum Physics: Fields, Particles, Algebras},
  author={Haag, R.},
  isbn={9783642614583},
  lccn={96018937},
  series={Theoretical and Mathematical Physics},
  url={https://books.google.ca/books?id=OlLmCAAAQBAJ},
  year={2012},
  publisher={Springer Berlin Heidelberg}
}

@article{Lessa_2025,
   title={Strong-to-Weak Spontaneous Symmetry Breaking in Mixed Quantum States},
   volume={6},
   ISSN={2691-3399},
   url={http://dx.doi.org/10.1103/PRXQuantum.6.010344},
   DOI={10.1103/prxquantum.6.010344},
   number={1},
   journal={PRX Quantum},
   publisher={American Physical Society (APS)},
   author={Lessa, Leonardo A. and Ma, Ruochen and Zhang, Jian-Hao and Bi, Zhen and Cheng, Meng and Wang, Chong},
   year={2025},
   month=mar }

@article{Ma_2025,
   title={Topological Phases with Average Symmetries: The Decohered, the Disordered, and the Intrinsic},
   volume={15},
   ISSN={2160-3308},
   url={http://dx.doi.org/10.1103/PhysRevX.15.021062},
   DOI={10.1103/physrevx.15.021062},
   number={2},
   journal={Physical Review X},
   publisher={American Physical Society (APS)},
   author={Ma, Ruochen and Zhang, Jian-Hao and Bi, Zhen and Cheng, Meng and Wang, Chong},
   year={2025},
   month=may }

@article{Lee_2023,
   title={Quantum Criticality Under Decoherence or Weak Measurement},
   volume={4},
   ISSN={2691-3399},
   url={http://dx.doi.org/10.1103/PRXQuantum.4.030317},
   DOI={10.1103/prxquantum.4.030317},
   number={3},
   journal={PRX Quantum},
   publisher={American Physical Society (APS)},
   author={Lee, Jong Yeon and Jian, Chao-Ming and Xu, Cenke},
   year={2023},
   month=aug }

@misc{rubio2024classifyingsymmetricsymmetrybrokenspin,
      title={Classifying symmetric and symmetry-broken spin chain phases with anomalous group actions}, 
      author={Jose Garre Rubio and Andras Molnar and Yoshiko Ogata},
      year={2024},
      eprint={2403.18573},
      archivePrefix={arXiv},
      primaryClass={quant-ph},
      url={https://arxiv.org/abs/2403.18573}, 
}

@misc{moon2019automorphicequivalencegappedphases,
      title={Automorphic equivalence within gapped phases in the bulk}, 
      author={Alvin Moon and Yoshiko Ogata},
      year={2019},
      eprint={1906.05479},
      archivePrefix={arXiv},
      primaryClass={math-ph},
      url={https://arxiv.org/abs/1906.05479}, 
}

@misc{deroeck2025puregappedgroundstates,
      title={Pure gapped ground states of spin chains are short-range entangled}, 
      author={Wojciech De Roeck and Martin Fraas and Bruno de O. Carvalho},
      year={2025},
      eprint={2511.14699},
      archivePrefix={arXiv},
      primaryClass={math-ph},
      url={https://arxiv.org/abs/2511.14699}, 
}

@misc{carvalho2024classificationsymmetryprotectedstates,
      title={Classification of symmetry protected states of quantum spin chains for continuous symmetry groups}, 
      author={Bruno de Oliveira Carvalho and Wojciech De Roeck and Tijl Jappens},
      year={2024},
      eprint={2409.01112},
      archivePrefix={arXiv},
      primaryClass={math-ph},
      url={https://arxiv.org/abs/2409.01112}, 
}

@article{Senthil_2015,
   title={Symmetry-Protected Topological Phases of Quantum Matter},
   volume={6},
   ISSN={1947-5462},
   url={http://dx.doi.org/10.1146/annurev-conmatphys-031214-014740},
   DOI={10.1146/annurev-conmatphys-031214-014740},
   number={1},
   journal={Annual Review of Condensed Matter Physics},
   publisher={Annual Reviews},
   author={Senthil, T.},
   year={2015},
   month=mar, pages={299–324} }

@article{Xiong_2018,
   title={Minimalist approach to the classification of symmetry protected topological phases},
   volume={51},
   ISSN={1751-8121},
   url={http://dx.doi.org/10.1088/1751-8121/aae0b1},
   DOI={10.1088/1751-8121/aae0b1},
   number={44},
   journal={Journal of Physics A: Mathematical and Theoretical},
   publisher={IOP Publishing},
   author={Xiong, Charles Zhaoxi},
   year={2018},
   month=oct, pages={445001} }

@misc{xue2024tensornetworkformulationsymmetry,
      title={Tensor network formulation of symmetry protected topological phases in mixed states}, 
      author={Hanyu Xue and Jong Yeon Lee and Yimu Bao},
      year={2024},
      eprint={2403.17069},
      archivePrefix={arXiv},
      primaryClass={cond-mat.str-el},
      url={https://arxiv.org/abs/2403.17069}, 
}

@article{de_Groot_2022,
   title={Symmetry Protected Topological Order in Open Quantum Systems},
   volume={6},
   ISSN={2521-327X},
   url={http://dx.doi.org/10.22331/q-2022-11-10-856},
   DOI={10.22331/q-2022-11-10-856},
   journal={Quantum},
   publisher={Verein zur Forderung des Open Access Publizierens in den Quantenwissenschaften},
   author={de Groot, Caroline and Turzillo, Alex and Schuch, Norbert},
   year={2022},
   month=nov, pages={856} }

@article{Buča_2012,
doi = {10.1088/1367-2630/14/7/073007},
url = {https://doi.org/10.1088/1367-2630/14/7/073007},
year = {2012},
month = {jul},
publisher = {IOP Publishing},
volume = {14},
number = {7},
pages = {073007},
author = {Buča, Berislav and Prosen, Tomaž},
title = {A note on symmetry reductions of the Lindblad equation: transport in constrained open spin chains},
journal = {New Journal of Physics}
}

@article{Lee_2025,
   title={Symmetry protected topological phases under decoherence},
   volume={9},
   ISSN={2521-327X},
   url={http://dx.doi.org/10.22331/q-2025-01-23-1607},
   DOI={10.22331/q-2025-01-23-1607},
   journal={Quantum},
   publisher={Verein zur Forderung des Open Access Publizierens in den Quantenwissenschaften},
   author={Lee, Jong Yeon and You, Yi-Zhuang and Xu, Cenke},
   year={2025},
   month=jan, pages={1607} }

@misc{sala2024spontaneousstrongsymmetrybreaking,
      title={Spontaneous Strong Symmetry Breaking in Open Systems: Purification Perspective}, 
      author={Pablo Sala and Sarang Gopalakrishnan and Masaki Oshikawa and Yizhi You},
      year={2024},
      eprint={2405.02402},
      archivePrefix={arXiv},
      primaryClass={quant-ph},
      url={https://arxiv.org/abs/2405.02402}, 
}

@article{Xu_2025,
   title={Average-exact mixed anomalies and compatible phases},
   volume={111},
   ISSN={2469-9969},
   url={http://dx.doi.org/10.1103/PhysRevB.111.125128},
   DOI={10.1103/physrevb.111.125128},
   number={12},
   journal={Physical Review B},
   publisher={American Physical Society (APS)},
   author={Xu, Yichen and Jian, Chao-Ming},
   year={2025},
   month=mar }

@misc{lu2024bilayerconstructionmixedstate,
      title={Bilayer construction for mixed state phenomena with strong, weak symmetries and symmetry breakings}, 
      author={Shuangyuan Lu and Penghao Zhu and Yuan-Ming Lu},
      year={2024},
      eprint={2411.07174},
      archivePrefix={arXiv},
      primaryClass={cond-mat.str-el},
      url={https://arxiv.org/abs/2411.07174}, 
}

@misc{sang2025mixedstatephaseslocalreversibility,
      title={Mixed-state phases from local reversibility}, 
      author={Shengqi Sang and Leonardo A. Lessa and Roger S. K. Mong and Tarun Grover and Chong Wang and Timothy H. Hsieh},
      year={2025},
      eprint={2507.02292},
      archivePrefix={arXiv},
      primaryClass={quant-ph},
      url={https://arxiv.org/abs/2507.02292}, 
}

@misc{ziereis2025strongtoweaksymmetrybreakingphases,
      title={Strong-to-Weak Symmetry Breaking Phases in Steady States of Quantum Operations}, 
      author={Niklas Ziereis and Sanjay Moudgalya and Michael Knap},
      year={2025},
      eprint={2509.09669},
      archivePrefix={arXiv},
      primaryClass={cond-mat.stat-mech},
      url={https://arxiv.org/abs/2509.09669}, 
}

@article{Chen:2024ofr,
    author = "Chen, Langxuan and Sun, Ning and Zhang, Pengfei",
    title = "{Strong-to-weak symmetry breaking and entanglement transitions}",
    eprint = "2411.05364",
    archivePrefix = "arXiv",
    primaryClass = "quant-ph",
    doi = "10.1103/PhysRevB.111.L060304",
    journal = "Phys. Rev. B",
    volume = "111",
    number = "6",
    pages = "L060304",
    year = "2025"
}

@misc{kawagoe2025anomalydiagnosissymmetryrestriction,
      title={Anomaly diagnosis via symmetry restriction in two-dimensional lattice systems}, 
      author={Kyle Kawagoe and Wilbur Shirley},
      year={2025},
      eprint={2507.07430},
      archivePrefix={arXiv},
      primaryClass={cond-mat.str-el},
      url={https://arxiv.org/abs/2507.07430}, 
}

@misc{kobayashi2025projectiverepresentationsbogomolovmultiplier,
      title={Projective Representations, Bogomolov Multiplier, and Their Applications in Physics}, 
      author={Ryohei Kobayashi and Haruki Watanabe},
      year={2025},
      eprint={2507.12515},
      archivePrefix={arXiv},
      primaryClass={cond-mat.str-el},
      url={https://arxiv.org/abs/2507.12515}, 
}

@article{Yang_2017,
   title={Irreducible projective representations and their physical applications},
   volume={51},
   ISSN={1751-8121},
   url={http://dx.doi.org/10.1088/1751-8121/aa971a},
   DOI={10.1088/1751-8121/aa971a},
   number={2},
   journal={Journal of Physics A: Mathematical and Theoretical},
   publisher={IOP Publishing},
   author={Yang, Jian and Liu, Zheng-Xin},
   year={2017},
   month=dec, pages={025207} }

@misc{ogata2025mixedstatetopologicalorder,
      title={Mixed state topological order: operator algebraic approach}, 
      author={Yoshiko Ogata},
      year={2025},
      eprint={2501.02398},
      archivePrefix={arXiv},
      primaryClass={math-ph},
      url={https://arxiv.org/abs/2501.02398}, 
}

@misc{ogata2026noteinvariantsmixedstatetopological,
      title={A note on invariants of mixed-state topological order in 2D}, 
      author={Yoshiko Ogata},
      year={2026},
      eprint={2601.09909},
      archivePrefix={arXiv},
      primaryClass={math-ph},
      url={https://arxiv.org/abs/2601.09909}, 
}

@article{Feng_2025,
   title={Hardness of Observing Strong-to-Weak Symmetry Breaking},
   volume={135},
   ISSN={1079-7114},
   url={http://dx.doi.org/10.1103/1xzd-g9s5},
   DOI={10.1103/1xzd-g9s5},
   number={20},
   journal={Physical Review Letters},
   publisher={American Physical Society (APS)},
   author={Feng, Xiaozhou and Cheng, Zihan and Ippoliti, Matteo},
   year={2025},
   month=nov }

@article{Gu:2024wgc,
    author = "Gu, Ding and Wang, Zijian and Wang, Zhong",
    title = "{Spontaneous symmetry breaking in open quantum systems: Strong, weak, and strong-to-weak}",
    eprint = "2406.19381",
    archivePrefix = "arXiv",
    primaryClass = "quant-ph",
    doi = "10.1103/3g6d-gn7b",
    journal = "Phys. Rev. B",
    volume = "112",
    number = "24",
    pages = "245123",
    year = "2025"
}

@book{eilers2018c,
  title={C*-Algebras and Their Automorphism Groups},
  author={Eilers, S. and Olesen, D.},
  isbn={9780128141236},
  series={Pure and Applied Mathematics},
  url={https://books.google.ca/books?id=0U84DwAAQBAJ},
  year={2018},
  publisher={Academic Press}
}

@misc{evington2022anomaloussymmetriesclassifiablecalgebras,
      title={Anomalous symmetries of classifiable C*-algebras}, 
      author={Samuel Evington and Sergio Girón Pacheco},
      year={2022},
      eprint={2105.05587},
      archivePrefix={arXiv},
      primaryClass={math.OA},
      url={https://arxiv.org/abs/2105.05587}, 
}

@article{Ogata_2022_tensorcat,
   title={A derivation of braided C*-tensor categories from gapped ground states satisfying the approximate Haag duality},
   volume={63},
   ISSN={1089-7658},
   url={http://dx.doi.org/10.1063/5.0061785},
   DOI={10.1063/5.0061785},
   number={1},
   journal={Journal of Mathematical Physics},
   publisher={AIP Publishing},
   author={Ogata, Yoshiko},
   year={2022},
   month=jan }

@article{Shi_2020,
   title={Fusion rules from entanglement},
   volume={418},
   ISSN={0003-4916},
   url={http://dx.doi.org/10.1016/j.aop.2020.168164},
   DOI={10.1016/j.aop.2020.168164},
   journal={Annals of Physics},
   publisher={Elsevier BV},
   author={Shi, Bowen and Kato, Kohtaro and Kim, Isaac H.},
   year={2020},
   month=jul, pages={168164} }

@article{Jones_2021,
   title={Remarks on Anomalous Symmetries of C*-Algebras},
   volume={388},
   ISSN={1432-0916},
   url={http://dx.doi.org/10.1007/s00220-021-04234-4},
   DOI={10.1007/s00220-021-04234-4},
   number={1},
   journal={Communications in Mathematical Physics},
   publisher={Springer Science and Business Media LLC},
   author={Jones, Corey},
   year={2021},
   month=oct, pages={385–417} }

@article{Doplicher1967,
  author    = {Doplicher, S. and Kadison, R. V. and Kastler, D. and Robinson, Derek W.},
  title     = {Asymptotically abelian systems},
  journal   = {Communications in Mathematical Physics},
  year      = {1967},
  month     = jun,
  volume    = {6},
  number    = {2},
  pages      = {101--120},
  doi       = {10.1007/BF01654127},
  url       = {https://doi.org/10.1007/BF01654127},
  issn      = {1432-0916}
}

@misc{QChannelLecture.pdf,
	author = {Michael M. Wolf},
	title = {Quantum Channels and Operations - Guided Tour},
	year = {2012},
    note = {Graue Literatur}
}

@article{Sang_2025,
   title={Stability of Mixed-State Quantum Phases via Finite Markov Length},
   volume={134},
   ISSN={1079-7114},
   url={http://dx.doi.org/10.1103/PhysRevLett.134.070403},
   DOI={10.1103/physrevlett.134.070403},
   number={7},
   journal={Physical Review Letters},
   publisher={American Physical Society (APS)},
   author={Sang, Shengqi and Hsieh, Timothy H.},
   year={2025},
   month=feb }

@misc{zou2026liebschultzmattisanomaliesanomalymatching,
      title={Lieb-Schultz-Mattis Anomalies and Anomaly Matching}, 
      author={Liujun Zou and Meng Cheng},
      year={2026},
      eprint={2604.00347},
      archivePrefix={arXiv},
      primaryClass={cond-mat.str-el},
      url={https://arxiv.org/abs/2604.00347}, 
}

@article{Crismale_2020,
   title={$C^*$-fermi systems and detailed balance},
   volume={11},
   ISSN={1664-235X},
   url={http://dx.doi.org/10.1007/s13324-020-00412-0},
   DOI={10.1007/s13324-020-00412-0},
   number={1},
   journal={Analysis and Mathematical Physics},
   publisher={Springer Science and Business Media LLC},
   author={Crismale, Vitonofrio and Duvenhage, Rocco and Fidaleo, Francesco},
   year={2020},
   month=Nov }

@misc{matsui2020splitpropertyfermionicstring,
      title={Split Property and Fermionic String Order}, 
      author={Taku Matsui},
      year={2020},
      eprint={2003.13778},
      archivePrefix={arXiv},
      primaryClass={math.OA},
      url={https://arxiv.org/abs/2003.13778}, 
}

@book{WilliamsCrossedProducts,
  author    = {Williams, Dana P.},
  title     = {Crossed Products of {$C^*$}-Algebras},
  series    = {Mathematical Surveys and Monographs},
  volume    = {134},
  publisher = {American Mathematical Society},
  address   = {Providence, RI},
  year      = {2007},
  doi       = {10.1090/surv/134}
}

@misc{divi2026localstrongtoweakspontaneoussymmetry,
      title={Local Strong-to-Weak Spontaneous Symmetry Breaking}, 
      author={Francisco Divi and Leonardo A. Lessa and Chong Wang},
      year={2026},
      eprint={2605.28967},
      archivePrefix={arXiv},
      primaryClass={quant-ph},
      url={https://arxiv.org/abs/2605.28967}, 
}

@misc{zhang2026localdiagnosticsstrongtoweakspontaneous,
      title={Local diagnostics for strong-to-weak spontaneous symmetry breaking and non-equilibrium phase transitions}, 
      author={Carolyn Zhang},
      year={2026},
      eprint={2605.29113},
      archivePrefix={arXiv},
      primaryClass={quant-ph},
      url={https://arxiv.org/abs/2605.29113}, 
}
\printnomenclature
\end{document}